\documentclass[11pt,letterpaper]{article}
\usepackage[T1]{fontenc}
\usepackage{lineno}
\usepackage{amsfonts,amsmath,amsthm,amssymb,dsfont,mathtools}
\usepackage[margin=1in]{geometry}
\usepackage[colorlinks,citecolor=blue,linkcolor=blue,urlcolor=red]{hyperref}
\usepackage{url}
\usepackage[linesnumbered,boxed,ruled,vlined]{algorithm2e}
\usepackage{thm-restate}
\usepackage{enumitem}
\usepackage{etoolbox}
\usepackage{graphicx}
\usepackage{mathdots}
\usepackage[normalem]{ulem}
\usepackage{xspace}
\usepackage{authblk}
\xspaceaddexceptions{]\}}
\usepackage{comment}
\usepackage{tabu}
\usepackage{framed}
\usepackage{float,wrapfig}
\usepackage{subfigure}
\usepackage{tikz}
\usepackage[draft]{fixme}
\usepackage[capitalise]{cleveref}

\newcommand*\linenomathpatch[1]{%
  \cspreto{#1}{\linenomath}%
  \cspreto{#1*}{\linenomath}%
  \csappto{end#1}{\endlinenomath}%
  \csappto{end#1*}{\endlinenomath}%
}

\linenomathpatch{equation}
\linenomathpatch{gather}
\linenomathpatch{multline}
\linenomathpatch{align}
\linenomathpatch{alignat}
\linenomathpatch{flalign}

\usetikzlibrary{arrows,arrows.meta,decorations.pathmorphing}

\theoremstyle{plain}
\newtheorem{theorem}{Theorem}[section]
\newtheorem{lemma}[theorem]{Lemma}

\newtheorem{fact}[theorem]{Fact}
\newtheorem{prob}{Problem}

\theoremstyle{definition}
\newtheorem{definition}[theorem]{Definition}

\fxsetup{mode=multiuser,theme=color, layout=inline}
\FXRegisterAuthor{tk}{atk}{\color{brown} {\bf Tomasz}}

\setlist[enumerate]{nosep, topsep=1ex}
\setlist[itemize]{nosep, topsep=1ex}
\setlist[description]{nosep}
\allowdisplaybreaks
\def\ShowAuthNotes{1}
\ifnum\ShowAuthNotes=1
\newcommand{\authnote}[2]{\ \\ \textcolor{red}{\parbox{0.9\linewidth}{[{\footnotesize {\bf #1:} { {#2}}}]}}\newline}
\else
\newcommand{\authnote}[2]{}
\fi

\newcommand{\poly}{\operatorname{\mathrm{poly}}}
\newcommand{\polylog}{\operatorname{\mathrm{polylog}}}

\newcommand{\Oh}{\mathcal{O}}
\newcommand{\Ohtilde}{\widetilde{\Oh}}

\newcommand{\caA}{\mathcal{A}}
\newcommand{\caB}{\mathcal{B}}

\def\fragmentoc#1#2{(#1\dd #2]}

\def\position#1{[#1]}

\newcommand{\dd}{\mathinner{.\,.\allowbreak}}

\newcommand{\ed}{\mathsf{ed}}

\newcommand{\LZ}{\mathsf{LZ}}
\newcommand{\rev}[1]{\overline{#1}}

\DeclareMathOperator*{\argmin}{arg\,min}
\DeclareMathOperator{\LF}{LF}
\DeclareMathOperator{\BWT}{BWT}
\DeclareMathOperator{\SA}{SA}
\DeclareMathOperator{\LCP}{LCP}
\DeclareMathOperator{\LCE}{LCE}
\DeclareMathOperator{\ISA}{ISA}
\DeclareMathOperator{\DSA}{DSA}

\def\twoheadleadsto{\tikz[baseline=(a.base)]{\draw[%
    decorate,decoration={zigzag,segment length=4, amplitude=.9},%
    ] (0,0) -- (.25, 0);%
    \draw[%
    -{Classical TikZ Rightarrow}.{Classical TikZ Rightarrow},%
    ] (.25, 0) -- (.4, 0);%
    \node (a) at (.4/2,-.03) {\phantom{\(\leadsto\)}};%
}}

\newcommand{\onto}{\twoheadleadsto}

\title{Near-Optimal Quantum Algorithms for Bounded Edit Distance and Lempel--Ziv Factorization}

\author[1]{Daniel Gibney}
\author[2]{Ce Jin}
\author[3]{Tomasz Kociumaka}
\author[4]{Sharma V. Thankachan}

\affil[1]{University of Texas at Dallas, Richardson, TX, U.S.; \href{mailto:daniel.j.gibney@gmail.com}{daniel.j.gibney@gmail.com}}
\affil[2]{Massachusetts Institute of Technology, Cambridge, MA, U.S.; \href{mailto:cejin@mit.edu}{cejin@mit.edu}}
\affil[3]{Max Planck Institute for Informatics, SIC, Saarbr\"ucken, Germany; \href{mailto:tomasz.kociumaka@mpi-inf.mpg.de}{tomasz.kociumaka@mpi-inf.mpg.de}}
\affil[4]{North Carolina State University, Raleigh, NC, U.S.; \href{mailto:sharma.thankachan@ncsu.edu}{sharma.thankachan@ncsu.edu}}

\date{\vspace{-1cm}}

\begin{document}
	\maketitle
	\thispagestyle{empty}
\begin{abstract}
Measuring sequence similarity and compressing texts are among the most fundamental tasks in string algorithms. In this work, we develop near-optimal \emph{quantum} algorithms for the central problems in these two areas: computing the edit distance of two strings [Levenshtein, 1965] and building the Lempel--Ziv factorization of a string [Ziv \& Lempel, 1977], respectively.

Classically, the edit distance of two length-$n$ strings can be computed in $\Oh(n^2)$ time and there is little hope for a significantly faster algorithm: an $\Oh(n^{2-\epsilon})$-time procedure would falsify the Strong Exponential Time Hypothesis.
Quantum computers might circumvent this lower bound, but even $3$-approximation of edit distance is not known to admit an $\Oh(n^{2-\epsilon})$-time quantum algorithm.
In the \emph{bounded} setting, where the complexity is parameterized by the value $k$ of the edit distance, there is an $\Oh(n+k^2)$-time classical algorithm [Myers, 1986; Landau \& Vishkin, 1988], which is optimal (up to sub-polynomial factors and conditioned on SETH) as a function of $n$ and $k$. Our first main contribution is a quantum $\Ohtilde(\sqrt{nk}+k^2)$-time algorithm that uses $\Ohtilde(\sqrt{nk})$ queries, where the $\Ohtilde(\cdot)$ notation hides polylogarithmic factors.
This query complexity is unconditionally optimal, and any significant improvement in the time complexity would break the quadratic barrier for the unbounded~\mbox{setting}.
Interestingly, our divide-and-conquer quantum algorithm reduces the bounded edit distance problem to the special case where the two input strings have small Lempel--Ziv factorizations. Then, it combines our quantum LZ compression algorithm with a classical subroutine computing edit distance between compressed strings.
The LZ factorization problem can be classically solved in $\Oh(n)$ time, which is unconditionally optimal in the quantum setting (even for computing just the size $z$ of the factorization).
We can, however, hope for a quantum speedup if we parameterize the complexity in terms of~$z$.
Already a generic oracle identification algorithm [Kothari 2014] yields the optimal query complexity of \smash{$\Ohtilde(\sqrt{nz})$} at the price of exponential running time. Our second main contribution is a quantum algorithm that also achieves the optimal time complexity of \smash{$\Ohtilde(\sqrt{nz})$}.
The key insight is the introduction of a novel LZ-like factorization of size $\Oh(z \log^2 n)$, which allows us to efficiently compute each new factor through a combination of classical and quantum algorithmic techniques. From this, we obtain the desired LZ factorization.
Using existing results [Kempa \& Kociumaka, 2020], we can then obtain the string's run-length encoded Burrows-Wheeler Transform (BWT)---another classical compressor~[Burrows \& Wheeler, 1994], and a  structure for longest common extensions (LCE) queries in \smash{$\Ohtilde(z)$} extra time~[I, 2017; Nishimoto et al., 2016].

Lastly, we obtain efficient indexes of size $\Ohtilde(z)$ for counting and reporting the occurrences of a given pattern and for supporting more general suffix array and inverse suffix array queries, based on the recent \emph{r-index} [Gagie, Navarro, and Prezza, 2020].
These indexes can be constructed in $\Ohtilde(\sqrt{nz})$ quantum time, which allows us to solve many fundamental problems, like longest common substring, maximal unique matches, and Lyndon factorization, in time~$\Ohtilde(\sqrt{nz})$.

\end{abstract}
	\newpage

\section{Introduction}\label{sec:intro}

String processing constitutes one of the oldest fields within theoretical computer science. Fifty years after the discovery of some of its foundations, such as the suffix tree~\cite{Weiner1973} and the linear-time exact pattern matching algorithm~\cite{Morris1970}, it remains a lively research area. Developments have been motivated both by the practical challenges of handling the rapidly growing volume of sequential 
data, especially in bioinformatics and data compression, and by the theoretical interest in demanding open questions.

More recently, the rapid progress in quantum computing brought increased attention to the development of quantum algorithms for fundamental string processing problems. As a precursor of this line of work, a paper of Hariharan and Vinay~\cite{DBLP:journals/jda/HariharanV03} demonstrated that a clever application of Grover search~\cite{DBLP:conf/stoc/Grover96} yields an $\Ohtilde(\sqrt{n})$-time\footnote{The $\Ohtilde(\cdot)$ notations hides factors polylogarithmic in the input size $n$. In particular $\Ohtilde(1)=\Oh(\log^{\Oh(1)} n)$.} quantum algorithm for exact pattern matching within a length-$n$ text; this time complexity is unconditionally optimal.
In the last few years, a series of works~\cite{DBLP:conf/innovations/GallS22,WY20,DBLP:conf/soda/AkmalJ22,JN23} resulted in a nearly-optimal quantum algorithm for the Longest Common Factor problem (also known as Longest Common Substring, it was the original motivation for the suffix trees) as well as other classic problems such as Lexicographically Minimal String Rotation; see also~\cite{boroujeni2021approximating,DBLP:conf/focs/AaronsonGS19,DBLP:conf/mfcs/AmbainisBIKKPSS20,quantdc} for quantum algorithms for some other string problems.
In this work, we develop quantum algorithms for two fundamental problems in string processing---Edit Distance and Lempel–Ziv (LZ77~\cite{DBLP:journals/tit/ZivL77}) Factorization---which are the central computational tasks in string similarity and text compression, respectively.

The edit distance (also known as the Levenshtein distance~\cite{Lev65}) between two strings is defined as the minimum number of character insertions, deletions, and substitutions (collectively called \emph{edits}) needed to transform one string into the other. 
Along with the Hamming distance (which allows substitutions only), it constitutes one of the two main measures of sequence (dis)similarity.
The edit distance of two strings of length at most $n$ can be classically computed in $\Oh(n^2)$ time using a textbook dynamic-programming algorithm~\cite{Vin68,NW70,Sel74,WF74}. One of the celebrated results of fine-grained complexity~\cite{BI18} is that any truly-subquadratic-time algorithm (working in $\Oh(n^{2-\epsilon})$ time for some constant $\epsilon > 0$) would violate the Strong Exponential Time Hypothesis (SETH)~\cite{IP01}. 
A quantum counterpart of SETH only excludes $\Oh(n^{1.5-\epsilon})$-time quantum algorithms~\cite{BPS21}, but no quantum speed-up is known for edit distance, and bridging the $\Omega(n^{1.5-\epsilon})$ lower bound with the $\Oh(n^2)$ upper bound remains a major open question~\cite{Rub19}.

One of the earliest ways to circumvent the quadratic complexity of Edit Distance is to parameterize the running time in terms of the value $k$ of the edit distance. A series of works from 1980s~\cite{Ukk85,Mye86,LV88} resulted in an $\Oh(n+k^2)$-time classical algorithm for this \emph{Bounded} Edit Distance problem.
This running time is optimal, up to subpolynomial factors and conditioned on SETH, as a function of $n$ and $k$.
More precisely, if the $\Oh(k^2)$ term can be reduced to $\Oh(k^{2-\epsilon})$ for some $k=\Theta(n^\kappa)$, with $\frac12<\kappa\le 1$, then 
a straightforward reduction translates this to an $\Oh(n^{2-\epsilon})$-time algorithm for unbounded Edit Distance.
The Bounded Edit Distance problem has been extensively studied: there are efficient sketching \& streaming algorithms~\cite{BZ16,JNW21,KPS21,BK23}, 
algorithms for compressed strings~\cite{GKLS22}, approximation algorithms~\cite{GKS19,KS20,BCFN22a}, and algorithms for preprocessed data~\cite{GRS20,BCFN22b}, to name just a few settings.
Many of these results are based on the general scheme of~\cite{LV88}, whose central component is an efficient implementation of longest common extension (LCE) queries\footnote{The LCE of two positions is (the length of) the longest substring that occurs at both positions.}.
Thus, when Jin and Nogler~\cite{JN23} obtained the optimal quantum trade-off for LCE queries, they asked whether their result could be applied to Bounded Edit Distance.
The best one could reasonably hope for would be $\Ohtilde(\sqrt{kn})$ query complexity and $\Ohtilde(\sqrt{kn}+k^2)$ time complexity. This is because already computing the number of $\texttt{1}$s in a length-$n$ binary string (which is its edit distance with $\texttt{0}^n$) requires $\Omega(\sqrt{kn})$ queries~\cite{Pat92,BBCMW01} and, through the aforementioned reduction, any improvement upon the $\Ohtilde(k^2)$ term (whenever it dominates, that is, for $k=\Theta(n^\kappa)$ with $\frac13 < \kappa \le 1$) would improve upon the $\Oh(n^2)$-time quantum algorithm for unbounded Edit Distance.
The first main contribution of this work is a quantum algorithm with the desired query and time complexities:

\begin{theorem}[Bounded Edit Distance]\label{thm:edit-main}
    There is a quantum algorithm that, given quantum oracle access to strings $X,Y\in \Sigma^{\le n}$, computes their edit distance $k:=\ed(X,Y)$, along with a sequence of $k$ edits transforming $X$ into $Y$, in query complexity $\Ohtilde(\sqrt{n+nk})$ and time complexity $\Ohtilde(\sqrt{n+nk} + k^2)$.
\end{theorem} 

Surprisingly, our algorithm uses neither quantum LCE queries nor the underlying technique of quantum string synchronizing sets~\cite{JN23};
that approach seems to get stuck at the query complexity of $\tilde{\Theta}(k\sqrt{n})$.
Instead of an LCE-based dynamic-programming procedure, we design a novel divide-and-conquer algorithm that crucially uses \emph{compressibility}.
If the input strings are compressible, we can retrieve their compressed representations and then solve the problem classically (the folklore implementation combines the Landau--Vishkin algorithm~\cite{LV88} with LCE queries on compressed strings~\cite{I17}; see~\cite{GKLS22}).
Otherwise, we exploit the locality of edit distance: in order to optimally align a given character, it suffices to compute the edit distance locally, on the largest \emph{compressible} context of that character (reusing the procedure mentioned above).
 Once we fix the alignment of a single character, the instance naturally decomposes into two \emph{independent} sub-instances asking to compute the edit distance to the left and the right of the aligned characters.
In the aforementioned statements, the \emph{compressibility} is quantified in terms of an upper bound $\Ohtilde(k)$ on the size of the LZ77 factorization.
Concurrently to this work, a similar divide-and-conquer strategy has been applied in a classical algorithm for \emph{weighted} bounded edit distance~\cite{CKW23}, where the LCE-based DP is incorrect~\cite{DGHKS23}.
Prior to that, compression has been used for computing edit distance only in the sketching algorithms of~\cite{KPS21,BK23},
where it comes more naturally because the sketches need to squeeze the input strings into $\Ohtilde(k^2)$ bits each so that their edit distance can be recovered.

The efficiency of our edit-distance algorithm thus depends on the construction of the LZ77 factorization~\cite{DBLP:journals/tit/ZivL77}.
The factorization is defined through a left-to-right scan of the text such that each new factor is either the leftmost occurrence of a symbol or an occurrence of the longest substring that also occurs earlier in the text; see Section~\ref{sec:prelim} for a formal definition and an example.

Finding the LZ factorization of a text is a fundamental problem on its own. In particular, it forms the basis of many practical compression algorithms, such as \texttt{zip}, \texttt{p7zip}, \texttt{gzip}, and \texttt{arj}.
This problem admits a classic linear-time solution~\cite{DBLP:journals/jacm/RodehPE81}, and it has been considered in a variety of other settings, including the external memory setting~\cite{DBLP:conf/dcc/KarkkainenKP14}, the dynamic setting~\cite{DBLP:journals/dam/NishimotoIIBT20}, and the packed setting~\cite{DBLP:conf/soda/BelazzouguiP16,DBLP:conf/soda/MunroNN17}. 
Our second main contribution is a novel quantum algorithm for LZ factorization, whose time complexity is optimal up to logarithmic factors.

\begin{restatable}[LZ77 factorization]{theorem}{lz}\label{thm:lz}
There is a quantum algorithm that, given a quantum oracle access to an unknown string $X\in \Sigma^n$, computes the LZ factorization $\LZ(X)$ using $\Ohtilde(\sqrt{zn})$ query and time complexity, where $z := |\LZ(X)|$ is the size of the factorization.
\end{restatable}

The size $z$ of the LZ77 factorization is known to be within a polylogarithmic factor away from numerous other compressibility measures~\cite{DBLP:journals/csur/Navarro21a}. This includes the sizes of the smallest context-free grammar~\cite{DBLP:journals/tit/CharikarLLPPSS05}, the smallest bidirectional macro scheme~\cite{DBLP:journals/jacm/StorerS82}, and the smallest string attractor~\cite{DBLP:conf/stoc/KempaP18}, all of which are, however, NP-hard to compute exactly. Recently, the size $r$ of the run-length-encoded Burrows--Wheeler transform (BWT)~\cite{burrows1994block}
has also been shown to satisfy $r = \Oh(z \log^2 \frac{n}{z})$~\cite{DBLP:journals/cacm/KempaK22} and $r=\Omega(z / \log\frac{n}{z})$~\cite{GNP18}. 
This measure is classically computable in linear time, and the underlying compressed string representation constitutes the basis of the practical \texttt{bzip2} algorithm. 
The run-length-encoded BWT can be constructed in $\Ohtilde(r)$ time from the LZ77 factorization~\cite{DBLP:journals/cacm/KempaK22},
so \cref{thm:lz} immediately yields the following important corollary:
\begin{restatable}[Run-length-encoded BWT]{corollary}{rlbwt}\label{cor:alg_rlbwt}
There is a quantum algorithm that, given a quantum oracle to an unknown string $X\in \Sigma^n$, computes the run-length-encoded Burrows--Wheeler transform of $X$ using $\Ohtilde(\sqrt{rn})$ query and time complexity, where $r$ is the number of runs~in~the~BWT.
\end{restatable}

For most applications, compressing the data is only half of the battle. We also need to be able to perform computation over this data quickly, which is the motivation behind compressed text indexing. 
Traditional indexes such as suffix arrays and suffix trees require linear-time preprocessing and, once constructed, occupy linear space.
A major achievement in compressed text indexing within the last two decades was the development of space-efficient 
representations of suffix trees/arrays in space close to ``optimal'' in terms of (higher-order) statistical entropy~\cite{FM05,DBLP:journals/siamcomp/GrossiV05,DBLP:journals/csur/NavarroM07,DBLP:journals/mst/Sadakane07}.
A recent breakthrough by Gagie, Navarro, and Prezza~\cite{DBLP:conf/soda/GagieNP18}, known as the \emph{r-index}, takes $\Oh(r)$ space and can answer pattern-matching queries (both counting and reporting the occurrences) in near-optimal time; also see the improvements in~\cite{NishimotoT21,NishimotoKT22}.
Its $\Oh(r\log\frac{n}{r})$-space version can support more general operations such as suffix array and inverse suffix array queries in $\Oh(\log\frac{n}{r})$ time~\cite{DBLP:journals/jacm/GagieNP20}. 
We show how to construct these indexes fast, as specified in the following result.

\begin{restatable}[Compressed Index]{theorem}{index}\label{thm:index}
There is an $\Ohtilde(\sqrt{rn})$-time quantum algorithm that, given a quantum oracle to an unknown string $X \in \Sigma^n$, constructs 
\begin{itemize}
    \item an $\Oh(r)$-space index that can count the occurrences of any length-$m$ pattern in $\Ohtilde(m)$ time and report these occurrences in time $\Ohtilde(m+\mathsf{occ})$, where $\mathsf{occ}$ is the number of occurrences;
    \item an $\Ohtilde(r)$-space index for suffix array and inverse suffix array queries in $\Ohtilde(1)$ time.
\end{itemize}
\end{restatable}

To handle LCE queries, we can use the $\Oh(z\log\frac{n}{z})$ space data structure with query time $\Oh(\log n)$ by I~\cite{I17}, which can be constructed from LZ factorization in $\Ohtilde(z)$ time. 
This structure, combined with \cref{thm:index}, enables us to solve numerous other classical string problems. A few examples provided in this work include:

\begin{itemize}

\item Finding the longest common substring between two strings of total length $n$ in $\Ohtilde(\sqrt{zn})$ time, where $z$ is the number of factors in the LZ77 parse of their concatenation. 
For highly compressible strings, this beats the best known $\Ohtilde(n^{2/3})$ time quantum algorithm~\cite{DBLP:conf/soda/AkmalJ22,DBLP:conf/innovations/GallS22}.
Similar time bounds can be obtained for finding the set of maximal unique matches (MUMs); the longest repeating substring/shortest unique substring of a given string.  
\item Obtaining the Lyndon factorization of a string in $\Ohtilde(\sqrt{\ell n})$ time, where $\ell = \Ohtilde(z)$ is the number of its Lyndon factors. 
\item Determining the frequencies of all distinct substrings of length $q$ ($q$-grams) in time $\Ohtilde(\sqrt{zn} + d_q)$,  where $d_q=\Oh(zq)$ is the number of distinct $q$-grams.
\end{itemize}
\section{Preliminaries}\label{sec:prelim}
A string is a finite sequence of characters from the alphabet $\Sigma$,
which we assume to be of the form $[0\dd \sigma)$ for an integer $\sigma = n^{\Oh(1)}$, where $n$ is the input size.
We denote the length of a string $X$ as~$|X|$. For any $i\in [1\dd |X|]$,
the $i$th character of $X$ is $X[i]$. 
For $0\le i \le j \le |X|$, a string of the form $X[i+1]\cdots X[j-1]X[j]$ is a \emph{substring} of $X$.
Its \emph{occurrence} in $X$ ending at position $j$ is called a \emph{fragment} of $X$ and denoted with $X(i\dd j]$; 
this fragment can also be referred to as $X[i+1\dd j]$, $X[i+1\dd j+1)$, or $X(i\dd j+1)$. 
Prefixes and suffixes are fragments of the form $X[1\dd i]$ and $X(i\dd |X|]$, respectively.
The string $X[|X|]\cdots X[2]X[1]$, 
called 
the \emph{reverse} of $X$, is denoted by~$\rev{X}$.

\paragraph*{Quantum algorithms}
We assume the input string $X\in \Sigma^{n} $ is accessed via a quantum oracle $O_X\colon \lvert i,b \rangle \mapsto \lvert i,b\oplus X[i] \rangle$, for any index $i\in [n]$ and any $b\in [0\dd 2^{\lceil \log \sigma\rceil})$, where $\oplus$ denotes the XOR operation. 
This quantum query model \cite{ambainis2004quantum,DBLP:journals/tcs/BuhrmanW02} is standard in the literature of quantum algorithms. 
The \emph{query complexity} of a quantum algorithm (with success probability at least $2/3$) is the number of quantum queries it makes to the input oracles.

More specifically, it suffices for us to have a computational model that supports the following:
\begin{itemize}
    \item We have quantum query access to the input oracle (as described above).
    \item We can run quantum subroutines on $\Oh(\log n)$ qubits.
    \item We have a classical working memory with random access (classical-read and classical-write).
\end{itemize}
The \emph{time complexity} of our algorithm counts the number of quantum queries, the number of elementary gates that implement the quantum subroutines, and the number of classical random-access operations.
Note that we do not need to assume QRAM for working memory, which was required in previous quantum algorithms for some other string problems \cite{DBLP:conf/innovations/GallS22,DBLP:conf/soda/AkmalJ22,JN23} in order to obtain good time complexity.

The key quantum subroutine that we use is the Grover search algorithm.
\begin{theorem}[Grover search \cite{DBLP:conf/stoc/Grover96}]
    \label{thm:grover}
   There is a quantum algorithm that, given quantum access to a function $f\colon [1\dd n]\to \{0,1\}$,  finds an index $i\in [1\dd n]$ such that $f(i)=1$ or reports that no such $i$ exists. The algorithm has $2/3$ success probability, $\Oh(\sqrt{n})$ query complexity, $\Ohtilde(\sqrt{n})$ time complexity, and uses only $\Oh(\log n)$ qubits.
\end{theorem}

A bounded-error algorithm can be boosted to have success probability $1-1/n^c$, for arbitrarily large constant $c$, by $\Oh(\log n)$ repetitions. 
In this paper, we do not optimize the $\poly \log(n)$ factors in the quantum query complexity (and time complexity) of our algorithms.

We can use 
\cref{thm:grover} to test the equality of two length-$\ell$ substrings of the input string(s) in $\Ohtilde(\sqrt{\ell})$ time.\footnote{For substrings $X'$ and $Y'$, define the function $f(i) = 1$ if $X'[i] \neq Y'[i]$ and $f(i) = 0$ otherwise.}
Combined with a binary search, this allows us to find the length of their longest common prefix (resp.,  suffix) in $\Ohtilde(\sqrt{\ell})$ time. We can then determine their
leftmost (resp, rightmost) position corresponding to a mismatch (if it exists) and hence their lexicographic (resp., co-lexicographic\footnote{Co-lexicographic order refers to the lexicographic order of the reversed strings.}) order in constant time.

\paragraph*{Edit Distance and Alignments}
The \emph{edit distance} (also known as \emph{Levenshtein distance}~\cite{Lev65}) between two
strings $X$ and $Y$, denoted by $\ed(X,Y)$, is the minimum number of
character insertions, deletions, and substitutions required to transform $X$ into~$Y$.
For a formal definition, we first rely on the notion of an \emph{alignment} between fragments of strings.
\begin{definition}[see~\cite{KPS21}]
    A sequence $\caA=(x_i,y_i)_{i=0}^{m}$ is an \emph{alignment}
    of $X\fragmentoc{x}{x'}$ onto $Y\fragmentoc{y}{y'}$, denoted by \(\caA: X\fragmentoc{x}{x'} \onto Y\fragmentoc{y}{y'}\),
    if it satisfies $(x_0,y_0)=(x,y)$, $(x_{i},y_{i})\in \{(x_{i-1}+1,\allowbreak y_{i-1}+1),(x_{i-1}+1,\allowbreak y_{i-1}),(x_{i-1},y_{i-1}+1)\}$ for $i\in [1\dd m]$, and $(x_m,y_m) =(x',y')$.
    \begin{itemize}
        \item If $(x_{i},y_{i})=(x_{i-1}+1,y_{i-1})$, we say that $\caA$ \emph{deletes}
            $X\position{x_i}$,
        \item If $(x_{i},y_{i})=(x_{i-1},y_{i-1}+1)$, we say that $\caA$ \emph{inserts} $Y\position{y_i}$,
        \item If $(x_{i},y_{i})=(x_{i-1}+1,y_{i-1}+1)$, we say that $\caA$ \emph{aligns}
            $X\position{x_i}$ and $Y\position{y_i}$. If~additionally $X\position{x_i}=
            Y\position{y_i}$, we say that $\caA$ \emph{matches} $X\position{x_i}$ and
            $Y\position{y_i}$; otherwise, $\caA$ \emph{substitutes} $Y\position{y_i}$ for
            $X\position{x_i}$.
    \end{itemize}
\end{definition}

The \emph{cost} of an alignment $\caA$ of $X\fragmentoc{x}{x'}$ onto $Y\fragmentoc{y}{y'}$,
denoted by $\ed_{\caA}(X\fragmentoc{x}{x'},Y\fragmentoc{y}{y'})$, is the total number of characters that $\caA$ inserts, deletes, or substitutes.
Now, we define the edit distance $\ed(X,Y)$ as the minimum cost of an alignment
of $X\fragmentoc{0}{|X|}$ onto~$Y\fragmentoc{0}{|Y|}$.
An alignment of $X$ onto $Y$ is \emph{optimal} if its cost is equal to $\ed(X, Y)$.

An alignment $\caA'':X\fragmentoc{x}{x'}\onto Z\fragmentoc{z}{z'}$ is a \emph{product} 
of alignments $\caA:X\fragmentoc{x}{x'}\onto Y\fragmentoc{y}{y'}$ and $\caA':Y\fragmentoc{y}{y'}\onto Z\fragmentoc{z}{z'}$ if, for every $(\bar{x},\bar{z})\in \caA''$, there is $\bar{y}\in [y\dd y']$ such that $(\bar{x},\bar{y})\in \caA$ and $(\bar{y},\bar{z})\in \caA'$. Note that such an alignment always exists and satisfies $\ed_{\caA''}(X\fragmentoc{x}{x'},Z\fragmentoc{z}{z'})\le \ed_{\caA}(X\fragmentoc{x}{x'},Y\fragmentoc{y}{y'})+\ed_{\caA'}(Y\fragmentoc{y}{y'},\allowbreak Z\fragmentoc{z}{z'})$. 
For an alignment $\caA:X\fragmentoc{x}{x'}\onto Y\fragmentoc{y}{y'}$
with \(\caA = (x_i,y_i)_{i=0}^m\),
we define the \emph{inverse alignment} $\caA^{-1} : Y\fragmentoc{y}{y'}\onto
X\fragmentoc{x}{x'}$ as $\caA^{-1} \coloneqq (y_i,x_i)_{i=0}^m$.
Note that $\ed_{\caA^{-1}}(Y\fragmentoc{y}{y'},X\fragmentoc{x}{x'})=\ed_{\caA}(X\fragmentoc{x}{x'},Y\fragmentoc{y}{y'})$.

\paragraph*{Lempel--Ziv Factorization}
We say that a fragment $X[i\dd i+\ell)$ is a \emph{previous factor}
if it has an earlier occurrence in $X$, i.e., $X[i\dd i+\ell)=X[i'\dd i'+\ell)$
holds for some $i'\in [1\dd i)$.  An \emph{LZ77-like factorization} of
$X$ is a factorization $X = F_1\cdots F_f$ into non-empty
\emph{phrases} such that each phrase $F_j$ with $|F_j|>1$ is a
previous factor.  In the underlying \emph{LZ77-like representation},
every phrase $F_j=X[i\dd i+\ell)$ that is a previous factor is
encoded as $(i',\ell)$, where $i'\in [1\dd i)$ satisfies
$X[i\dd i+\ell)=X[i'\dd i'+\ell)$ (and is chosen arbitrarily in case of multiple
possibilities); if $F_j=X[i]$ is not a previous factor, it is encoded
as $(X[i],0)$; see \cref{fig:lz77_example} for an example.

\begin{figure}[h!]
    \resizebox{\textwidth}{!}{%
        \centering
        \begin{tabular}{m{12mm} m{7mm}| m{7mm} | m{7mm} | m{7mm} | m{7mm} m{7mm} | m{7mm} m{7mm} m{7mm} m{7mm} m{7mm} | m{7mm} m{7mm} m{7mm} |m{7mm}}
    Index &1 & 2 & 3 & 4 & 5 & 6 & 7 & 8 & 9 & 10 & 11 & 12 & 13 & 14 & 15\\
    $X$ & $a$  & $b$ & $a$  & $c$ & $a$ & $b$ & $c$ & $a$ & $b$ & $c$ & $a$ & $a$ & $a$ & $a$ & $b$ \\
    \end{tabular}
    }
    \caption{The LZ77 factorization of a string $X = abacabcabcaaaab$ of length $n = 15$. The resulting encoding has $z=8$ elements: $(a, 0)$, $(b, 0)$, $(1,1)$, $(c, 0)$, $(1, 2)$, $(4, 5)$, $(11, 3)$, $(9,1)$.}
    \label{fig:lz77_example}
\end{figure}

The LZ77 factorization~\cite{DBLP:journals/tit/ZivL77} (or the LZ77 parsing) of a string
$X$, denoted $\LZ(X)$ is then just an LZ77-like factorization constructed by greedily
parsing $X$ from left to right into longest possible phrases.  More
precisely, the $j$th phrase $F_j=X[i\dd i+\ell)$ is the longest previous factor
starting at position $i$; if no previous factor starts there, then $F_j$ consists of a single character.  
This greedy approach is known to produce the shortest possible LZ77-like factorization.

The size $|\LZ(X)|$ of the LZ77 factorization of $X$ is closely related to other compressibility measures; see~\cite{DBLP:journals/csur/Navarro21a} for a survey. This includes \emph{substring complexity} $\delta(X)$, which is defined as $\max_{q=1}^{|X|} \frac{d_q(X)}{q}$, where $d_q(X)$ is the number of distinct length-$q$ substrings ($q$-grams) in $X$. 
It has been implicitly introduced in~\cite{DBLP:journals/algorithmica/RaskhodnikovaRRS13} and thoroughly studied in~\cite{KNP23}.
The substring complexity enjoys many desirable features such as invariance under string reversal, monotonicity with respect to taking substring (in comparison, $|\LZ(X)|$ is only monotone with respect to taking prefixes), sub-additivity with respect to concatenations (shared with $|\LZ(X)|$), and stability with respect to edits (that is, $|\delta(X)-\delta(X')|\le \ed(X,X')$). 
Due to the relation $\delta(X) \le |\LZ(X)| = \Oh\big(\delta(X) \log \frac{|X|}{\delta(X)}\big)$ proved in~\cite{DBLP:journals/algorithmica/RaskhodnikovaRRS13}, we derive the following fact about the LZ77 factorization size:
\begin{fact}\label{fct:lz}
    Strings $X,Y$ of length at most $n$ satisfy  $|\LZ(\rev{X})|= \Oh(|\LZ(X)|\log n)$,
    $|\LZ(XY)|=\Oh((|\LZ(X)|+|\LZ(Y)|)\log n)$, and $|\LZ(Y)|=\Oh((|\LZ(X)|+\ed(X,Y))\log n)$.
\end{fact}


LZ-End was introduced by Kreft and Narvarro~\cite{DBLP:journals/tcs/KreftN13} to speed up the extraction of substrings relative to traditional LZ77. Unlike LZ77, LZ-End forces any new phrase that is not a leftmost occurrence of a symbol to match an occurrence ending at a previous phrase boundary, i.e., phrase $X[i\dd i+\ell)$ is taken as the longest fragment that is a suffix of $X[1\dd j)$, where $j$ is the start of a previous phrase. 
Like LZ77, LZ-End can be computed in linear time~\cite{DBLP:conf/dcc/KempaK17,DBLP:conf/esa/KempaK17}. Moreover, the LZ-End encoding size is close to the size of LZ77 encoding, as shown in the following result: 

\begin{theorem}[Kempa \& Saha~\cite{DBLP:conf/soda/KempaS22}]
\label{lem:z_e_and_z}
For any string $X[1\dd n]$ with LZ77 factorization size $z$ and LZ-End factorization size $z_e$, we have $z_e = \Oh(z \log^2 \frac{n}{z})$.
\end{theorem}

\paragraph*{Suffix Trees, Suffix Arrays, and the Burrow Wheeler Transform}

We assume that the last symbol in $X[1\dd n]$ is a special $\$$ symbol that occurs only once and is lexicographically smaller than the other symbols in $X$. The suffix tree of a string $X$ is a compact trie constructed from all suffixes of $X$. The tree leaves  are labeled with the starting indices of the corresponding suffixes and are sorted by the lexicographic order of the suffixes. These values in this order define the suffix array $\SA$, i.e., $\SA[i]$ is such $X[\SA[i]\dd n]$ is the $i$th smallest suffix lexicographically. 
The inverse suffix array $\ISA$, is defined as $\ISA[\SA[i]] = i$; equivalently, $\ISA[i]$ is the lexicographic rank of the suffix $X[i\dd n]$. The Burrows--Wheeler Transform (BWT) of a text $X$ is a permutation of the symbols of $X$ such that 
$\BWT[i] = X[\SA[i]-1]$ if $\SA[i]\neq 1$ and is $\$$ otherwise.
The longest common extension of two suffixes $X[i\dd n]$, $X[j\dd n]$, denoted as $\LCE(i,j)$, is equal to the length of their longest common prefix. 
The suffix tree, suffix array, and the Burrows--Wheeler Transform can all be built in linear time for polynomially-sized integer alphabets~\cite{DBLP:conf/focs/Farach97}. While suffix trees and arrays require space $\Oh(n)$ space (equivalently, $\Oh(n\log n)$ bits), the BWT requires only $n \lceil\log \sigma\rceil$ bits. Further, we can apply run-length encoding to achieve $\Oh(r)$ space.

\paragraph*{FM-index and Repetition-Aware Suffix Trees}
The FM-index provides the ability to count and locate occurrences of a given pattern efficiently. 
It is constructed based on the BWT described previously and uses the \emph{LF-mapping} to perform pattern matching. The LF-mapping is defined as $\LF[i] = \ISA[\SA[i]-1]$ if $\SA[i] \neq 1$ and is $1$ otherwise. 
The FM-index was developed by Ferragina and Manzini~\cite{FM05} to be more space efficient than traditional suffix trees and suffix arrays. However, supporting location queries utilized sampling $\SA$ in evenly spaced intervals, in a way independent of the runs in the BWT of the text, preventing an $\Oh(r)$ data space structure with optimal (or near optimal) query time.

The r-index and subsequent fully functional text indexes were developed to utilize only $\Oh(r)$ or $\Oh(r\log \frac{n}{r})$ space. The r-index developed by Gagie, Navarro, and Prezza~\cite{DBLP:conf/soda/GagieNP18} was designed to occupy $\Oh(r)$ space and support counting and locating queries in near-optimal time. It was based on the observation that suffix array samples are necessary only for the run boundaries of the BWT and subsequent non-boundary suffix array values can be obtained in polylogarithmic time. The fully functional indexes~\cite{DBLP:journals/jacm/GagieNP20} use $\Oh(r\log \frac{n}{r})$ space and provide most of the capabilities of a suffix tree. The data structure allows one to determine in $\Oh(\log \frac{n}{r})$ time arbitrary $\SA$, $\ISA$, and $\LCE$ values, which in turn lets one determine properties of arbitrary nodes in suffix tree, such as subtree size.

\section{Technical Overview}\label{sec:technical_overview}

\subsection{Edit Distance}
\label{subsec:techoverview-edit}
Our algorithm for computing the edit distance between $X,Y\in \Sigma^{\le n}$ uses a divide-and-conquer approach: letting $x := \lceil |X|/2 \rceil$, we would like to \emph{optimally align} the middle character $X[x]$ to some character $Y[y]$.
More formally, we would like to find $y\in [0\dd |Y|]$ such that $ \ed(X,Y) = \ed(X[1\dd x),Y[1\dd y) ) + \ed(X[x\dd |X|],Y[y\dd |Y|] )$.
Such a pair $(x,y)$, called here an \emph{edit anchor}, allows us to decompose the problem of computing $\ed(X,Y)$ into two independent subproblems to be solved recursively. Therefore,  a crucial component of our divide-and-conquer algorithm is to efficiently find an edit anchor $(x,y)$ for $X,Y$, given the promise that $\ed(X,Y)\le k$. 

\paragraph*{Edit anchor and LZ compression}
Our plan for computing the edit anchor $(x,y)$ is to pick it from a \emph{locally optimal} alignment $\caA'$, which is hopefully faster to compute than a globally optimal alignment $\caA$. More specifically, we select a suitable window $(i\dd j]$ that contains~$x$ and define $\caA'$ as an optimal alignment  between the fragments $X':=X(i\dd j]$ and $Y':=Y(i\dd j+|Y|-|X|]$ (which satisfy $\ed(X',Y')\le \ed(X,Y)$).
To ensure correctness, this window $(i\dd j]$ should be long enough to eliminate ambiguity: any $(x,y) \in \caA'$ must be guaranteed to be a (global) edit anchor provided that $\ed(X,Y)\le k$.
On the other hand, for efficiency's sake, this window should not be too long. 

Interestingly, our criteria for selecting the window $(i\dd j]$ crucially rely on \emph{string compressibility}. 
Specifically, we define its right boundary as the largest $j \in [x\dd |X|]$ such that $|\LZ(X(x\dd j])| \le c\cdot k$ (for a sufficiently large constant $c$), and we define the left boundary $i$ symmetrically. 
Intuitively, $X'=X(i\dd j]$ is the maximal compressible context of $X[x]$.
Using our LZ-compression algorithm (\cref{thm:lz}), the window boundaries $i,j$ can be found in $\Ohtilde (\sqrt{kn})$ quantum time by binary search.
After we retrieve the LZ compression of the fragments $X'$ and $Y'$ with $\Ohtilde(k)$ phrases, we compute an anchor $(x,y)$ contained in an optimal alignment $\caA':X'\onto Y'$. Using the classical Landau--Vishkin algorithm~\cite{LV88} and appropriate LCE query implementation~\cite{I17} for compressed strings, this requires $\Ohtilde(k^2)$ additional time complexity:

\begin{theorem}[see also~\cite{GKLS22}]\label{thm:lzed}
  There exists an algorithm that, for any two strings $X,Y\in \Sigma^*$, given $\LZ(X)$ and $\LZ(Y)$, computes $k := \ed(X,Y)$, along with a sequence of $k$ edits transforming $X$ into $Y$, in $\Ohtilde(|\LZ(X)|+|\LZ(Y)|+k^2)$ time.
\end{theorem}

It remains to explain why the anchor $(x,y)$ derived from the optimal local alignment $\caA':X'\onto Y'$ is globally optimal for $X,Y$.
 This would follow from the following  key claim: any optimal global alignment $\caA \colon X\onto Y$ must intersect the alignment $\caA' \colon X'\onto Y'$ at two points $(x_\ell,y_\ell),(x_r,y_r)$ such that $(x_\ell,y_\ell)\preceq (x,y) \preceq (x_r,y_r)$.
Indeed, this claim implies that we can replace the part of $\caA$ between these two points with the corresponding part of $\caA'$, which contains the anchor $(x,y)$, without increasing the cost of the alignment. So $(x,y)$ is an edit anchor for $X, Y$.

To see why the claim above is true, here we focus on the intersection $(x_r,y_r)$ to the right of $(x,y)$, and without loss of generality assume the right boundary $j$ of the window achieves the equality $|\LZ(X(x\dd j])| = c\cdot k$. 
If $\caA'$ does not intersect $\caA$ at any point to the right of of $(x,y)$, then we can restrict both of them to the fragment $Y'_r:= Y(y\dd j+|Y|-|X|]$ and obtain two \emph{disjoint} $\Oh(k)$-cost alignments: $\caA'_r \colon X(x\dd j]\onto Y'_r$ and $\caA_r\colon X(\tilde x\dd \tilde{\jmath}] \onto Y'_r$ for some $\tilde x = x\pm \Oh(k)$ and $\tilde{\jmath} =j \pm \Oh (k)$. 
Without loss of generality, assume $\tilde x< x$. Then, the product $\caA_r^{-1}\circ \caA'_r\colon X(x\dd j] \onto X(\tilde x\dd \tilde{\jmath}]$ is 
 an $\Oh(k)$-cost alignment that, due to disjointness and $\tilde{x}<x$, matches each unedited character $X[t]$ to an earlier character \smash{$X[\tilde t]$} with $\tilde t<t$. 
This gives an LZ-like factorization of $X(x\dd j]$ into $\Oh(k)$ phrases, which contradicts the assumption that $|\LZ(X(x\dd j])| = c\cdot k$ (the constant $c$ is large enough).

We remark that a similar strategy, albeit with a different compressibility measure called \emph{self-edit distance}, has been concurrently applied to efficiently solve the bounded \emph{weighted} edit distance problem~\cite{CKW23}.
A related compression argument appeared in~\cite[Lemma III.10]{KPS21} in a different context of sketching edit distance. 
However, it was only applied to masked strings (with matched characters replaced by $\#$s) and achieved a weaker $\Oh(k^2)$ bound that does not suffice here. 

\paragraph*{Ideal analysis of divide and conquer}
So far, we have described a quantum algorithm which, given strings $X,Y$ with promise $\ed(X,Y) \le k$, finds an edit anchor $(x,y)$ in  
$T_0(|X|,|Y|,k):=\Ohtilde (\sqrt{k(|X|+|Y|)})$ query complexity and $\Ohtilde (\sqrt{k(|X|+|Y|)} + k^2)$ time complexity.
Let us try to analyze the query complexity of our divide-and-conquer approach based on this anchor-finding subroutine (the time complexity has a similar analysis, which we omit from this overview).

Suppose that the anchor $(x,y)$  (with $x= \lceil |X|/2\rceil $) decomposes the input strings into $X=X_1X_2$ and $Y=Y_1Y_2$, resulting in two subproblems $\ed(X_1,Y_1)$ and $\ed(X_2,Y_2)$.
Suppose that, before recursively solving these subproblems, we can somehow obtain upper bounds $k_1\ge \ed(X_1,Y_1)$ and $k_2\ge \ed(X_2,Y_2)$ such that $k_1+k_2=k$.  In this ideal scenario, the overall query complexity of is 
\[
    T(|X|,|Y|,k) = T(|X_1|, |Y_1|, k_1) + T(|X_2|, |Y_2|, k_2) + T_0(|X|,|Y|,k) 
    \le \Oh( T_0(|X|,|Y|,k)\cdot \log |X| ), \]
which can be shown by applying the Cauchy--Schwarz inequality to all the subproblems $(X_i,Y_i,k_i)$ at each of the $\lceil \log |X| \rceil $ levels of recursion (for vectors $(\sqrt{|X_i|+|Y_i|})_i$ and $(\sqrt{k_i})_i$).

This query complexity meets our target of $\tilde \Oh(\sqrt{kn})$ but relies on the unrealistic assumption about knowing the upper bounds of $\ed(X_1,Y_1)$ and $\ed(X_2,Y_2)$. To remove this assumption, we face the following situation: when solving each subproblem $(X,Y)$ in the recursion tree, we a priori do not know an upper bound on $\ed(X,Y)$, but we still want the query complexity spent on this subproblem to be bounded in terms of the true edit distance $k=\ed(X,Y)$.

\paragraph*{Reducing the overhead of exponential search}
A first attempt to resolve this issue is to estimate the distance $\ed(X,Y)$ with exponential search: We iteratively try a sequence of gradually increasing thresholds $k_1,k_2,\dots$ (where the usual choice is $k_i=2^i$) and, in the $i$-th iteration, pretend $k_i \ge \ed(X,Y)$. We apply the aforementioned anchor-finding subroutine in $T_0(|X|,|Y|,k_i)$ quantum query complexity, and then recurse on the subproblems defined by this anchor. 
In the first few iterations, where $k_i< \ed(X,Y)$, the found anchor might be incorrect, of course, 
causing the whole recursive call to eventually fail. 
But hopefully the total cost can be still bounded in terms of the cost of the successful iteration (i.e., the first one where $k_i \ge \ed(X,Y)$ holds).

Unfortunately, this standard exponential search idea no longer works in our recursive scenario: at each level, the wasted work incurs at least a constant-factor overhead, which accumulates multiplicatively across the $\lceil \log |X| \rceil$ levels of recursion, resulting in at least a polynomial-factor overall overhead.
One possible solution is to decrease the recursion depth by enlarging the branching factor, and hence decrease the overall overhead to a subpolynomial factor.
Instead of this generic idea, we use a more problem-specific approach to carefully implement the exponential search, so that a lot of redundant work can be avoided, and the total overhead is decreased to only a polylogarithmic~factor!

  As before, we iteratively try a sequence of increasing thresholds, where each iteration leads to a recursion based on the anchor found using the corresponding threshold $k_i$.
    Our goal is to reduce the cost of the wasted computations so that it becomes almost negligible compared to the correct recursive calls (i.e., based on the true edit-distance anchor).

The key insight here is that, in order to tell whether an anchor $(x_i,y_i)$, computed under a promise $\ed(X,Y)\le k_i$, is a correct anchor, we do not actually need to wait until its entire recursion finishes. Instead, we can pause this recursion after a certain amount of time and proceed to the larger threshold $k_{i+1}$.
A procedure similar to the aforementioned anchor-finding subroutine can tell us whether the earlier anchor $(x_i,y_i)$ is still correct under a weakened promise $\ed(X,Y)\le k_{i+1}$.
If it is, then we can continue running the earlier paused recursive call using $(x_i,y_i)$ (instead of starting from scratch using a new anchor); otherwise, we can abort that call, because we already know it is useless, and start a new recursive call using a new anchor $(x_{i+1},y_{i+1})$ instead.

This strategy allows us to control the total complexity contributed by recursive calls generated by incorrect anchors. We want their contribution to be at most a $1/\polylog n$ fraction of the complexity of the correct calls and, for this reason, we adjust the threshold sequence of the exponential search to $k_i=(\log n)^{2i}$ instead of $k_i=2^i$.
More details on implementing this strategy are given in \cref{subsec:edit-recur}.

Let us remark that the related divide-and-conquer procedure for bounded weighted edit distance~\cite{CKW23} faces an analogous issue.
However, since its target complexity is $\Ohtilde(k\sqrt{kn})$, recursive calls can afford measuring compressibility with respect to the global budget. As a result, all anchors are verified under the global promise of $\ed(X,Y)\le k$, and the failed recursive calls are avoided.
\subsection{LZ77 Factorization}

Let us derive our quantum LZ77 factorization algorithm by attempting a straightforward approach and seeing why it fails. Recall that the LZ77 factorization of a text $X$ can be found by processing text from left to right. Assume we are trying to determine $i$-th factor in the LZ77 factorization. We use $s_i$ to denote the starting location of this factor and suppose the factorization of the prefix $X[1\dd s_i)$ has been already computed. We next need to determine the largest $\ell_i \geq 0$ such that $X[s_i\dd s_i+\ell_i)$ has an occurrence starting at some position $p_i < s_i$.
To do this efficiently, we wish to maintain a data structure over the previously processed text. This data structure should (i) occupy $\Ohtilde(\sqrt{zn})$ space, (ii) allow us to determine the largest $\ell_i$ as described above efficiently, and (iii) support efficient insertions once we find the new factor. The target time complexity for adding a factor of length $\ell_i$ is $\Ohtilde(\sqrt{\ell_i})$. If we can achieve this, then the entire LZ77 factorization will be found in time polylogarithmic factors from $\sum_{i=1}^z \sqrt{\ell_i} \leq \sqrt{zn}$, where we used that $\sum_{i=1}^z \ell_i = n$.

\paragraph{Suffix-Array Approach}
As an initial attempt, consider maintaining the suffix array of the reversed prefix $\rev{X[1\dd s_i)}$. Doing so allows us to determine the $k$-th co-lexicographically largest prefix of $X[1\dd s_i)$ in constant time.  Based on this, we can try to find the \emph{non-overlapping} LZ77 factorization, a variation of LZ77 that requires $p_i+\ell_i\le s_i$ and results in a factorization size within a logarithmic factor away from the optimum. 
To determine the largest $\ell_i$ such that $X[s_i\dd s_i+\ell_i)$ matches a substring of $X[1\dd s_i)$, we combine exponential search on $\ell_i$ and binary search on this sorted set of prefixes. 
To check, for a particular $\ell_i$, whether $X[s_i\dd s_i+\ell_i)$ has an earlier occurrence, we start with the median of the sorted prefixes and apply Grover search (\cref{thm:grover}) in $\Ohtilde(\sqrt{\ell_i})$ time to determine if a mismatch exists between $X[s_i\dd s_i+\ell_i)$ and the median prefix. If a mismatch exists, we determine whether $X[s_i\dd s_i+\ell)$ is co-lexicographically larger or smaller, and continue the binary search accordingly. This technique allows us to determine the $i$-th factor in $\Ohtilde(\sqrt{\ell_i})$ time. 
Unfortunately, the time for updating the suffix array is $\Oh(\ell_i)$, resulting in a linear time overall.

\paragraph{LZ-End Approach}
A natural approach to overcome this is to use the LZ-End factorization rather than non-overlapping LZ77 factorization. Recall that LZ-End differs in that each new factor is either the first occurrence of a symbol or the longest substring whose earlier occurrence ends at the end of a previous factor. It was recently shown that the number of factors in the LZ-End factorization satisfies $z_e = \Oh(z\log^2 n)$~\cite{DBLP:conf/soda/KempaS22}. As discussed more in \cref{sec:other_encodings}, once the LZ-End factorization is computed, we can obtain the actual LZ77 factorization in $\Ohtilde(z)$ time. The reason we consider LZ-End is that, since each (non-trivial) factor has an occurrence ending at the end of a previous factor, we only need to maintain the co-lexicographically sorted order of the prefixes $X[1\dd s_{i'})$ for $i' \le i$. For a given $\ell$, this order suffices for checking in $\Ohtilde(\sqrt{\ell})$ time if $X[s_i\dd s_i+\ell)$ has a previous occurrence ending at the end of a previous factor. However, this idea alone does not work because exponential search on $\ell$ may fail due to the lack of monotonicity. In other words, $X[s_i\dd s_i+\ell)$ may have an earlier occurrence ending at a previous factor
while $X[s_i\dd s_i+\ell-1)$ does not have one.

\paragraph{\boldmath LZ-End+$\tau$ Approach}To facilitate using exponential search for finding the next factor length, we introduce a variation on LZ-End that we call LZ-End+$\tau$. We define LZ-End+$\tau$ by making each new factor $X[s_i\dd s_i+\ell_i)$ either the first occurrence of a new symbol or the longest substring
whose earlier occurrence is of the form $X[p_i\dd q_i)$, where $q_i \le s_i$ satisfies $q_i \bmod \tau = 0$ or $q_i = s_{i'}$ for $i'\le i$.
We let $z_{e+\tau}$ denote the number of LZ-End+$\tau$ factors. A crucial observation is that we still have $z_{e+\tau} = \Oh(z\log^2 n)$. Additionally, the number of prefixes that have to be checked for finding LZ-End+$\tau$ factors is $\Ohtilde(z_{e+\tau} + n/\tau)$. 
Suppose we are trying to find the next factor starting at index~$s_i$. 
We say the \emph{$\tau$-far property} holds for an index $j \ge s_i$  if there exists $h \in [\max(s_i, j-\tau)\dd  j]$ such that $X[s_i\dd h)$ has an earlier occurrence of the form $X[p\dd q)$ with $q\le s_i$ satisfying $q \bmod \tau = 0$ or $q=s_{i'}$ for $i'\le i$.
By definition, the $\tau$-far property is monotone, in that if it holds for $j>s_i$, then it also holds for $j-1$. Hence, the problem reduces to testing the $\tau$-far property for a given $j$ and maintaining the order of prefixes $X[1\dd q)$ for $q\le s_i$ satisfying $q \bmod \tau = 0$ or $q=s_{i'}$ for $i'\le i$.

We first consider the problem of checking whether the $\tau$-far property holds for a given $j \ge s_i$. Our algorithm utilizes the dynamic LCE data structure of Nishimoto et al.~\cite{DBLP:conf/mfcs/NishimotoIIBT16}.
Starting from an initially empty string, it supports insertions of individual characters and arbitrary substrings of the existing text in $\Ohtilde(1)$ time. 
In addition, it answers LCE queries between two arbitrary indices in $\Ohtilde(1)$ time. To describe the process of finding the next factor $X[s_i\dd s_i+\ell_i)$, we assume that we have the co-lexicographic ordering of the selected prefixes $X[1\dd q)$ (across all $q\le s_i$ such that $q \bmod \tau = 0$ or $q=s_{i'}$ for $i'\le i$). We denote this set of prefixes as $\mathcal{P}_{i-1}$.
We also assume that the dynamic LCE data structure has been constructed for $\rev{X[1\dd s_i)}$. 
Our solution first creates at most $\tau+1$ ranges of indices into $\mathcal{P}_{i-1}$, each corresponding to the prefixes in $\mathcal{P}_{i-1}$ that have $X[\max(s_i, j-\tau)\dd h)$ as a suffix. This is accomplished in $\Ohtilde(\tau)$ time by using the LCE data structure. Next, we observe that, for arbitrary $k$, we can find the $k$-th co-lexicographically largest prefix contained in these ranges in $\Ohtilde(\tau)$ time. Based on this observation, we can then apply binary search and the Grover search algorithm to determine if the $\tau$-far property holds for $j$, similar to the LZ-End case. Combining with exponential search on $j$, we can find the next factor of length $\ell_i$ in $\Ohtilde(\tau + \sqrt{\ell_i})$ time.

To update the sorted list of prefixes $\mathcal{P}_{i-1}$ to $\mathcal{P}_i$, we again utilize the dynamic LCE data structure. Assume we just determined the new factor $X[s_i\dd s_{i+1})$. Then, this new factor is either a previously occurring substring with a location determined in the previous step or the first occurrence of a symbol. As mentioned, the LCE data structure supports appending such a substring in $\Ohtilde(1)$ time. To insert into the sorted order the new prefixes $X[1\dd q)$ with $q \in (s_i\dd s_{i+1}]$ such that $q \bmod \tau = 0$ or $q = s_{i+1}$, we apply LCE queries to compare $X[1\dd q)$ to the prefixes in the currently sorted list as needed, resulting in $\Oh(\log n)$ queries and $\Ohtilde(1)$ time per inserted prefix.

Overall, the algorithm takes $\Ohtilde(\frac{n}{\tau}+\sqrt{z_{e+\tau}n}+z_{e+\tau}\tau)$ time, which is $\Ohtilde(\sqrt{z_{e+\tau}n})=\Oh(\sqrt{zn})$
for optimal $\tau$. To complete the proof of \cref{thm:lz}, we convert the LZ-End+$\tau$ factorization to an LZ77 factorization in $\Ohtilde(z_{e+\tau})=\Ohtilde(z)$ time using  data structure presented in~\cite{DBLP:journals/cacm/KempaK22}, which allows us to determine the leftmost occurrence of any substring in $\Ohtilde(1)$ time.

\subsection{Compressed Text Indexing and Applications}

We next outline the construction of the
indexes from \cref{thm:index}. 
Let  $\tau = \Theta(\sqrt{n/r})$. 
The first step is to obtain a (less efficient) index in time $\Ohtilde(n/\tau+\tau r) = \Ohtilde(\sqrt{nr})$ that can support $\SA$ and $\ISA$ queries in time $\Ohtilde(\tau)$.
This is done by preprocessing the RL-BWT encoding, obtained using \cref{cor:alg_rlbwt},
so that we can determine in $\Ohtilde(1)$ time the result of applying the $\LF$ mapping a total of $\tau$ times starting at position $i$, i.e., $\LF^\tau[i]$. The computation utilizes prefix doubling and alphabet replacement techniques. Thanks to the prefix doubling, we only need to perform $\Oh(\log \tau)$ alphabet replacement steps, making the overall computation time $\Ohtilde(\tau r)$. We next 
take $n/\tau$ $\SA$ samples, evenly spaced by text position, in $\Ohtilde(n/\tau)$ time. These  samples make any $\SA$ value computable in $\Ohtilde(\tau)$ time. Combined with the LCE data structure, this supports $\ISA$ queries in $\Ohtilde(\tau)$ time.

The $\Oh(r)$-space index for pattern matching can be constructed in $\Ohtilde(r)$ time, given the text positions corresponding to BWT-run boundaries~\cite{DBLP:conf/soda/GagieNP18}. The later can be achieved via $\Oh(r)$ $\SA$ queries using  the index described above. 
Next, we demonstrate that the suffix array index described in~\cite{DBLP:journals/jacm/GagieNP20} can also be constructed using $\Ohtilde(r)$ 
queries. The main technical challenge lies in efficiently determining, for a given range of $[s\dd e]$ of indices in RL-BWT, the smallest $k \geq 0$ such that $\LF^k([s\dd e])$ contains a BWT run-boundary as well as the interval $\LF^k([s\dd e])$ itself. We outline how to accomplish this using $\Ohtilde(1)$ queries on our less-efficient index. Thus, in both cases, the construction time is $\Ohtilde(\sqrt{rn})$. 
We also maintain the $\Oh(z\log \frac{n}{z})$ space LCE structure~\cite{I17} so that $\ISA$ queries can be supported in $\Ohtilde(1)$ time using binary search on $\SA$ values. 
These indexes allow us to solve several fundamental problems efficiently, as described in \cref{sec:applications}.
\section{Quantum Algorithm for Bounded Edit Distance}\label{sec:edit}

\subsection{Recursive algorithm}
\label{subsec:edit-recur}

Our quantum algorithm for computing edit distance can be viewed as a deterministic classical algorithm that makes the following oracle calls: (1) compute the LZ77 factorization of a substring of the input strings; (2) check the equality of two substrings of the input strings.  These two tasks can be efficiently solved using quantum subroutines, by \cref{thm:lz} and Grover search  (\cref{thm:grover}). 
By applying error reduction (via a logarithmic number of repetitions) and a global union bound, we assume that all these quantum subroutines work correctly, so in the following, we can ignore the analysis of failure probabilities.

We need the following notion.
\begin{definition}[Edit anchor]\label{defn:anchor}
 	We say that $(x,y) \in  [i\dd j]\times [i'\dd j']$ is an \emph{edit anchor} of fragments $X(i\dd j]$ and $Y(i'\dd j']$ if 
    $\ed(X(i\dd j],Y(i'\dd j'])=\ed(X(i \dd x],Y(i'\dd y])+\ed(X(x \dd j],Y(y\dd j'])$, that is, $(x,y)\in \caA$ for some optimal alignment $\caA:X(i\dd j]\onto Y(i'\dd j']$.

    Moreover, for an integer $k\ge 0$, we say that $(x,y)$ is a $k$-edit anchor of $X(i\dd j]$ and $Y(i'\dd j']$ if 
    $(x,y)$ is their edit anchor or  $\ed(X(i\dd j],Y(i'\dd j']) > k$.
\end{definition}
\SetKwFunction{FindAnchor}{FindAnchor}
\SetKwFunction{IsAnchor}{IsAnchor}

We will prove the following two lemmas in \cref{sec:anchor}.
\begin{restatable}[Finding an anchor: {$\protect\FindAnchor(X,Y,k,x)$}]{lemma}{find}\label{lem:findanchor-algo}
    There exists a quantum algorithm that, given strings $X,Y\in \Sigma^{\le n}$, an integer $k\ge 1$, and a position $x\in [0\dd |X|]$,
    finds a position $y\in [0\dd |Y|]$ such that $(x,y)$ is a $k$-edit anchor of $X,Y$. The algorithm has query complexity $\Ohtilde(\sqrt{kn})$ and time complexity $\Ohtilde(\sqrt{kn} + k^2)$.
\end{restatable}

\begin{restatable}[Testing an anchor: {$\protect\IsAnchor(X,Y,k,(x,y))$}]{lemma}{keepordiscard}\label{lem:keepordiscard-algo}
    There exists a quantum algorithm that, given strings $X,Y\in \Sigma^{\le n}$, an integer $k\ge 1$, and a pair $(x,y)\in [0\dd |X|]\times [0\dd |Y|]$,
    \begin{itemize}
        \item if $\ed(X,Y)\le k$, decides whether $(x,y)$ is an edit anchor of $X,Y$;
        \item otherwise, returns an arbitrary answer.
    \end{itemize}
    The algorithm has query complexity $\Ohtilde(\sqrt{kn})$ and time complexity $\Ohtilde(\sqrt{kn} + k^2)$.
\end{restatable}

\newcommand{\solve}{\mathsf{Solve}}

The recursive algorithm for computing $\ed(X,Y)$ is given in \cref{alg:recur}.
The outermost function call is $\solve(X,Y)$ with preconditions $X\neq Y$ and $|X|>0$; the cases of $X=Y$ and $|X|=0$ can be handled at the very beginning. 
As discussed earlier in the technical overview, \cref{alg:recur} may pause the recursive calls it makes after they have spent certain amount of time or queries.
In order to cleanly formalize these conditions, we consider two types of \emph{tokens}, called q-tokens and t-tokens, respectively, that our algorithm \emph{burns} in \cref{line:burn0,line:burn}. 
Recursive calls can be paused (or terminated) if they exceed certain quotas for the number of burnt tokens; see \cref{line:running}.

Let $n = \max\{|X|,|Y|,2\}$ be the global input length, and define a global parameter $r = \lceil 5 \log n\rceil$.
We will use the following function
\[T_q(|X|,|Y|,d):= 10 \sqrt{d\cdot (|X|+|Y|)}\cdot r^{3}\cdot \left (\tfrac{r+2}{r}\right )^{\lceil \log |X| \rceil }\]
to measure the number q-tokens burnt by \cref{alg:recur} and its recursive calls.  Analogously, we measure the number of t-tokens burnt using the following function:
\[T_t(|X|,|Y|,d):= 10 \cdot d^2\cdot r^{9}\cdot \left (\tfrac{r+2}{r}\right )^{\lceil \log |X| \rceil }.\]
\begin{algorithm}[t]
\DontPrintSemicolon
\caption{$\solve(X,Y)$ (preconditions: $X\neq Y$ and $|X|>0$)}
\label{alg:recur}
\If{$|X|= 1$}{
    Burn $(|Y|+1)$ q-tokens and $(|Y|+1)$ t-tokens\;\label{line:burn0}
    \For{$y\gets 1$ \KwSty{to} $|Y|$}{
        \lIf{$X[1]=Y[y]$}{\Return{$|Y|-1$}}\label{line:onea}
    }
    \Return{$\max(1,|Y|)$}\;\label{line:oneb}
}
Initialize anchor $a \gets (\bot,\bot)$\\
Initialize program $A \gets \bot$\\
\For{$i \gets  0,1,2\dots$}{
    Burn $\sqrt{r^{2i+2}(|X|+|Y|)}$ q-tokens and $(r^{2i+2})^2$ t-tokens\label{line:burn}\;
    \If(\tcp*[f]{\cref{lem:keepordiscard-algo}}){\KwSty{not} $\IsAnchor(X,Y,r^{2i+2},a)$}{\label{line:checkanchor}
        $a \gets \FindAnchor(X,Y,r^{2i+2}, \lceil{|X|/2}\rceil)$\tcp*[r]{\cref{lem:findanchor-algo}} \label{line:findanchor}
        Check whether $X(0\dd a_x]=Y(0\dd a_y]$\tcp*[r]{\cref{thm:grover}}\label{line:eq1}
        Check whether $X(a_x\dd |X|]=Y(a_y \dd |Y|]$\tcp*[r]{\cref{thm:grover}}\label{line:eq2}
        Define program $A_i := \big [\textbf{return } \solve(X(0\dd a_x],Y(0\dd a_y]) + \solve(X(a_x\dd |X|],Y(a_y \dd |Y|]) \big ]$, skipping the corresponding recursive call if equality is found in \cref{line:eq1} or~\ref{line:eq2}\label{line:program}\\
        Terminate current $A$, and redefine $A\gets A_i$\label{line:terminate}
    }
    Resume $A$ and run it until it attempts to burn more than $T_q(|X|,|Y|,r^{2i})$ q-tokens or more than $T_t(|X|,|Y|, r^{2i})$ t-tokens since its beginning \label{line:running}\\
    \If{$A$ \upshape{has already finished, with return value} $d$} {
       \lIf{$d< r^{2i+2}$}{\Return{$d$}\label{line:return}}
    }
}
\end{algorithm} 

Our main claim is the following:
\begin{lemma}\label{lem:edit-query-time}
Given strings $X,Y\in \Sigma^{\le n}$ satisfying $|X|>0$ and $\ed(X,Y) =  d\ge 1$, the procedure $\solve(X,Y)$ correctly returns $d$, 
and (including recursive calls) burns at most $T_q(|X|,|Y|,d)$ q-tokens and at most $T_t(|X|,|Y|,d)$ t-tokens.
\end{lemma}
\begin{proof}
We first prove the correctness of \cref{alg:recur}.  
Suppose $|X|\ge 2$ (otherwise, the correctness is clear; see \cref{line:onea,line:oneb})
and denote $j=\lfloor{\log_{r^2} d}\rfloor$.
Observe that \cref{alg:recur} always returns the cost of a valid alignment between $X$ and $Y$. We need to show it indeed returns the cost of an optimal alignment.
Due to the check at \cref{line:return}, the algorithm can only return in iteration $i\ge j$.
Starting from iteration $j$, the program $A$ is based on an $r^{2j+2}$-edit anchor $a$ (due to \cref{line:checkanchor,line:findanchor}), which belongs to an optimal alignment of $X,Y$ by \cref{defn:anchor} since $d< r^{2j+2}$. Hence, once program $A$ terminates, it indeed returns the correct answer $d=\ed(X,Y)$ (assuming its recursive calls are correct, by induction).

Next, we shall prove that the algorithm does not burn too many tokens.
If $|X|= 1$, then $|Y|+1\le 3d$ because $d\ge \max(1,|Y|-1)$.
Thus, the number of burnt q-tokens satisfies \[|Y|+1 \le \sqrt{3d\cdot (|Y|+1)}<10\sqrt{d\cdot (|X|+|Y|)} < T_q(|X|,|Y|,d).\]
Similarly, the number of burnt t-tokens is 
\[|Y|+1 \le 3d < 10d^2 < T_t(|X|,|Y|,d).\]
Henceforth, we assume that $|X|\ge 2$ and analyze the number of burnt tokens in three parts.
\begin{itemize}
    \item 
We first consider the program $A$ that is running during iteration $j$.
Note that later iterations $i>j$ will not kill program $A$: due to $d < r^{2j+2}$, the procedure 
$\IsAnchor(X,Y,r^{2i+2},a)$ will not discard the $r^{2j+2}$-edit anchor $a$. So $A$ is the program that eventually returns on \cref{line:return}. 
The program $A$ consists of two recursive calls,
$ \solve(X(0\dd a_x],Y(0\dd a_y]) $ and $ \solve(X(a_x\dd |X|],\allowbreak Y(a_y \dd |Y|])$
and, by induction, they return edit distances $d_1$ and $d_2$, respectively, with $d_1+d_2 = d$.
 If either one of $d_1,d_2$ is zero, then the corresponding recursive call is skipped and does not burn any tokens.
 Let $n_1 =  |X(0\dd a_x]|$ and $n_2 = |X(a_x\dd |X|]|$, which satisfy $n_1=\lceil\frac12 |X|\rceil \ge \lfloor\frac12|X|\rfloor = n_2$ by definition of $a$ at \cref{line:findanchor}.
 Let  $m_1 =  |Y(0\dd a_y]|$ and $m_2 = |Y(a_y \dd |Y|]|$.
By the inductive hypothesis, and due to the Cauchy--Schwarz inequality, the number of q-tokens burnt by the two recursive calls is at most
\begin{align*}
    T_q(n_1,m_1,d_1) + T_q(n_2,m_2,d_2) &\le  
	10\left (\sqrt{d_1(n_1+m_1)} + \sqrt{d_2(n_2+m_2)}\right )\cdot r^{3} \cdot \left (\tfrac{r+2}{r}\right )^{\lceil \log n_1 \rceil } \\
    & \le 
	10 \sqrt{d\cdot (|X|+|Y|)} \cdot r^{3} \cdot \left (\tfrac{r+2}{r}\right )^{\lceil \log |X|\rceil -1}\\
    & = \tfrac{r}{r+2} \cdot T_q(|X|,|Y|,d).
\end{align*}	
The number of burnt t-tokens satisfies
\begin{align*} T_t(n_1,m_1,d_1) + T_t(n_2,m_2,d_2) & \le 10 \cdot (d_1^2 + d_2^2)\cdot r^{9} \cdot \left (\tfrac{r+2}{r}\right )^{\lceil \log n_1 \rceil }\\
    & \le 10\cdot d^2 \cdot r^{9}\cdot \left (\tfrac{r+2}{r}\right )^{\lceil \log |X|\rceil -1} \\
    & =  \tfrac{r}{r+2} \cdot T_t(|X|,|Y|,d).
\end{align*}
Note that these two sums are smaller than  $T_q(|X|,|Y|,r^{2j+2})$ and $T_t(|X|,|Y|,r^{2j+2})$ respectively, which means that $A$ will finish running in or before iteration $j+1$, since $A$ will have burnt at most $T_q(|X|,|Y|,r^{2j+2})$ q-tokens and at most $T_t(|X|,|Y|,r^{2j+2})$ t-tokens in iteration $j+1$ at \cref{line:running}.

\item Now, we analyze the number of tokens burnt on \cref{line:burn}.
For each iteration $i\ge 0$, we burned $\sqrt{r^{2i+2}(|X|+|Y|)}$ q-tokens and $r^{4i+4}$ t-tokens.
Since we only performed iterations $i\in [0\dd {j+1}]$, the total number of q-tokens burnt is at most 
\begin{align*}
 \sum_{i=0}^{j+1} \sqrt{r^{2i+2}(|X|+|Y|)} &=  \sum_{i=0}^{j+1}  r^{i+1}\sqrt{|X|+|Y|}  \\
 & \le \tfrac{1}{r-1} \cdot r^{j+3} \sqrt{|X|+|Y|}\\
 & \le \tfrac{1}{r-1} \cdot \sqrt{d}\cdot \sqrt{|X|+|Y|}\cdot r^3 \cdot \left (\tfrac{r+2}{r}\right )^{\lceil \log |X| \rceil }\\
    & = \tfrac{0.1}{r-1}\cdot  T_q(|X|,|Y|,d).
\end{align*}
The total number of t-tokens burnt, on the other hand, is at most
\begin{align*}\sum_{i=0}^{j+1} (r^{2i+2})^2 &= \sum_{i=0}^{j+1} r^{4i+4}\\
    &\le \tfrac{1}{r^4-1}\cdot r^{4j+12}\\
    &\le \tfrac{1}{r-1}\cdot r^{4j+9}\\
    &\le \tfrac{1}{r-1}\cdot d^2 \cdot r^{9}\cdot  \left (\tfrac{r+2}{r}\right )^{\lceil \log |X| \rceil }\\
    &= \tfrac{0.1}{r-1}\cdot  T_t(|X|,|Y|,d).
    \end{align*}

\item Now, we analyze the number of tokens burnt by earlier programs that were terminated. 
Suppose the correct anchor for program $A$ was computed in iteration $i^*\le j$.
Then, every previous wrong anchor was terminated in some iteration $i\in [1\dd i^*]$, where we found it did not pass the check at \cref{line:checkanchor}, which means that the corresponding wrong program has only burnt at most $T_q(|X|,|Y|, r^{2i-2})$ q-tokens and at most  $T_t(|X|,|Y|, r^{2i-2})$ t-tokens before it was paused at \cref{line:running} in iteration $i-1$.
Summing over all such possible wrong executions, the number of q-tokens is at most
\begin{align*}
     \sum_{i=1}^{i^*} T_q(|X|,|Y|,r^{2i-2}) &\le \sum_{i=1}^{j} 10\sqrt{r^{2i-2}\cdot (|X|+|Y|)}\cdot r^{3} \cdot \left (\tfrac{r+2}{r}\right )^{\lceil \log |X| \rceil }\\
     & < \tfrac{10\cdot r^{j}}{r-1}\cdot \sqrt{|X|+|Y|} \cdot r^{3}\cdot \left (\tfrac{r+2}{r}\right )^{\lceil \log |X| \rceil }\\
     & \le \tfrac{10}{r-1}\cdot \sqrt{d(|X|+|Y|)} \cdot r^{3}\cdot \left (\tfrac{r+2}{r}\right )^{\lceil \log |X| \rceil }\\
     & = \tfrac{1}{r-1}\cdot T_q(|X|,|Y|,d).
\end{align*}
The total number of t-tokens burnt by these terminated calls is at most
\begin{align*} \sum_{i=1}^{i^*} T_t(|X|,|Y|,r^{2(i-1)}) &\le \sum_{i=1}^{j} 10\cdot r^{4(i-1)}\cdot r^{9}\cdot  \left (\tfrac{r+2}{r}\right )^{\lceil \log |X| \rceil }\\
    & \le \tfrac{10}{r^4-1}\cdot r^{4j}\cdot r^{9}\cdot  \left (\tfrac{r+2}{r}\right )^{\lceil \log |X| \rceil }\\
    & \le \tfrac{10}{r-1} \cdot d^2 \cdot r^{9}\cdot \left (\tfrac{r+2}{r}\right )^{\lceil \log |X| \rceil } \\
    & = \tfrac{1}{r-1}\cdot T_t(|X|,|Y|,d).
\end{align*}
\end{itemize}
Finally, summing up the three parts, the total number of q-tokens burnt by \cref{alg:recur} is at most 
\[ T_q(|X|,|Y|,d) \cdot \left( \tfrac{r}{r+2} + \tfrac{0.1}{r-1} + \tfrac{1}{r-1}\right) \le T_q(|X|,|Y|,d),\]
and similarly the number of burnt t-tokens is at most
\[ T_t(|X|,|Y|,d) \cdot \left( \tfrac{r}{r+2} + \tfrac{0.1}{r-1} + \tfrac{1}{r-1}\right) \le T_t(|X|,|Y|,d),\]
where we used $r \ge 5 > \tfrac{14}{3}$.
\end{proof}

Next, we describe a (classical) scheduler that is used in the implementation of \cref{alg:recur} to keep track of the quotas for the number of burnt tokens.
Recall that the recursion of \cref{alg:recur} has $\lceil \log n\rceil$ levels.  When we are at a certain node $v$ of the recursion tree, each ancestor node $p$ holds two counters $q_p,t_p$ that keep track of the remaining tokens that the program corresponding to $p$ is allowed to burn.
When the current node $v$ attempts to burn $T$ t-tokens and $Q$ q-tokens (\cref{line:burn}), we check the quotas of all ancestors $p$ of $v$.
If $T\le t_p$ and $Q\le q_p$ holds for all ancestors $p$ of $v$, then we can safely burn the tokens and decrease the quotas, setting $t_p \gets t_p-T$ and $q_p \gets q_p-Q$ for all ancestors $p$. Otherwise, we choose the nearest ancestor $p$ with $t_p<T$ or $q_p<Q$ and pass control from $v$ to the parent of $p$. We also save a back pointer to $v$ so that we know we should resume at~$v$ if the program corresponding to $p$ is resumed with increased quotas.
This scheduler incurs an $\Oh(\log n)$-time additive overhead at \cref{line:running,line:burn}.

It remains to analyze the complexity of \cref{alg:recur}.
By \cref{lem:findanchor-algo,lem:keepordiscard-algo,thm:grover}, the oracle calls in \cref{line:checkanchor,line:findanchor,line:eq1,line:eq2}
make $\Ohtilde(\sqrt{r^{2i+2}(|X|+|Y|)})$ quantum queries and take $\Ohtilde(\sqrt{r^{2i+2}(|X|+|Y|)}+(r^{2i+2})^2)$ quantum time (including the amplification of success probability).  We charge these queries to q-tokens burnt at \cref{line:burn},
whereas the running time is charged to t-tokens burnt at \cref{line:burn}.
\Cref{line:onea,line:oneb} make $\Oh(1+|Y|)$ quantum queries and run in $\Oh(1+|Y|)$ quantum time,
which we charge to the q-tokens and t-tokens burnt at \cref{line:burn0}.
No other subroutines make any quantum queries.
The classical control instructions (including the scheduler implementation) take $\Oh(\log n)=\Ohtilde(1)$ time per iteration of the main \KwSty{for} loop, which we can also charge to the number of tokens burnt. 
Overall, \cref{lem:edit-query-time} implies that the total query complexity is $\Ohtilde(T_q(|X|,|Y|,d))=\Ohtilde(\sqrt{dn})$,
whereas the time complexity is $\Ohtilde(T_q(|X|,|Y|,d)+T_t(|X|,|Y|,d))=\Ohtilde(\sqrt{dn}+d^2)$.

It remains to explain how to modify \cref{alg:recur} so that a witness sequence of edits is reported along with every distance.
Each recursive call remembers locations of currently processed strings $X$ and $Y$ within the global inputs so that the edits reported use global position numbering.
Technically, each call to $\solve(X,Y)$, along with the answer $d$, reports a linked list of $d$ edits that allow transforming $X$ into $Y$.
If the algorithm terminates at \cref{line:onea}, we report insertions of $Y[1],\ldots,Y[y-1],Y[y+1],\ldots,Y[|Y|]$.
If the algorithm terminates at \cref{line:oneb} with $|Y|=0$, we report a deletion of $X[1]$.
If the algorithm terminates at \cref{line:oneb} with $|Y|>0$, we report a substitution of $X[1]$ for $Y[1]$
and insertions of $Y[2],\ldots,Y[|Y|]$.
The program $A_i$ defined in \cref{line:program} not only adds the distances but also concatenates the lists reported by the recursive calls (skipped calls correspond to empty lists).
If the algorithm terminates at \cref{line:return}, we pass the list returned by $A$ along with the answer $d$.
In all cases, the extra time needed to handle edits is proportional to the time complexity of control instructions.
This completes the proof of \cref{thm:edit-main}.

\subsection{Finding and Testing Anchors}
\label{sec:anchor}
Similar to \cite{KPS21}, we use connections between LZ77 factorization and edit distance alignments.

\newcommand{\hi}{\hat{\imath}}
\newcommand{\hj}{\hat{\jmath}}
\begin{lemma}[Disjoint alignments imply compression]\label{lem:disjointalignmentLZ}
	Consider strings $X,Y\in \Sigma^*$ and alignments $\caA : X(i\dd j]\onto Y$ and $\caA':X(i'\dd j']\onto Y$.  If $\caA \cap \caA' =  \emptyset$, then
    \[ |\LZ(X(\hi \dd \hj])| \le |i-i'| + 2\ed_{\caA}(X(i\dd j],Y) + 2\ed_{\caA'}(X(i'\dd j'],Y) + 1\]
    holds for every fragment $X(\hi \dd \hj]$ of $X$ with $\min\{i,i'\}\le \hi \le \hj \le \max\{j,j'\}$.
\end{lemma}
\begin{proof}
Let $\caB = (\caA')^{-1}\circ \caA: X(i\dd j]\onto X(i'\dd j']$ be an alignment obtained as a product of $(\caA')^{-1}$ and $\caA$.
Note that $\ed_\caB(X(i\dd j], X(i'\dd j']) \le \ed_{\caA}(X(i\dd j],Y) + \ed_{\caA'}(X(i'\dd j'],Y)$ and, for every $(x,x')\in \caB$, there is $y\in [0\dd |Y|]$ such that $(x,y) \in \caA$ and $(x',y) \in \caA'$. 
    Since $\caA,\caA'$ are disjoint, we must have $(x,y)\notin \caA'$, and hence $x\neq x'$ for all $(x,x')\in \caB$.
    By symmetry, we assume without loss of generality that $i < i'$; then, $x<x'$ holds for all $(x,x')\in \caB$ and, in particular, $j < j'$.
    We consider two cases:
    \begin{itemize}
        \item If $\hi \ge j$, then $|\LZ(X(\hi \dd \hj])|\le \hj-\hi \le j'-j \le i'-i+\ed_\caB(X(i\dd j], X(i'\dd j'])  \le (i'-i)+ \ed_{\caA}(X(i\dd j],Y) + \ed_{\caA'}(X(i'\dd j'],Y)$.
        \item Otherwise, there a position $\hi'\in [i'\dd j']$ such that $(\hi,\hi')\in \caB$.
    The alignment $\caB$ induces a decomposition of $X(\hi'\dd j']=f'_1\cdots f'_z$ into $z\le 2\cdot \ed_{\caB}(X(\hi\dd j], X(\hi'\dd j']) + 1$ phrases, each of which is either a single character inserted or substituted under $\caB$ or a fragment $X(x'\dd x'+s]$ such that $X(x\dd x+s] \simeq_{\caB} X(x'\dd x'+s]$ for some $(x,x')\in \caB$. 
    Since $x<x'$ for all $(x,x')\in \caB$, this implies an LZ-like factorization $X[\hi+1]\cdots X[\hi']\cdot f'_1\cdots f'_z$ of $X(\hi\dd j']$.
    Hence, $|\LZ(X(\hi \dd \hj])|\le |\LZ(X(\hi \dd j'] )| \le (\hi'-\hi)+z \le (i'-i)+\ed_{\caB}(X(i\dd \hi], X(i'\dd \hi']) + 2\cdot \ed_{\caB}(X(\hi\dd j], X(\hi'\dd j']) + 1 \le (i'-i)+2\ed_\caB(X(i\dd j], X(i'\dd j'])+1 \le (i'-i)+ 2 \ed_{\caA}(X(i\dd j],Y) + 2\ed_{\caA'}(X(i'\dd j'],Y) +1$. 
    \qedhere
    \end{itemize}
\end{proof}

\begin{lemma}\label{lem:intersect}
    Consider strings $X,Y\in \Sigma^{*}$, an integer $k\ge 0$, and a pair $(x,y)\in [0\dd |X|]\times [0\dd |Y|]$.
    If $\ed(X,Y)\le k$, then $(x,y)$ is an edit anchor of $X$ and $Y$ if and only if it is an edit anchor of fragments $X'=X(i\dd j]$ and $Y'=Y(i\dd j+|Y|-|X|]$ defined in terms of the minimum $i\in [0\dd x]$ such that $|\LZ(\rev{X(i\dd x]})|\le 6k+2$ and the maximum $j\in [x\dd |X|]$ such that $| \LZ(X(x\dd j])|\le 6k+2$.
\end{lemma}
\begin{proof}
    Consider optimal alignments $\caA:X\onto Y$ and $\caA' : X'\onto Y' $ such that $(x,y)\in \caA\cup \caA'$.
    The monotonicity of edit distance guarantees $\ed_{\caA'}(X',Y')\le \ed_{\caA}(X,Y) \le k$.

    We will prove the existence of $(x_\ell,y_\ell)\in \caA\cap \caA'$  such that $x_\ell\le x$ and $y_\ell \le y$.
    We proceed with a proof by contradiction and bound $|\LZ(\rev{X(i \dd x]})|$ by considering the following two cases:
  \begin{itemize}
      \item $(x,y) \in \caA$. Suppose that $\caA'$ aligns $X(i \dd x']$ with $Y(i \dd y]$, whereas $\caA$ aligns $X(i' \dd x]$ with $Y(i \dd y]$.
      These alignments are disjoint; otherwise, their intersection point $(x_\ell,y_\ell) \in \caA \cap \caA'$ satisfies $(x_\ell,y_\ell) \le (x,y)$.
      By \cref{lem:disjointalignmentLZ} applied to the alignment of $\rev{X(i \dd x']}$
      and $\rev{Y(i \dd y]}$ and the alignment of $\rev{X(i' \dd x]}$ and $\rev{Y(i \dd y]}$, we have
      \[ |\LZ(\rev{X(i \dd x]})| \le |x-x'| + 2\ed_{\caA}(X(i' \dd x], Y(i \dd y])+2\ed_{\caA'}(X(i \dd x'], Y(i \dd y])+1.\]
      \item $(x,y)\in \caA'$.       Suppose that $\caA'$ aligns $X(i \dd x]$ with $Y(i \dd y]$, whereas $\caA$ aligns $X(i'\dd x']$ with $Y(i \dd y]$.
      These alignments are disjoint; otherwise, their intersection point $(x_\ell,y_\ell) \in \caA \cap \caA'$ satisfies $(x_\ell,y_\ell) \le (x,y)$. By \cref{lem:disjointalignmentLZ} applied to the alignment of $\rev{X(i'\dd x']}$ and $\rev{Y(i \dd y]}$ and the alignment of $\rev{X(i \dd x]}$ and $\rev{Y(i \dd y]}$, we have
      \[ |\LZ(\rev{X(i \dd x]})| \le |x-x'| + 2\ed_{\caA}(X(i'\dd x'], Y(i \dd y])+2\ed_{\caA'}(X(i \dd x], Y(i \dd y])+1.\]
  \end{itemize}
  In both cases, we obtained
  \[ |\LZ(\rev{X(i \dd x]})| \le |x-x'| + 2\ed_\caA(X,Y)+2\ed_{\caA'}(X',Y')+1 \le 6k+1,\]
  where the last inequality follows from $|x-x'| \le |x-y| + |y-x'|\le \ed_\caA(X,Y)+\ed_{\caA'}(X',Y')\le 2k$.
  If $i>0$, this implies $|\LZ(\rev{X(i-1\dd x]})| \le  |\LZ(\rev{X(i \dd x]})| + 1\le (6k+1) + 1  = 6k+2$,
  which contradicts the definition of $i$.
  Consequently, we may assume that $i=0$. In that case, however,  $(0,0)\in \caA\cap \caA'$ is an intersection point satisfying $ 0\le x$ and $0\le y$.
  This completes the existence proof of $(x_\ell,y_\ell)\in \caA\cap \caA'$  such that $x_\ell\le x$ and $y_\ell \le y$.
  A symmetric argument yields $(x_r,y_r)\in \caA\cap \caA'$  such that $x_r \ge x$ and $y_r \ge y$.

  If $(x,y)\in \caA'$, then we can replace the part of $\caA$ between $(x_\ell,y_\ell)$ and $ (x_r,y_r)$ by the corresponding part in $\caA'$ and obtain an optimal alignment of $X,Y$ that goes through $(x,y)$. 
  Hence, if $(x,y)$ is an edit anchor of $X'$ and $Y'$, then it is an edit anchor of $X$ and $Y$.
  Symmetrically, if $(x,y)\in \caA$, then we can replace the part of $\caA'$ between $(x_\ell,y_\ell)$ and $ (x_r,y_r)$ by the corresponding part in $\caA$ and obtain an optimal alignment of $X',Y'$ that goes through $(x,y)$. 
  Hence, if $(x,y)$ is an edit anchor of $X$ and $Y$, then it is an edit anchor of $X'$ and $Y'$.
\end{proof}

Now we prove \cref{lem:findanchor-algo,lem:keepordiscard-algo}.
\find*
\begin{proof}
   Define fragments $X'=X(i\dd j]$ and $Y'=Y(i\dd j+|Y|-|X|]$ as in \cref{lem:intersect}. 
   By monotonicity of $|\LZ(\cdot)|$ with respect to prefixes, we can compute positions $i$ and $j$ using binary search, with \cref{thm:lz} employed to implement the $|\LZ(\cdot)|\le 6k+2$ test on substrings of $X$ and $\rev{X}$. 
   Overall, this step requires $\Ohtilde(\sqrt{kn})$ query complexity and time complexity.
   By \cref{fct:lz}, we have 
   \[ |\LZ(X')| \le |\LZ(X(i \dd x])| + |\LZ(X(x \dd j])| \le |\LZ(\rev{X(i \dd x]})|\cdot \Oh(\log n) + |\LZ(X(x \dd j])| \le \Oh(k\log n).\]
   Thus, using \cref{thm:lz}, we can compute $\LZ(X')$ in $\Ohtilde(\sqrt{kn})$ query complexity and time complexity.
      If $\ed(X',Y')\le k$, then \cref{fct:lz} yields
   \[ |\LZ(Y')| \le (|\LZ(X')|+k)\cdot \Oh(\log n)  \le \Oh(k\log^2 n).\]
   Consequently, we can use  \cref{thm:lz} in $\Ohtilde(\sqrt{kn})$ query complexity and time complexity
   to compute $\LZ(Y')$ or report that $|\LZ(Y')|$ exceeds the $\Oh(k\log^2 n)$ threshold derived above.
   In the latter case, we conclude that $\ed(X,Y)\ge \ed(X',Y') > k$, so $(x,0)$ trivially satisfies the definition of a $k$-edit anchor.
   Otherwise, we use \cref{thm:lzed} to check whether $\ed(X',Y')\le k$ and, if so, retrieve an optimal sequence of edits transforming $X'$ into $Y'$. 
   Since we already know the LZ-factorizations of $X'$ and $Y'$, this step takes $\Ohtilde(k^2)$ additional time complexity and zero query complexity.
   If $\ed(X',Y') > k$, then we return $0$ again. 
   Otherwise, our algorithm scans the list of edits transforming $X'$ into $Y'$
   to derive and return $y\in [0\dd |Y|]$ such that $(x,y)$ belongs to the underlying optimal alignment $\caA' : X'\onto Y'$.
   By \cref{lem:intersect}, if $\ed(X,Y)\le k$,  then $(x,y)$ must be an edit anchor of $X$ and $Y$.
\end{proof}

\keepordiscard*
\begin{proof}
The proof is similar to that of \cref{lem:findanchor-algo}.
First, we find the positions $i,j$ defined in \cref{lem:intersect}.
Next, we retrieve $\LZ(X(i\dd x]),\LZ(X(x\dd j])$, and $\LZ(X')=\LZ(X(i\dd j])$,
as well as $\LZ(Y(i\dd y]), \LZ(Y(y\dd j+|Y|-|X|])$, and $\LZ(Y')=\LZ(Y(i\dd j+|Y|-|X|])$.
If $\ed(X',Y')\le k$, then the sizes of all these LZ factorizations are in $\Oh(k\log^2 n)$.
Consequently, we can use \cref{thm:lz} in $\Ohtilde(\sqrt{kn})$ query complexity and time complexity
to either compute all these LZ factorizations or conclude that $\LZ(X',Y')>k$ (in that case, we return \KwSty{false}).
If $\ed(X',Y')\le k$, then $(x,y)$ is an edit anchor for $(X',Y')$ if and only if $\ed(X(i\dd x],Y(i\dd y])+\ed(X(x\dd j],Y(y\dd j+|Y|-|X|])=\ed(X',Y')\le k$, and our goal is to return \KwSty{true} if and only if this condition holds.
Consequently, we apply \cref{thm:lzed} in $\Ohtilde(k^2)$ additional time complexity (and zero query complexity)
to evaluate the three edit distances involved in our test or discover that some of these distances exceed $k$.

It remains to prove the correctness of our algorithm.
Either output is valid if $\ed(X,Y) > k$. Hence, we assume $\ed(X,Y) \le k$ in the following.
In this case, we must have $\ed(X',Y') \le k$ by monotonicity of edit distance. 
Moreover, by \cref{lem:intersect}, $(x,y)$ is an edit anchor of $X,Y$ if and only if it is an edit anchor of $X'$ and $Y'$. Thus, the algorithm correctly decides whether $(x,y)$ is an edit anchor of $X$ and $Y$.
\end{proof}
\section{Quantum Algorithms for Lempel--Ziv Factorization}\label{sec:lz}

\subsection{Algorithms with Near-Optimal Query Complexity}

This section provides two preliminary solutions with optimal and near-optimal query times. The first has optimal $\Oh(\sqrt{zn})$ query complexity but requires exponential time. The second has a near-optimal $\Ohtilde(\sqrt{zn})$ query complexity but requires $\Ohtilde(n)$ time. The second algorithm introduces ideas that will be expanded on in Section \ref{sec:sublinear_time_alg} for our main algorithm
with $\Ohtilde(\sqrt{zn})$ query and time complexity. 

\subsubsection{Achieving Optimal-Query Complexity in Exponential Time}
\label{sec:oracle_id_algorithm}

A naive approach is to first obtain the input string from the oracle (in the worst case using $n$ oracle queries).
Then, any compressed representation can be computed without further input queries. The first approach discussed here shows how to find the input string using fewer queries, specifically $\Oh(\sqrt{zn})$ queries for binary strings. We will prove this query complexity is optimal in Section \ref{sec:lb}. This algorithm is based on a solution for the problem of identifying an oracle (in our case, an input string) in the minimum number of oracle queries by Kothari~\cite{DBLP:conf/stacs/Kothari14}. 
Kothari's solution builds on a previous `halving' algorithm by Littlestone~\cite{DBLP:journals/ml/Littlestone87}. 

We next describe the basic halving algorithm as applied to our problem.
Assuming that $z$ is known, we enumerate all binary strings of length $n$ with at most $z$ LZ77 factors. Call this set~$\mathcal{S}$. Since an encoding with $z$ factors requires at most $2z\log n$ bits, there are at most $\sum_{i=0}^{2z\log n} 2^i =  2^{2z\log n + 1} - 1 = 2n^{2z} - 1$ such strings in $\mathcal{S}$.
We construct a string $M$ of length $n$ from $\mathcal{S}$, where $M[i] = 1$ if at least half of strings in $\mathcal{S}$ are $1$ at the $i^{th}$ position, and $M[i] = 0$ otherwise.
Note that the construction of $M$ requires time exponential in $z$ but does not require any oracle queries. Grover's search is then used to find a mismatch if one exists between $M$ and the oracle string with $\Oh(\sqrt{n})$ queries. If a mismatch occurs at position $i$, we can then eliminate at least half of the potential strings in $\mathcal{S}$. We repeat this process until no mismatches are found, at which point we have completely recovered the oracle (input string). Known algorithms can then obtain all compressed forms of text. 

Naively applying this approach would result in an algorithm with $\Oh(\sqrt{n} \log |\mathcal{S}|) = \Oh(z \sqrt{n}\log n)$ query complexity. 
Kothari's improvements on this basic halving algorithm give us a quantum algorithm that uses $\Oh(\sqrt{n \log |\mathcal{S}| / \log n}) = \Oh(\sqrt{zn})$ input queries.
We can avoid assuming the knowledge of $z$ by progressively trying different powers of $2$ as our guess of $z$, still resulting in $\sum_{i = 0}^{\log z} \Oh(\sqrt{2^in}) = \Oh(\sqrt{zn})$ queries overall. 
As noted above, this approach is not time-efficient.

\subsubsection{Achieving Near-Optimal Query Complexity in Near-Linear Time}
\label{sec:linear_time}

An algorithm with a similar query complexity and far improved time complexity is possible by using a more specialized approach.
Specifically, one can obtain the non-overlapping LZ77 factorization. 
For non-overlapping LZ77, every factor, say $X[s_i\dd s_i+\ell_i)$, that is not a new symbol must reference a previous occurrence completely contained in $X[1\dd s_i)$. 
This only increases the size of this factorization by at most a logarithmic factor. 
That is, if  $z_{no}$ is the number of factors for the non-overlapping LZ77 factorization, then $z \leq z_{no}\le \Oh(z\log n)$~\cite{DBLP:journals/csur/Navarro21a}.
This factorization can be converted into other compressed forms in near-linear time, as described in Section \ref{sec:other_encodings}.

We obtain the factorization by processing $X$ from left to right as follows: Suppose inductively that we have determined the factors $F_1, F_2,\ldots, F_{i-1}$, and we want to obtain the $i^{th}$ factor.
Let $s_i$ denote the starting index of the $i^{th}$ factor and $\ell_i$ its length.
Assume that we have the prefixes $X[1\dd k]$, for $k \in [1\dd  s_i)$, sorted in co-lexicographic order. 
To find the next factor $X[s_i\dd  j]$, we apply exponential search\footnote{Recall that exponential search checks ascending powers of $2$ until an interval $[2^{x-1}, 2^x]$ for some $x \geq 1$ containing the solution is found, at which point binary search is applied to the interval.} on $j$. 
To evaluate a given $j$ we use binary search on the sorted set of prefixes. 
To compare a prefix $X[1\dd k]$, we find the rightmost mismatch of the substrings $X[s_i\dd  j]$ and $X[k-(j-s_i)\dd k]$. 
If no rightmost mismatch is found, then $X[s_i\dd  j]$ has occurred previously as a substring, and we continue the exponential search on $j$.
Otherwise, we compare the symbol at the rightmost mismatch to identify which half of the sorted set of prefixes to continue the binary search. If $\ell_i$ is the length of $i^{th}$ factor found, this requires $\Oh(\log^2 n \cdot \sqrt{\ell_i})$ queries and time. 

To proceed to the $(i+1)^{th}$ factor, we now must obtain the co-lexicographically sorted order of the $\ell_i$ new prefixes. 
This can be done using a standard linear-time suffix tree construction algorithm.
Specifically, if we consider the suffix tree of the reversed text $\rev{X}$, we are prepending $\ell_i$ symbols to a suffix of $\rev{X}$. These are accessed from either the oracle directly only in the case the new factor is a new symbol, and otherwise from the previously obtained string. Since we are prepending to $\rev{X}$, a right-to-left suffix construction algorithm such as McCreight's~\cite{DBLP:journals/jacm/McCreight76} can be used. 

The query complexity is
$\sum_{i=1}^{z_{no}} \sqrt{\ell_i} \cdot \log^2 n = \Ohtilde(\sum_{i=1}^{z_{no}} \sqrt{\ell_i}).
$ At the same time, we have $\sum_{i=1}^{z_{no}} \ell_i = n$, so the sum is maximized when each $\ell_i = \frac{n}{z_{no}}$ making $\sum_{i=1}^{z_{no}} \sqrt{\ell_i} \leq \sqrt{z_{no} n}$.
Hence, the query complexity is $\Ohtilde(\sqrt{zn})$. The time complexity is $\Oh(n + \log^3 n \cdot \sqrt{z_{no}n})$, which is $\Ohtilde(n)$. 
We will focus for the rest of this section on developing these ideas and utilizing more complex data structures to obtain a sublinear-time algorithm.

\subsection{Main Algorithm: Optimal Query and Time Complexity}
\label{sec:sublinear_time_alg}

On a high level, the algorithm will proceed very much like the near-linear-time algorithm from Section~\ref{sec:linear_time}.
It proceeds from left to right finding the next factor and utilizes a co-lexicographically sorted set of prefixes of $X$.
After the next factor is found, a set of new prefixes of $X$ is added to this sorted set. However, we face two major obstacles: (i) we cannot afford to explicitly maintain a sorted order of all prefixes needed to check all possible previous substrings efficiently; (ii) if we utilize a factorization other than LZ77, like LZ77-End, where fewer potential positions have to be checked, then the monotonicity of being a next factor is lost, i.e., for LZ77-End, $X[s_i\dd j]$ may have occurred as a substring ending at a previous factor, but $X[s_i\dd  j)$ may not have occurred as a substring ending at a previous factor.

To overcome these problems, we introduce a new factorization scheme that extends the LZ-End factorization scheme discussed in Section \ref{sec:prelim}. It allows for more potential places ending locations for each new factor obtained by the algorithm.

\subsubsection{LZ-End+\texorpdfstring{$\tau$}{tau} Factorization} 
Let $\tau \geq 1$ be an integer parameter. 
The LZ-End+$\tau$ factorization of the string $X$ is constructed from left to right.
Initially, $i \gets 1$. 
For $i \geq 1$, if $X[i]$ does not occur in $X[1\dd i)$, then we make $X[i]$ a new factor and set $i \gets i+1$.
Otherwise, let $j$ be the largest index such that $X[i\dd j]$ has an occurrence ending at either the last position of an earlier factor or at a position $k < i$ such that $k \equiv 1 \pmod{\tau}$.
Let $z_{e+\tau}$ denote the number of factors created by the LZ-End+$\tau$ factorization.

Note that there exist strings where $z_e < z_{e+\tau}$.
The smallest binary string example where this is true is $00010011011$, which has an LZ-End factorization with seven factors $0$, $0$, $0$, $1$, $001$, $1$, $011$ 
and an LZ-End+$\tau$ for $\tau = 2$ with eight factors $0$, $0$, $0$, $1$, $001$, $10$, $1$, $1$. 
Loosely speaking, the LZ77-End+$\tau$ algorithm can be `tricked' into taking a longer factor earlier on; in this case, the factor `$10$' which is possible for LZ77-End+$\tau$ but not LZ77-End, and limits future choices.
Fortunately, the same bounds in terms of $z$ established by Kempa and Saha~\cite{DBLP:conf/soda/KempaS22} for $z_e$ also hold for $z_{e+\tau}$.

\begin{lemma}
\label{lem:LZ-End-tau}
Let $z_{e+\tau}$ (resp., $z$) denote the number of factors in the LZ-End+$\tau$ (resp., LZ77) factorization of a given text $X[1\dd n]$. 
Then, $z_{e+\tau} = \Oh(z\log^2 n)$.
\end{lemma}

\begin{proof}
We outline Kempa and Saha's proof of the bound for LZ-End and why it continues to hold for LZ-End+$\tau$. We refer the reader to \cite{DBLP:conf/soda/KempaS22} for more details. In the proof, a factor is considered \emph{special} if its length is at least half the length of the previous factor. Every special factor is assigned a set of substrings of $X$. In particular, if one of these special factors is of length $\ell$, it is assigned $\ell$ substrings of length $2^k$ for $k \in [0\dd \lceil \log n \rceil + 4]$. The bound then follows by showing that: (i) each distinct substring of length $2^k$ is assigned to at most two factors, and (ii) all substrings assigned to a factor are distinct.
Both (i) and (ii) use a proof by contradiction and work because a longer factor is possible, i.e., a longer substring occurs ending at a factor end. 

For (i), if a substring is assigned to three or more factors, with $X[s_i\dd s_i+\ell_i)$ being the leftmost and $X[s_j\dd s_j+\ell_j)$ the rightmost, then it is shown that, for some $\delta > 0$, there exists a substring $X[s_{j-1}\dd s_{j-1}+\ell_{j-1} + \delta)$ that also ends at a previous factor, contradicting that $X[s_{j-1} \dd s_{j-1}+\ell_{j-1})$ was chosen as a factor. 
This argument is based on the lengths of the substrings and factor $X[s_j\dd \allowbreak {s_j+\ell_j})$ being special. 
These properties continue to hold for LZ-End+$\tau$. Moreover, because our LZ-End+$\tau$ is also greedy and takes the largest factor possible, it could have used $X[s_{j-1}\dd s_{j-1}+\ell_{j-1} + \delta)$ as a factor instead of $X[s_{j-1}\dd s_{j-1}+\ell_{j-1})$.
Hence, we arrive at the same contradiction. 

For (ii), the contradiction is achieved by showing that, if some substring is assigned to $X[s_j\dd \allowbreak {s_j+\ell_j})$ two or more times, then an instance of $X[s_j\dd s_j+\ell_j + \delta)$ for some $\delta > 0$ exists ending at a previous factor. 
This argument is based on the lengths of the substrings and the repeated substring causing periodicity. It continues to hold for LZ-End+$\tau$. Again, because LZ-End+$\tau$ is also greedy and could use $X[s_j\dd s_j+\ell_j + \delta)$ as a factor instead of $X[s_{j}\dd s_{j}+\ell_{j})$, we arrive at the same contradiction.
\end{proof}

Next, we describe how new LZ77-End+$\tau$ factors of $X$ are obtained by using the concept of the \emph{$\tau$-far property} and a dynamic longest common extension (LCE) data structure. Following this, we describe how the co-lexicographically sorted prefixes required by the algorithm are maintained.

\subsubsection{Maintaining the Colexicographic Ordering of Prefixes}

\begin{figure}
    \centering
    \includegraphics[width=.5\textwidth]{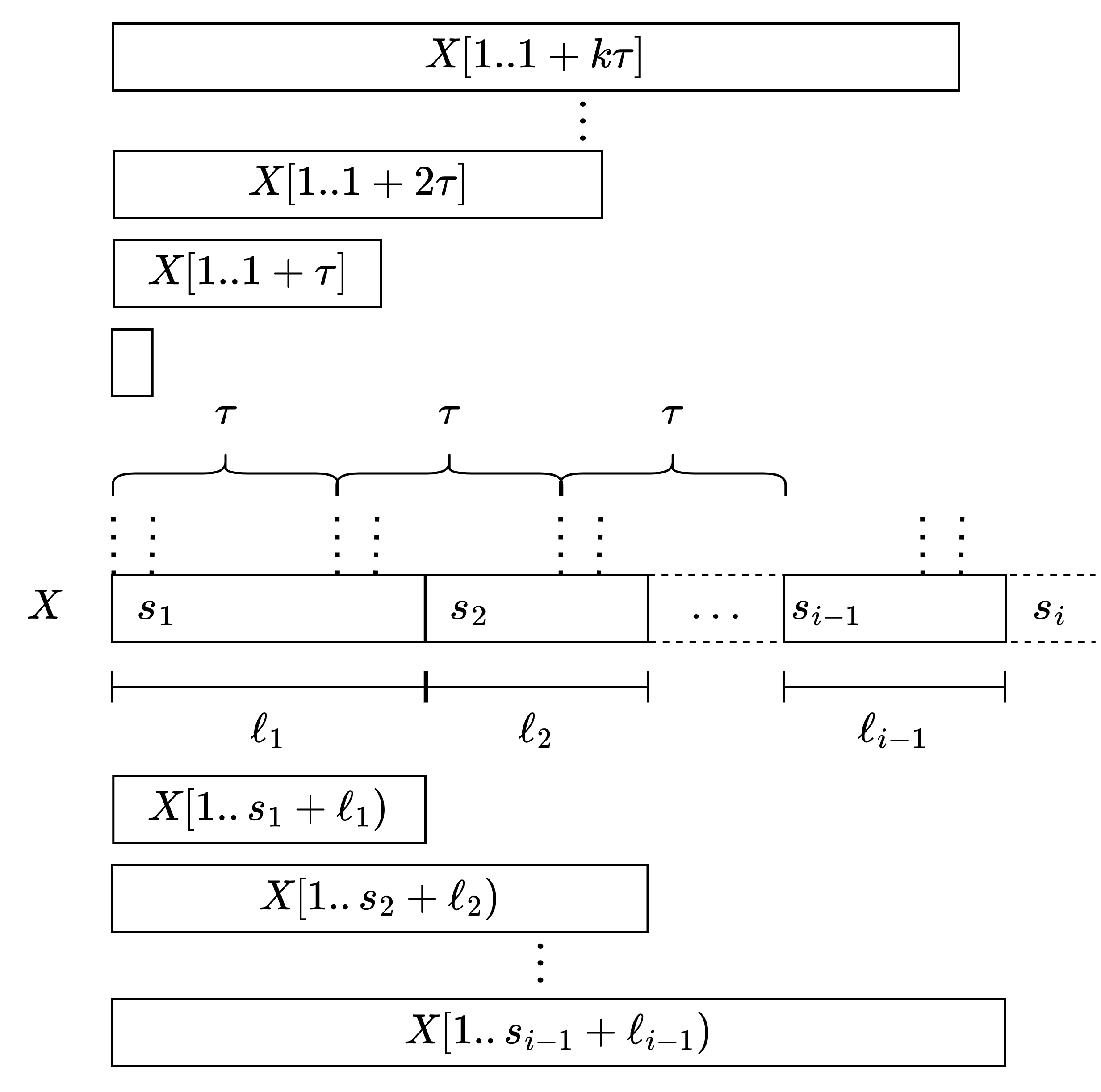}
    \caption{The colexicographic order of the prefixes  $X[1\dd s_1 + \ell_1 )$, $X[1\dd s_2 + \ell_2)$, $\hdots$, $X[1\dd s_{i-1} + \ell_{i-1})$ (shown below $X$) and $X[1\dd 1],X[1\dd \tau+1],\ldots, X[1\dd  k\tau+1]$ where $k$ is the largest natural number such that $1 + k\tau  < s_i$ (shown above $X$) are known prior to iteration $i$.}
    \label{fig:factors}
\end{figure}

The first factor $F_1$ is always $X[1]$. Assume inductively that the factors $F_1,F_2,\ldots,F_{i-1}$ have already been determined. 
Recall that, for factors $F_j$ of the form $(\alpha)$ with $\alpha \in \Sigma$, we also store $(s_j, 1)$, where $s_j$ is the starting position of the $j^{th}$ factor in $X$.
We assume inductively that we have the colexicographically sorted order of prefixes of 
\begin{align*}
\mathcal{P}_{i-1} &\coloneqq \{X[1\dd  s_j + \ell_j - 1] \mid  (s_j, \ell_j) = F_j, 1 \leq j \leq i-1\}\\ 
&\cup \{X[1\dd  j] \mid 1\leq j \leq s_{i-1} + \ell_{i-1}-1,~~ j \equiv 1 \pmod{\tau}\}.
\end{align*}
See Figure \ref{fig:factors} for an illustration of the prefixes contained in $\mathcal{P}_i$. For each of them, we store the ending position of the prefix.

The following section shows how to obtain the factor $F_i$. 
For now, suppose we just determined the $i^{th}$ factor starting position $s_i$.
After the factor length $\ell_i$ is found, we need to determine where to insert the prefixes $X[1\dd  s_i + \ell_i - 1]$ and $X[1\dd j]$ for $j \in [s_i\dd  s_i + \ell_1 - 1]$ such that $j \equiv 1 \pmod{\tau}$, in the colexicographically sorted order of $\mathcal{P}_{i-1}$  to create $\mathcal{P}_i$. To do this, we use the dynamic longest common extension (LCE) data structure of Nishimoto et al.~\cite{DBLP:conf/mfcs/NishimotoIIBT16} (see Lemma \ref{lem:dynamic_lce}).

\begin{lemma}[Dynamic LCE data structure~\cite{DBLP:conf/mfcs/NishimotoIIBT16}]
\label{lem:dynamic_lce}
An LCE query on a text $S[1\dd m]$ consists of two indices $i$ and $j$ and returns the largest $\ell$ such that $S[i\dd i+\ell) = S[j\dd j+\ell)$. 
There exists a data structure that requires $\Oh(m)$ time to construct, supports LCE queries in $\Ohtilde(1)$ time, and supports insertion of either a substring of $S$ or a single character into $S$ at an arbitrary position in $\Ohtilde(1)$ time\footnote{Polylogarithmic factors here are with respect to the final string length after all insertions.}. 

\end{lemma}
The main idea is to use the above dynamic LCE structure over the reverse of the prefix of $X$ found thus far. 
We initialize the dynamic LCE data structure with the first LZ-End+$\tau$ factor of~$X$, which is a single character. 
For every factor found after that, we prepend the reversed factor to the current reversed prefix and update the data structure, all in $\Ohtilde(1)$ time. In particular, if the $i^{th}$ factor of $X$ found is a new character, we prepend that character to our dynamic LCE structure for $S \coloneqq \rev{X[1\dd s_i)}$.
If the $i^{th}$ factor found is $X[s_i\dd  s_i + \ell_i - 1] = X[x\dd y]$, for $x,y \in [1\dd  s_{i})$, then we prepend the substring $\rev{X[x\dd y]} = S[s_i-y\dd  s_i-x]$
to string representation of our dynamic LCE structure.
Once the reversed $i^{th}$ factor is prepended to the reversed prefix in the dynamic LCE structure, to compare the colexicographic order of the new prefixes in $\mathcal{P}_i$, we find the LCE of the two reversed prefixes being compared and compare the symbol in the position after their furthest match.
Applying this comparison technique and binary search on $\mathcal{P}_{i-1}$, we determine where each prefix in $\mathcal{P}_i \setminus \mathcal{P}_{i-1}$ should be inserted in the sorted order in polylogarithmic time.

\subsubsection{Finding the Next LZ-End+\texorpdfstring{$\tau$}{tau} Factor}
\label{sec:next_factor}

We now show how to obtain the new factor $F_i = (s_i,\ell_i)$.
Firstly, $s_i = s_{i-1} + \ell_{i-1}$. 
We say $X[s_i\dd h]$ is a \emph{potential factor} if either $h=s_i$ and $X[s_i\dd h] \in \Sigma $ is the leftmost occurrence of a symbol in $X$ or $X[s_i\dd h] = X[x\dd y]$, where $y < s_i$ and $y$ is the end of a previous factor or $y \equiv 1 \pmod{\tau}$.
We say the \emph{$\tau$-far property} holds for an index $j \geq s_i$ if there exists $h$ such that $j-\tau \leq h \leq j$ and $X[s_i\dd h]$ is a potential factor.

\begin{figure}
    \centering
    \begin{minipage}{.6\textwidth}
    \includegraphics[width=\textwidth]{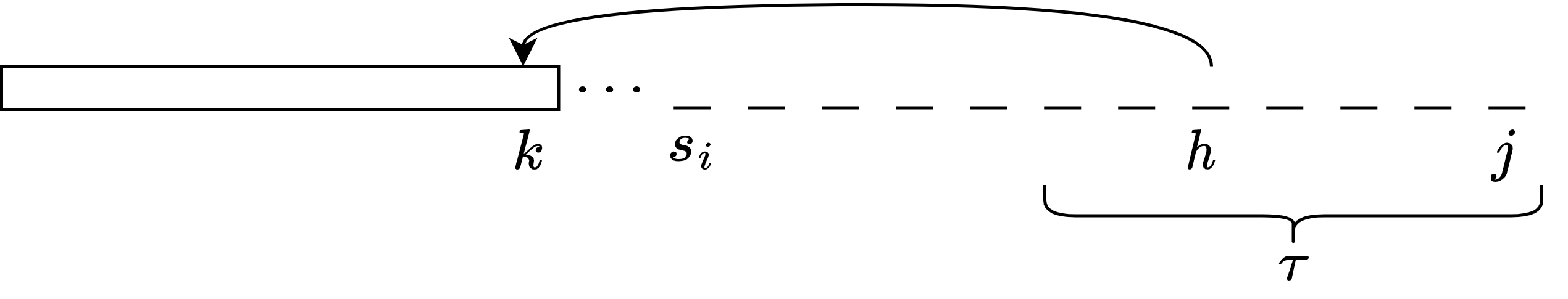}
    \end{minipage}
    
    \vspace{1em}
    \begin{minipage}{.75\textwidth}
    \includegraphics[width=\textwidth]{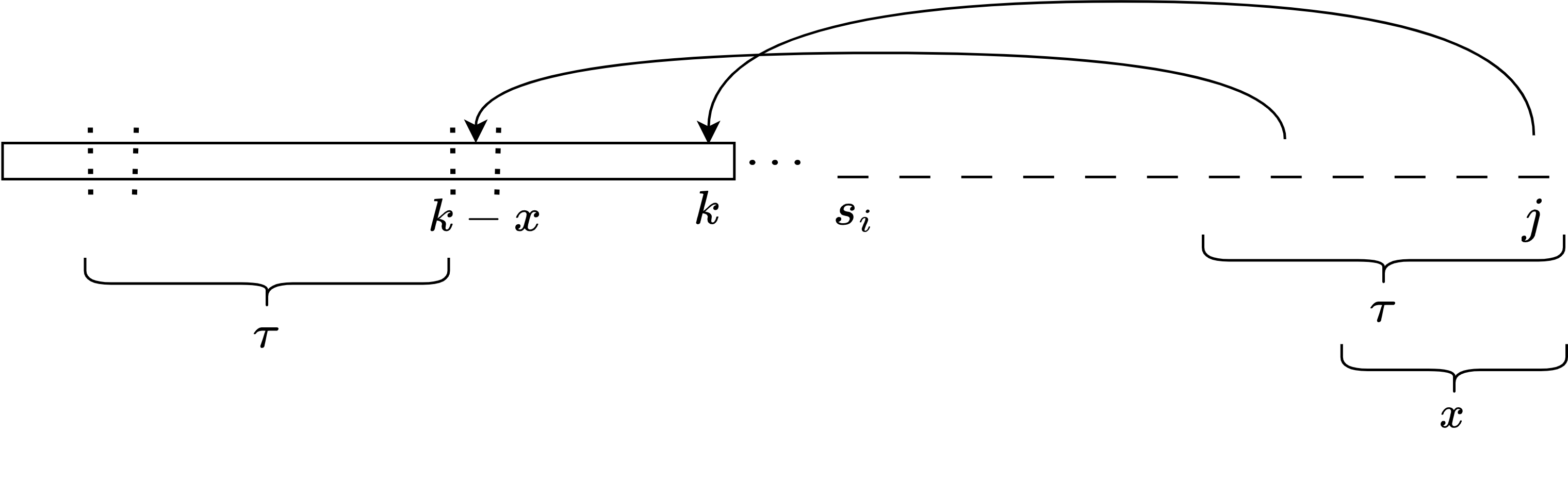}
    \end{minipage}
    
    \caption{The two cases given the proof of Lemma \ref{lem:tau_far_property}. 
    On the top is Case 1, where $X[s_i\dd  j]$ is not a potential factor.
    On the bottom is Case 2, where $X[s_i\dd  j]$ is a potential factor. Here, we are implying $k-x \equiv 1 \pmod{\tau}$.}
    \label{fig:monotonicity_proof}
\end{figure}

\begin{lemma}[Monotonicity of $\tau$-far property]
\label{lem:tau_far_property}
When finding a new factor starting at position $s_i$, if the $\tau$-far property holds for $j > s_i$, then it holds for $j-1$.
\end{lemma}

\begin{proof}
There are two cases; see Figure \ref{fig:monotonicity_proof}.
Case 1: If $X[s_i\dd  j]$ is not a factor, since the $\tau$-far property holds for $j$, there exists an $h\in [j-\tau\dd  j)$ and $k < s_i$ such that $X[s_i\dd h] = X[k-(h-s_i)\dd  k]$ is a potential factor. 
Then, this $h$ demonstrates that the $\tau$-far property holds for $j-1$. 
Case 2: Suppose instead that $X[s_i\dd  j]$ is a potential factor and matches some $X[k-(j-s_i)\dd k]$, where $k$ is the last position in a previous factor or $k < s_i$ and $k \equiv 1 \pmod{\tau}$. 
If $j-s_i+1 > \tau$, then there exists some $x \in [1\dd  \tau]$ such that $k-(j-s_i) \leq k-x < k$ and $k-x \equiv 1 \pmod{\tau}$; hence, $X[s_i\dd  j-x] = X[k-(j-s_i)\dd  k-x]$, making $X[s_i\dd  j-x]$ a potential factor.
Since $j-1 - \tau \leq j-x \leq j-1$, the $\tau$-far property holds for $j-1$. If instead $j-s_i+1 \leq \tau$, then $j-\tau \leq s_i \leq j$ and $X[s_i\dd s_i]$ is always a potential factor since either it is the first occurrence of a symbol or we can refer to the factor created by the first occurrence of $X[s_i]$. This proves that the property still holds for $j-1$.
\end{proof}

By Lemma \ref{lem:tau_far_property}, monotonicity holds for the $\tau$-far property when trying to find the next factor starting at position $s_i$.
Thus, to find the largest $j$ such that the $\tau$-far property holds, we can now use exponential search. 
At its core, we need to determine whether the $\tau$-far property holds for a given $j > s_i$. Once this largest $j$ is determined, the largest $h \in [\max(s_i, j-\tau) \dd j]$ such that $X[s_i, h]$ is a potential factor must be determined as well.

We show a progression of algorithms to accomplish the above task. 
Firstly, we make some straightforward, yet crucial, observations.
Let $S$ be any string. 
Since $\mathcal{P}_{i-1}$ is colexicographically sorted, all prefixes that have the same string (say $S$) as a suffix can be represented as a range of indices. 
 This range is empty when $S$ is not a suffix of any prefix in $\mathcal{P}_{i-1}$.
Moreover, this range can always be computed in $\Ohtilde(|S|)$ time using binary search.  However, if $S$ has an occurrence within the  $X[1\dd s_{i})$ (i.e., the prefix seen thus far) and is specified by the start and end position of that occurrence, we can use LCE queries and improve the time for finding the range to $\Ohtilde(1)$.

\paragraph{Next factor in $\Ohtilde(\tau + \ell_i)$ time:}
\label{sec:linear_sol_LCE}
For $h \in [\max(s_i, j-\tau)\dd  j]$, let $k_h \in [0\dd  h-s_i+1]$ be the largest value such that $X[h-k_h+1\dd  h]$ is a suffix of a prefix in $\mathcal{P}_{i-1}$. 
We initialize $h = j$.
Since the prefixes in $\mathcal{P}_{i-1}$ are co-lexicographically sorted, we can find  $k_h$ in $\Ohtilde(k_h)$ time by using binary search on $\mathcal{P}_{i-1}$.
To do so, symbols are prepended one by one
and binary search is used to check if the corresponding sorted index range of $\mathcal{P}_{i-1}$ is non-empty.

Next, we compute $k_h$ for $h \in [\max(s_i, j-\tau)\dd j)$ in the descending order $h$. 
We keep track of $h' := \argmin_{y \in [h+1 \dd  j]} (y-k_y)$.
If $h'-k_{h'} +1 \leq h$, then $X[h'-k_{h'} + 1 \dd h]$ has an occurrence in $X[1\dd s_i)$, and we can now use LCE queries to determine the range of $X[h'-k_{h'} + 1 \dd h]$ in $\mathcal{P}_{i-1}$. 
If this range is empty, we conclude that $k_h < k_{h'}- (h'-h)$ and LCE can be used to find $k_h$ in $\Ohtilde(1)$ time.
Otherwise, we proceed by prepending symbols one by one until $k_h$ is found.

The time per $h \in [\max(s_i, j-\tau) \dd  j]$ is $\Ohtilde(1)$ for LCE queries, in addition to $\Ohtilde(x_h)$ where $x_h$ is the number of symbols we prepended for $h$. Since we always use the smallest $h'-k_{h'}$ value seen thus far, $\sum_{h = j-\tau}^j x_h \leq j-s_i = \Oh(\ell_i)$.
This makes it so checking if the $\tau$-far property holds for $j$ takes $\Ohtilde(\ell_i)$ time. The algorithm also identifies the rightmost $h \in [\max(s_i, j-\tau) \dd j]$ such that $X[s_i\dd  h]$ is a potential factor (if one exists). This only provides at best a near-linear time algorithm.

\paragraph{Next factor in $\Ohtilde(\sqrt{\tau\ell_i})$ time:}

Instead of prepending characters individually and using binary search after exhausting the reach of the LCE queries, we can instead find the rightmost mismatch and then use binary search on $\mathcal{P}_{i-1}$. 
Specifically, suppose that, for a given $h \in [\max(s_i,j-\tau)\dd j]$, we apply the LCE query and identify a non-empty range of prefixes in $\mathcal{P}_{i-1}$ with $X[w\dd h]$ as a suffix.
On this set of prefixes, we continue the search from $w-1$ downward using exponential search and identifying whether a mismatch occurs with the right-most mismatch algorithm. 

For $h \in [j-\tau\dd  j]$, let $x_h$ now be the number of characters searched using exponential search and the right-most mismatch algorithm. As before $\sum_{h=j-\tau}^j x_h = \Ohtilde(j-s_i)$. The total time required for this is logarithmic factors from $\sum_{h=j-\tau}^j \sqrt{x_j} = \Ohtilde(\sqrt{\tau \ell_i})$. This will give us a sub-linear time algorithm if we choose $\tau$ appropriately; however, it will not be sufficient to obtain our goal.

\paragraph{Next factor in $\Ohtilde(\tau + \sqrt{\ell_i})$ time:}
Here we do not apply the rightmost-mismatch search for every $h \in [\max(s_i,j-\tau)\dd j]$.
Instead, for each $h$, we identify a set of prefixes in $\mathcal{P}_{i-1}$ such that $X[s_i\dd  h]$ shares a suffix of length at least $d_h = h-\max(s_i, j-\tau)+1$. This set is represented by the range of indices, $[s_h, e_h]$, in the sorted $\mathcal{P}_{i-1}$ corresponding to prefixes sharing this suffix of length $d_h = h-\max(s_i, j-\tau)+1$. By using the same LCE technique and prepending and stopping at index $\max(s_i, j-\tau)$, this can be accomplished in $\Ohtilde(\tau)$ time. After this, we have a set of ranges in $\mathcal{P}_{i-1}$. Note that for a given $h$, 
if we delete the last $d_h$ characters in each
prefix represented in $[s_h\dd e_h]$, they remain co-lexicographically sorted. 
We want to search for $X[s_i\dd j-\tau)$ as a suffix on these ranges, each with their appropriate suffix removed. 
Since the rightmost-mismatch algorithm is costly, we can first merge these ranges (each with their appropriate suffix removed), and then use binary search.
However, merging these sorted ranges would be too costly.
Instead, we can take advantage of the following lemma to avoid this cost.

\begin{lemma}[\cite{shiwangshiwang}]
\label{lem:k_stat}
Given $\tau$ sorted arrays $A_1,\ldots, A_\tau$ of $n$ elements in total,
the $x^{th}$ largest element in the array formed by merging them can be found using $\Oh(\tau \log \tau \cdot \log (n/\tau))$ comparisons.
\end{lemma}

Using the LCE data structure to compare any to prefixes, the $x^{th}$ largest element in the merged array can be found in $\Ohtilde(\tau)$ time. 
Using Lemma \ref{lem:k_stat}, we can find whichever rank prefix in the subset of $\mathcal{P}_{i-1}$ we are concerned with, then find the rightmost mismatch and compare it to $X[s_i\dd j-\tau)$. Doing so, the total time needed for obtaining the next factor is $\Ohtilde(\tau + \sqrt{\ell_i})$.

\subsubsection{Time and Query Complexity}

Taken over the entire string, the time complexity of finding the factors and updating the sorted order of the newly added prefixes is up to logarithmic factors bound by
\[
  \tfrac{n}{\tau} + z_{e + \tau} + \sum_{i=1}^{z_{e+\tau}}(\tau + \sqrt{\ell_i})
\leq 
 \tfrac{n}{\tau} + z_{e+\tau}+ \tau z_{e+\tau} + \sqrt{z_{e+\tau} n}  = \Oh(\tfrac{n}{\tau}+\tau z_{e+\tau} + \sqrt{z_{e+\tau}n} ),
\]
where the inequality follows from $\sum \ell_i = n$.  Combined with Lemma \ref{lem:LZ-End-tau}, which bounds $z_{e+\tau}$ to be logarithmic factors from $z$, and a logarithmic number of repetitions of each call to Grover's algorithm, the total time complexity is $\Ohtilde(\sqrt{z n} + \tau z + \frac{n}{\tau})$. 
To minimize the time complexity, we should set $\tau = \sqrt{n/z}$, bringing the total time to the desired $\Ohtilde(\sqrt{z n})$.

Note that we do not know $z$ in advance to set $\tau$. However, the desired time complexity can be obtained by increasing our guess of $z$ as follows: 
Let $z_{guess}$ be initially $1$,  set $\tau_{guess} = \lceil \sqrt{n/z_{guess}}\rceil$, and run the above algorithm until either the entire factorization of the string $X$ is obtained or the number of factors encountered is greater than $z_{guess}$. 
For a given $z_{guess}$, the time complexity is bound by $\Ohtilde(\sqrt{z_{guess} n} + \tau_{guess}z_{guess} +  \frac{n}{\tau_{guess}})$, which is $\Ohtilde(\sqrt{z_{guess} n})$.
If a complete factorization of $X$ is not obtained, we set $z_{guess} \gets 2 \cdot z_{guess}$, similarly update $\tau_{guess}$, and repeat our algorithm for the new $\tau_{guess}$. 
The total time taken over all guesses is logarithmic factors from $\sqrt{n}\sum_{i=1}^{\lceil \log z \rceil} (2^i)^{\frac{1}{2}}$ which, again, is $\Ohtilde(\sqrt{zn})$. 

The following lemma summarizes our result on LZ-End+$\tau$ factorization.
\begin{lemma}
Given a text $X$ of length $n$ having $z$ LZ77 factors, there exists a quantum algorithm that 
obtains the LZ-End+$\tau$ factorization of $X$ in $\Ohtilde(\sqrt{zn})$ time and input queries.
\end{lemma}

\subsection{Obtaining the LZ77, SLP, RL-BWT Encodings}
\label{sec:other_encodings}

To obtain the other compressed encodings, we utilize the following result by Kempa and Kociumaka, stated here as Lemma \ref{lem:ipm_query}. 
We need the following definitions: a factor $F_i = (s_i, \ell_i)$ is called \emph{previous factor} if $X[s_i\dd  s_i + \ell_i) = X[j\dd  j+\ell_i)$ for some $j < s_i$. 
We say a factorization $X = F_1, \hdots, F_f$ of a string is \emph{LZ77-like} if each factor $F_i$ is non-empty and  $|F_i| > 1$ implies $F_i$ is a previous factor. 
Note that LZ-End+$\tau$ is LZ77-like with $f = \Oh(z\log^2 n)$ as shown in Lemma \ref{lem:LZ-End-tau}.

\begin{lemma}[\cite{DBLP:journals/cacm/KempaK22} Thm.~6.11] 
\label{lem:ipm_query}
Given an LZ77-like factorization of a string $X[1\dd n]$ into $f$ factors, we can in $\Oh(f \log^4 n)$ time construct a data structure that, for any pattern $P$ represented by its arbitrary occurrence in $X$, returns the leftmost occurrence of $P$ in $X$ in $\Oh(\log^3 n)$ time.
\end{lemma}

Starting with the LZ-End+$\tau$ factorization obtained in Section \ref{sec:sublinear_time_alg}, we construct the data structure from Lemma \ref{lem:ipm_query}. To obtain the LZ77 factorization, we again work from left to right and apply exponential search to obtain the next factor. In particular, if the start of our $i^{th}$ factor is $s_i$ and $X[s_i\dd s_i+\ell_i)$ if the leftmost occurrence of the substring is at position $j < s_i$, then we continue the search by increasing $\ell$. Since $\Oh(\log^3 n)$ time is used per query, we get that $\Oh(\log^4 n)$ time is used to obtain each new factor. Therefore, once the data structure from Lemma \ref{lem:ipm_query} is constructed, the required time to obtain the LZ77 factorization is $\Oh(z\log^4 n)$. The total time complexity of constructing all LZ77 factorization starting from the oracle for $X$ is 
$
\Ohtilde(\sqrt{z_{e+\tau} n} + z)
$, which is $\Ohtilde(\sqrt{z n})$, as summarized below. 

\lz*

To obtain the RL-BWT of the text we directly apply an algorithm by Kempa and Kociumaka.
In particular, they provide a Las-Vegas randomized algorithm that, given the LZ77 factorization of a text $X$ of length $n$, computes its RL-BWT in $\Oh(z \log^8 n)$ time (see Theorem~5.35 in \cite{DBLP:journals/cacm/KempaK22}).
Combined with $r = \Oh(z\log^2 n)$~\cite{DBLP:journals/cacm/KempaK22}, we obtain the following result.

\rlbwt*

Obtaining a balanced CFG of size $\Ohtilde(z)$ is similarly an application of previous results, and can be derived by using either the original LZ77 to balanced grammar conversion algorithm of Rytter~\cite{DBLP:journals/tcs/Rytter03}, or more recent results for converting LZ77 encodings to grammars~\cite{kempa2021fast}, and even balanced run-length straight-line programs~\cite{DBLP:journals/cacm/KempaK22}. 
\section{Compressed Text Indexing and Applications}
\label{sec:SA_index}

We start this section having obtained the RL-BWT and the $\Ohtilde(z)$-space LCE data structure for the input text. There are two main stages to the remaining algorithm for obtaining a suffix array index. 
The first is to obtain a less efficient index with fast construction in time $\Ohtilde(\sqrt{rn})$, which supports $\SA$ and $\ISA$ queries in time $\Ohtilde(\sqrt{n/r})$.
We accomplish this by applying a form of prefix doubling and alphabet replacement.
These techniques allow us to `shortcut' the LF-mapping described in Section \ref{sec:prelim}. 
Using this shortcutted LF-mapping, we then sample the suffix array values every $\tau$ text indices apart, where $\tau = \sqrt{n/r}$,
similar to the construction of the original FM-index. 
Once this less efficient index construction is complete, we move on to build the other indexes in~\cite{DBLP:journals/jacm/GagieNP20}.

\subsection{Computing LF\texorpdfstring{$^\tau$}{tau} and Suffix Array Samples}
\label{sec:lf_tau}

Recall that the $\LF$-mapping of an index $i$ of the BWT is defined as $\LF[i] = \ISA[\SA[i] - 1]$. 
The RL-BWT can be equipped rank-and-select structures in $\Ohtilde(r)$ time to support computation of the $\LF$-mapping of a given index in $\Ohtilde(1)$ time.

For a given BWT run corresponding to the interval $[s\dd e]$, we have for $i \in [s\dd e)$ that $\LF[i+1] - \LF[i] = 1$, i.e., intervals contained in BWT runs are mapped on to intervals by the LF-mapping. If we applied the LF-mapping again to each $i \in [\LF[s]\dd  \LF[e]]$, the BWT-runs occurring $[\LF[s]\dd  \LF[e]]$ may split the $[\LF[s]\dd  \LF[e]]$ interval. We define the \emph{pull-back} of a mapping $\LF^2$ as $[s\dd e] \mapsto [s_1\dd  e_1], [s_2\dd  e_2],\hdots, [s_k\dd  e_k]$ that satisfies $s_1 = s$, $s_{i+1} = e_i + 1$ for $i \in [1\dd  k)$, $e_k = e$, $\BWT[\LF[s_{i+1}]] \neq \BWT[\LF[e_{i}]]$, and $j \in [s_i\dd  e_i)$, $i\in[1\dd k]$ implies $\BWT[\LF[j]] = \BWT[\LF[j+1]]$.
We also assign to each interval $[s_i\dd  e_i]$ created by the $\LF^2$ pull-back: (i) a string of length two (specifically, $[s_i\dd  e_i]$ is assigned the string $\BWT[\LF[s_i]] \circ \BWT[s_i]$) (ii) the indices that $s_i$ and $e_i$ map to, that is $\LF^2[s_i]$, $\LF^2[e_i]$; see Figure \ref{fig:lf_tau}.

\begin{figure}
\centering
\includegraphics[width=.7\textwidth]{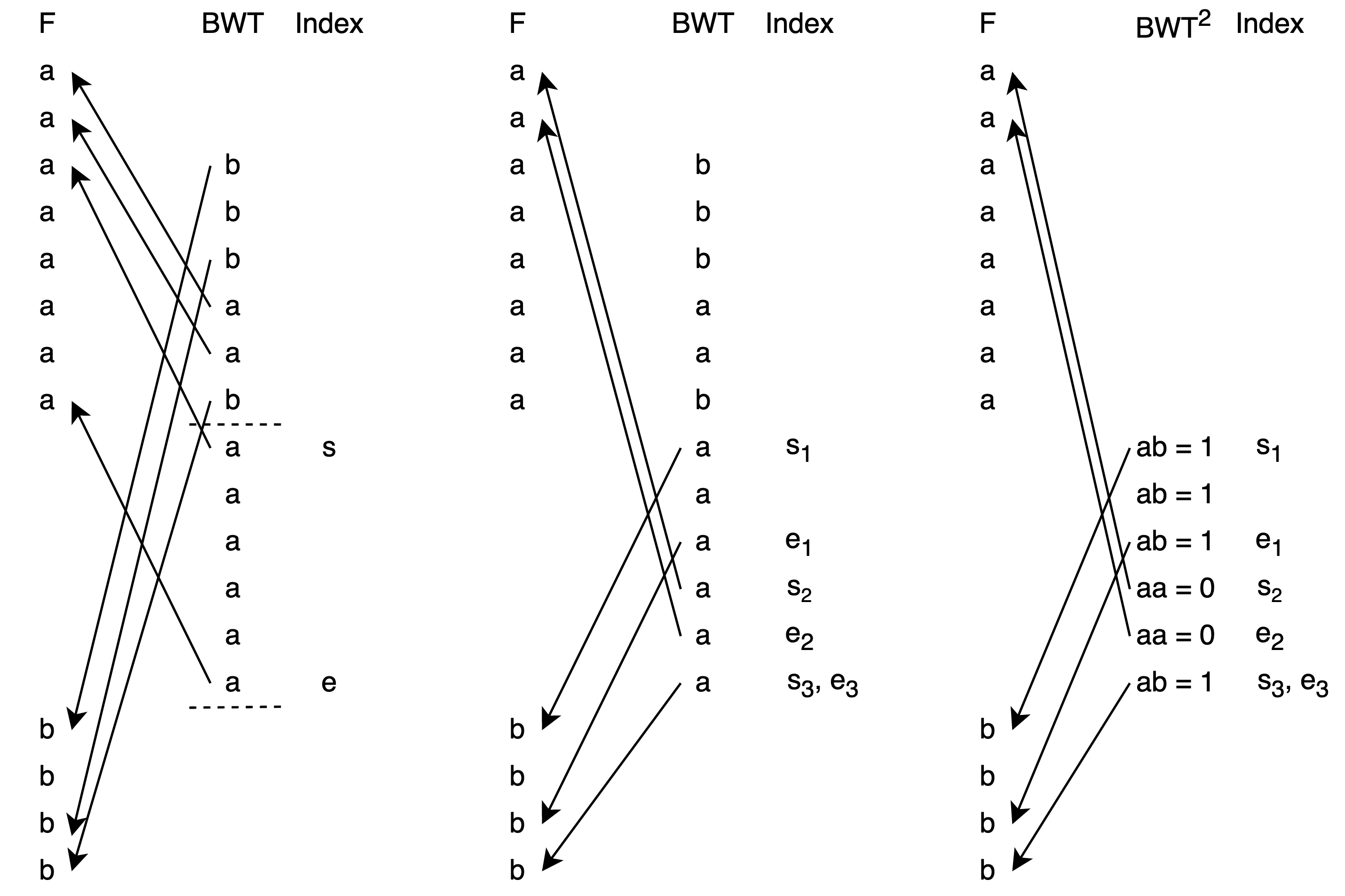}
\caption{
(Left) The initial LF-mapping from the interval $[s\dd e]$. (Middle) The intervals created for $[s\dd e]$ by the $\LF^2$ pull-back. (Right) Alphabet replacement is applied to each interval to create $\BWT^2$. Note that the contents of the F-column are not important and are not used. 
}
\label{fig:lf_tau}
\end{figure}

Observe that the LF-mapping maps distinct BWT-runs onto disjoint intervals.
Hence, each BWT run boundary appears in exactly one interval $[\LF[s]\dd  \LF[e]]$, where $[s\dd e]$ is a BWT run.
As a consequence, if we apply $\LF^2$ pull-back to all BWT run intervals and split each BWT run interval according to its $\LF^2$ pull-back, the number of intervals at most doubles.   

For a given $\tau$ that is a power of two, we next describe how to apply this pull-back technique and alphabet replacement to precompute mappings for each BWT run. 
These precomputed mappings make it so that, given any $i \in [1\dd n]$, we can compute $\LF^\tau[i]$ in $\Ohtilde(1)$ time. The time and space needed to precompute these mappings are $\Ohtilde(\tau r)$. 

We start as above and compute the $\LF^2$ pull-back for every BWT-run interval. We then replace each distinct string of length two with a new symbol. This could be accomplished, for example, by sorting and replacing each string by its rank. However, it should be noted that the order of these new symbols is not important. 
This process assigns each index in $[1\dd n]$ a new symbol. 
We denote this assignment as $\BWT^2$ and observe that the run-length encoded $\BWT^2$ is found by iterating through each $\LF^2$ pull-back. 

We now repeat this entire process for the runs in $\BWT^2[j]$. 
Doing so gives us, for each interval $[s\dd e]$ corresponding to a $\BWT^2$ run, a set of intervals $[s\dd e] \mapsto [s_1\dd  e_1], [s_2\dd  e_2],\hdots, [s_k\dd  e_k]$ that satisfies $s_1 = s$, $s_{i+1} = e_i + 1$ for $i \in [1\dd  k)$, $e_k = e$, $\BWT[\LF^2[s_{i+1}]] \neq \BWT[\LF^2[e_{i}]]$, and $j \in [s_i\dd  e_i)$, $i\in[1\dd k]$ implies $\BWT[\LF^2[j]] = \BWT[\LF^2[j+1]]$. We call this the $\LF^4$ pull-back. Next, we again apply alphabet replacement (that is, replacing a pair of symbols with a single symbol in $[1\dd n]$), defining $\BWT^4$ accordingly. The next iteration will compute the $\LF^8$ pull-back and $\BWT^8$. Repeating $\log \tau$ times, we get a set of intervals corresponding to $\LF^\tau$ pull-back. 

Recall that each pull-back step at most doubles the number of intervals. Hence, by continuing this process $\log \tau$ times, the number of intervals created by corresponding $\LF^\tau$ pull-back is $2^{\log \tau} r = \tau r$. For a given index $i$, to compute $\LF^\tau[i]$ we look at the $\LF^\tau$ pull-back. 
Suppose that $i \in [s'\dd  e']$, where $[s'\dd  e']$ is an interval computed in the $\LF^\tau$ pull-back. 
We look at the mapped onto interval $[\LF^\tau[s']\dd  \LF^\tau[e']]$, which we have stored as well, and take $\LF^\tau[i] = \LF^\tau[s'] + (i-s')$.

We are now ready to obtain our suffix array samples. 
We start with the position for the lexicographically smallest suffix (which, by concatenating a special symbol $\$$ to $X$, we can assume is the rightmost suffix).
Utilizing LF$^\tau$, we compute and store $\frac{n}{\tau}$ suffix array values evenly spaced by text position and their corresponding positions in the RL-BWT. 
This makes the $\SA$ value of any position in the BWT obtainable in $\tau$ applications of $\LF^1$ and computable in $\Ohtilde(\tau)$ time. 
Inverse suffix array, $\ISA$, queries can be supported with additional logarithmic factor overhead by using the LCE data structure (simply binary search over $\SA$ values). 
In summary, we have the following.

\begin{lemma} \label{lem:slow_index}
In time $ \Ohtilde(\sqrt{rn})$, we can obtain an index that answers $\SA$ and $\ISA$ queries 
in time $\Ohtilde(\sqrt{n/r})$ time.
\end{lemma}

For a given range $[s\dd e]$,
the longest common prefix, $\LCP$, of all suffixes $X[\SA[i]\dd n]$, $i \in [s\dd e]$ 
is $\LCE(\SA[s],\SA[e])$. Therefore, such queries can also be supported in time $\Ohtilde(\tau)$.

\subsection{Constructing the Index for Suffix Array and Inverse Suffix Array Queries}
The r-index ($\Oh(r)$ space version) for locating and counting pattern occurrences can be constructed by sampling suffix array values at the boundaries of BWT runs~\cite{DBLP:conf/soda/GagieNP18}. Requiring $\Ohtilde(r)$ queries on the structure in Lemma~\ref{lem:slow_index}, this takes $\Ohtilde(\sqrt{rn})$ time. 
Next, we describe how to construct the suffix array index in~\cite{DBLP:journals/jacm/GagieNP20} also in $\Ohtilde(\sqrt{rn})$ time. 
We start with the notion of the differential suffix array ($\DSA[i] = \SA[i] - \SA[i-1]$ for all $i > 1$) and a related lemma:

\begin{lemma}[\cite{DBLP:journals/jacm/GagieNP20}]
\label{lem:q_pointers}
Let $[i-1\dd i+s]$ be within a BWT run, for some $1 < i \leq n$ and $0 \leq s \leq n-i$. 
Then, there exists $q \neq i$ such that $\DSA[q\dd q+s] = \DSA[i\dd i+s]$ and $[q-1\dd q+s]$ contains the first position of a BWT run.
\end{lemma}

Lemma~\ref{lem:q_pointers} implies that the LF-mapping applied to a portion of the $\DSA$ completely contained within a BWT run preserves the $\DSA$ values, making the $\DSA$ highly compressible. 
With this observation, Gagie et al. obtained their index, which consists of $\Oh(\log \frac{n}{r})$ levels, with $\Oh(r)$ nodes each.
The nodes on a given level maintain pointers to nodes on the next level based on the $q$ values from Lemma \ref{lem:q_pointers}. 
Along with $\DSA$ values, $\SA$ values, and `offsets' kept for each node, these pointers are sufficient for efficiently recovering any suffix array value.
We refer the reader to \cite{DBLP:journals/jacm/GagieNP20} for further details.

A key operation to construct this data structure is finding these pointers, or $q$ values, from Lemma \ref{lem:q_pointers}. Following this, the remaining values needed per node are easily obtained from $\SA$ and $\ISA$ queries using Lemma~\ref{lem:slow_index}. As the next lemma demonstrates, for an arbitrary range $[s\dd e]$ we can obtain such a pointer in $\Ohtilde(\tau)$ time using the previously computed values from Section \ref{sec:lf_tau}. 

\begin{lemma}\label{lem:LF_k}
Given that $\SA[i]$ for arbitrary $i$ can be computed in $\Ohtilde(\tau)$ time and (reversed) LCE queries in $\Ohtilde(1)$ time, for a given range $[s\dd e]$, we can find $k$, $s'$, $e'$ such that $k \geq 0$ is the smallest value where $\LF^k([s\dd e]) = [s'\dd  e']$ and $[s'\dd e']$ contains the start of BWT-run.
\end{lemma}

\begin{proof}
We assume $k \geq 1$; otherwise, we determine from the RL-BWT that $[s\dd e]$ contains a run. 
We use exponential search on $k$. For a given $k$, we compute $s' = \ISA[\SA[s]-k]$ and $e' = \ISA[\SA[e]-k]$. We check whether $e'-s' =e -s$ and whether $\LCP([s'\dd  e']) \geq k$.
If both of these conditions hold, then no run boundary has been encountered yet. The largest $k$, $s'$, and $e'$ for which these conditions hold is returned.
Utilizing Lemma~\ref{lem:slow_index}, this takes $\Ohtilde(\tau)$ time per $\SA$ or $\ISA$ query, resulting in $\Ohtilde(\tau)$ time overall.
\end{proof}

In summary, the construction
of the suffix array index
requires a total of $\Ohtilde(r)$
queries on the structure in Lemma~\ref{lem:slow_index} and 
calls to the algorithm described in the proof of Lemma \ref{lem:LF_k}.
Inverse suffix array queries are supported with the help of LCE data structure as before.
This completes the proof of Theorem~\ref{thm:index}.

\index*

\subsection{Applications}
\label{sec:applications}
Having constructed suffix array index for $X$, we can solve a number of other problems in sub-linear time:

\paragraph{Longest Common Substring:} Given two strings $S_1$ and $S_2$, we let $z_{1,2}$ be the number of LZ77 factors of $S_1\$S_2$. In $\Ohtilde(\sqrt{z_{1,2}(|S_1| + |S_2|)})$ time, we construct the LZ77 parse and LCE-data structure for the $S_1\$S_2$. 

\begin{lemma}
\label{lem:lcs}
If $X$ is the longest common substring of $S_1$ and $S_2$, then there exists $i_1$,$j_1$, $i_2$, $j_2$ such that $S_1[i_1\dd j_1] = X$ and $S_2[i_2\dd j_2] = X$. For the $\BWT$ and $\ISA$ of $S_1\$S_2$, we have $|\ISA[i_1] - \ISA[|S_1|+1+i_2]| = 1$. Moreover, for this instance, $\BWT[\ISA[i_1]] \neq \BWT[\ISA[i_2]]$.
\end{lemma}

\begin{proof}
Suppose all occurrences of $X$ satisfy $|\ISA[i_1] - \ISA[|S_1|+1+i_2]| > 1$, and let $i_1$ and $i_2$ be such that $|\ISA[i_1] - \ISA[|S_1|+1+i_2]|$ is minimized. 
Suppose w.l.o.g. that $\ISA[i_1] < \ISA[|S_1|+1+i_2]$. 
Then, there must exist indices $x$ and $y$ such that $\ISA[i_1] \leq x \leq y \leq \ISA[|S_1|+1+i_2]$ such that $y-x = 1$ and $\SA[x],\SA[y]$ are indices into ranges for different strings, i.e., $\SA[x] \in [1\dd  |S_1|]$ and $\SA[y] \in [|S_1|+2\dd  |S_1|+|S_2|+1]$ or $\SA[x] \in [|S_1|+2\dd  |S_1|+|S_2| + 1]$ and $\SA[y] \in [1\dd  |S_1|]$, and $\LCE(\SA[x], \SA[y]) \geq \LCE(\SA[i_1], \SA[i_2])$. 
Hence, the prefix $X$ is shared by $(S_1\$S_2)[x\dd ]$ and $(S_1\$S_2)[y\dd ]$, a contradiction.
At the same time, if $\BWT[i_1] = \BWT[|S_1|+1+i_2]$, then the longest common substring could be extended to the left, a contradiction.
\end{proof}

Based on Lemma \ref{lem:lcs}, we can find the longest common string of $S_1$ and $S_2$ by checking the runs in BWT of $S_1\$S_2$, checking if adjacent $\SA$ values correspond to suffixes of different strings, and through the LCE queries, taking the pair with the longest shared prefix.

Note that $\Omega(\sqrt{|S_1| + |S_2|})$ time is necessary for the problem when $z = \Theta(1)$, as 
determining whether $S_1 = 0^{n}$ or whether $S_1$ contains a single $1$ requires $\Omega(\sqrt{n})$ time. The reduction simply sets $S_2 = `1\textrm'$.

\paragraph{Maximal Unique Matches:} Given two strings $S_1$ and $S_2$, we can identify all maximal unique matches in $\Ohtilde(r)$ additional time after constructing the RL-BWT and our index for $S_1\$S_2$. To do so we iterate through all run boundaries in the RL-BWT. We wish to identify all indices $i$ where: 
\begin{itemize}
\item $\BWT[i] \neq \BWT[i+1]$; 
\item $\SA[i] \in [1\dd |S_1|]$ and $\SA[i+1] \in [|S_1|+2\dd  |S_1|+|S_2|]$, or $\SA[i] \in [|S_1|+2\dd  |S_1|+|S_2|]$ and $\SA[i+1] \in [1\dd  |S_1|]$; 
\item and either 

\begin{itemize}
    \item $i = 1$ and $\LCE(\SA[i], \SA[i+1]) > \LCE(\SA[i+1], \SA[i+2])$,
    
    \item $i > 1$ and $i+1 < n$ and \[
    \LCE(\SA[i], \SA[i+1]) > \max(\LCE(\SA[i-1], \SA[i]), \LCE(\SA[i+1], \SA[i+2])),\] or
    
    \item $i > 1$ and $i + 1 = n$ and $\LCE(\SA[i], \SA[i+1]) > \LCE(\SA[i-1], \SA[i])$.
\end{itemize}
\end{itemize}

\paragraph{Lyndon Factorization:} 
The Lyndon factors of a string are determined by the values in $\ISA$ that are smaller than any previous value, i.e., $i \in [1\dd n]$ s.t. $\ISA[i] < \ISA[j]$ for all $j < i$. After constructing the $\SA$ index above in $\Ohtilde(\sqrt{zn})$, each $\ISA$ value can be queried in $\Ohtilde(1)$ time as well. 
Let $f$ denote the total number of Lyndon factors and let $i_j$ be the index where the $j^{th}$ Lyndon factor starts.
To find all Lyndon factors, we proceed from left to right keeping track of the minimum $\ISA$ value encountered thus far. 
Initially, this is $\ISA[i_1] = \ISA[1]$. 
We use exponential search on the rightmost boundary of the subarray being searched and Grover's search to identify the left-most index  $x \in [i_j+1\dd n]$ such that $\ISA[x] < \ISA[i_j]$. If one is encountered, we set $i_{j+1} = x$ and continue. 
Thanks to the exponential search, the total time taken is logarithmic factors from
\[
\sum_{j=1}^f \sqrt{i_{j+1} - i_j} \leq \sqrt{fn}.
\]
At the same time, the number of Lyndon factors $f$ is always bound by $\Ohtilde(z)$~\cite{DBLP:conf/stacs/KarkkainenKNPS17,DBLP:conf/cpm/UrabeNIBT19}. This yields the desired time complexity of $\Ohtilde(\sqrt{zn})$.

\paragraph{Q-gram Frequencies:} 
We need to identify the nodes $v$ of the suffix tree at string depth $q$ and the number of leaves in their respective subtrees.
A string matching the root-to-$v$ path, where $v$ has string depth $q$, is a q-gram and the size of the subtree its frequency.
To do this, we start with $i = 1$. Our goal is to find the largest $j$ such that $\LCE(\SA[i], \SA[j]) \geq q$. This can be found using exponential search on $j$ starting with $j = i$. Once this largest $j$ is found, the size of the subtree is equal to $j-i+1$. The q-gram itself is $X[\SA[i]\dd \SA[i]+q)$. 
We then set $i = j+1$ and repeat this process until the $j$ found through exponential search equals $n$. 
The overall time after index construction is $\Ohtilde(occ)$, where $occ$ is the number of q-grams. 

The $\Ohtilde(\sqrt{zn})$ time complexity is near optimal as a straightforward reduction from the Threshold problem discussed in \cref{sec:lb} with $q = 1$ and binary strings indicates that $\Omega(\sqrt{zn})$ input queries are required.

\vspace{1em}
These are likely only a small sample of the problems that can be solved using the proposed techniques. 
It is also worth restating that the compressed forms of the text can be obtained in $\Oh(\sqrt{zn})$ input queries, hence all string problems can be solved with that many input queries, albeit, perhaps with greater time needed.
\section{Lower Bounds}
\label{sec:lb}

\subsection{Hardness for Extreme \texorpdfstring{$z$}{z} and \texorpdfstring{$r$}{r}}

We first consider the case of $z, r = \Theta(1)$.
Let $X$ be a binary string and $S$ be the set of indices $i$ in $X$ such that $X[i] = 1$.
One can easily show, through an application of the adversarial method of Ambainis, that $\Omega(\sqrt{n})$ queries are required to determine if $|S| = 1$, even given the promise that $|S| = 0$ or $|S| = 1$~\cite{DBLP:conf/stoc/Ambainis00}. In the case where $|S| = 0$, we have $z = 2$ (and $r = 1$), and in the case where $|S| = 1$, we have $3 \leq z \leq 5$ and $2 \leq r \leq 3$. 
Note also that, even if we only obtained a suffix array index and not the compressed encodings, this would allow us to determine if there exists a $1$ in $X$ with a single additional query.

To show hardness for $z=r=n$ (and a large alphabet), we consider the problem of determining the $\left(\frac{n}{2}\right)^{th}$ largest value output from an oracle $f$ that outputs distinct values for each index.
This unordered searching problem has a known $\Omega(n)$ lower bound~\cite{DBLP:conf/stoc/NayakW99}.
For the string representation $X = \bigcirc_{i=1}^n f(i)$, if we could obtain an LZ77 encoding of $X$ using $o(z)$, or equivalently $o(n)$, oracle queries, then we could decompress the result and return the median element to solve the unordered searching problem. Similarly, if we could obtain the RL-BWT of $X$ in $o(r)$ queries, we could decompress it to obtain the median element in $o(n)$ queries.
Note also that even if we only had a suffix array index and not the compressed encoding, with a single additional query we could find the median element.

The above observations yield the following results.

\begin{theorem}
\label{thm:lb_no_qram}
Obtaining the LZ77 factorization of a text $X[1..n]$ with $z$ LZ77 factors, or RL-BWT with $r$ runs, requires $\Omega(\sqrt{zn})$ queries ($\Omega(\sqrt{rn})$ queries resp.) when $z, r = \Theta(1)$ or $z,r = \Theta(n)$.
\end{theorem}

\begin{theorem}
Constructing a data structure that supports $o(\sqrt{zn})$ time 
 suffix array queries of a text $X[1..n]$ with $z$ LZ77 factors, or RL-BWT with $r$ runs, requires  $\Omega(\sqrt{zn})$ ($\Omega(\sqrt{rn})$ resp.) input queries when $z, r = \Theta(1)$ or $z,r = \Theta(n)$.
\end{theorem}

\subsection{Parameterized Hardness of Obtaining the LZ77 Factorization and RL-BWT}
\label{sec:hardness_parse}

In this section, we aim to show that $\Omega(\sqrt{zn})$ ($\Omega(\sqrt{rn})$) input oracle queries are still required to obtain the LZ77 (RL-BWT resp.) encoding, even when $z$ ($r$ resp.) is restricted to a small range of possible values that are not necessarily constant or close to $n$.
We provide a reduction from the Threshold Problem, which is defined as follows:

\begin{prob}[Threshold Problem]
Given an oracle $f:[1\dd n]\rightarrow \{0,1\}$ and integer $t \ge 0$, determine if there exist at least $t$ inputs $i$ such that $f(i) = 1$, i.e., if  $S = |\{i \in [1\dd n] \mid f(i) = 1\}|$ satisfies $|S| \leq t$.
\end{prob}

Known lower bounds on the quantum query complexity state that, for $0\leq t < \frac{n}{2}$, at least $\Omega(\sqrt{(t+1)n})$ queries to the input oracle are required to solve the Threshold Problem~\cite{DBLP:journals/eatcs/HoyerS05,Pat92}.
Here, we treat $t$ as a polynomial function of $n$, i.e., $t = n^\xi$ where $\xi \in (0,1)$ is a constant.

One obstacle in using the Threshold problem to establish hardness results is that we cannot make assumptions concerning the size of the set $|S|$, and we do not assume knowledge of $z$ for the problem of finding the LZ77 factorization. 
The following Lemma helps to relate the size of the set $S$ and the number of LZ77 factors, $z$, in the binary string representation of $f$.

\begin{lemma}
\label{lem:z_leq_t}
If the string representation $X = \bigcirc_{i=1}^n f(i)$ has $|S| \geq 0$ ones, then the LZ77 factorization size of $X$ is $z \leq 3|S| + 2$.
\end{lemma}

\begin{proof}
Each run of $0$'s contributes at most two factors to the factorization and each $1$ symbol contributes at most one factor; hence, each $0^x1^y$ substring, $x \geq 0$, $y \geq 1$ accounts for at most $3$ factors.
Additional two factors account for a possible suffix of all $0$'s.
\end{proof}

For finding the RL-BWT, we can establish a similar result.
\begin{lemma}
\label{lem:rlbwt_runs}
If the string representation $X = \bigcirc_{i=1}^n f(i)$ has $|S| \geq 0$ ones, then the number of runs in BWT of $X$ is $r \leq 2|S|+1$.
\end{lemma}

\begin{proof}
Consider any permutation of a binary string with $|S|$ ones and $n-|S|$ zeroes.
To maximize the number of runs, we alternate between ones and zeroes. 
This creates at most two runs per every one. 
Accounting for the suffix of zeroes, this creates at most $2|S| + 1$ runs. 
Finally, note that the BWT is a permutation of $X$.
\end{proof}

Let the instance of the Threshold Problem with $f:[1\dd n] \rightarrow \{0,1\}$ and $t$ be given. Assume for the sake of contradiction that there exists an algorithm $A$ for finding the LZ77 factorization of $X$ in $q(n,z) = o(\sqrt{zn})$ queries for strings satisfying $t^{1-\varepsilon} \le z \le t$ for some constant $\varepsilon >0$. 
As an example, suppose that there exists an algorithm for finding the LZ77 factorization in $q(n,z)$ time whenever $n^{\frac{1}{2} - \varepsilon} \le z \le n^{\frac{1}{2}}$ for some constant $\varepsilon > 0$. 

Based on $n$, $t$, and $q(n,t)$, we create a query threshold $\kappa(n,t) \in \omega(q(n,t)) \cap o(\sqrt{(t+1)n})$. In particular, we can take $\kappa(n,t) = \sqrt{q(n,t)} \cdot \left((t+1)n\right)^\frac{1}{4}$, which has $\lim_{n\rightarrow \infty} \kappa(n,t) / \sqrt{(t+1)n} = 0$ and $\lim_{n\rightarrow \infty} q(n,t) / \kappa(n,t) = 0$. To solve the Threshold Problem with the LZ77 Factorization algorithm, do as follows:

\begin{enumerate}

\item Run the $\Ohtilde(\sqrt{zn})$-query algorithm from Section \ref{sec:lz}, but halt if the number of input queries reaches $\kappa(n,t)$.
If we obtain a complete encoding without halting, we output the solution; otherwise, we continue to Step 2.
This solves the Threshold Problem instance in the case where $z \leq t^{1-\varepsilon}$ in $o(\sqrt{(t+1)n})$ input queries.

\item If we did not obtain an encoding in Step 1, we next run the algorithm $A$ but halt early if the number of input queries exceeds $\kappa(n,t)$. 
If we halt with a completed encoding, we output the solution based on the encoding of $X$. 
If we do not halt with a completed encoding, we output that $|S| > t$.
\end{enumerate}

In the case where $|S| \le t$, we have $z \leq 3|S| + 2 \leq 3t + 2$, and our algorithm solves the Threshold Problem in Step 1 using $\Ohtilde(\sqrt{zn})$ queries if $z \le t^{1-\varepsilon}$, which is $o(\sqrt{(t+1)n})$, or in Step 2 using $q(n,z) = o(\sqrt{(t+1)n})$ queries if $z \in [t^{1-\varepsilon}\dd t]$.
In the case where $|S| > t$, the number of queries is still bounded by $\kappa(n,t) = o(\sqrt{(t+1)n})$ in Steps 1 and 2, and we output that $|S| > t$. 
In either case, we solve the Threshold Problem with $o(\sqrt{(t+1)n})$ input queries. As such, the assumption that algorithm $A$ exists contradicts the known lower bounds.

Using Lemma \ref{lem:rlbwt_runs}, a near-identical argument holds for computing the RL-BWT of $X$ with $z$ replaced by $r$, showing that $\Omega(\sqrt{rn})$ queries are required. The above proves the following.

\begin{theorem}\label{thm:lb_}
No quantum algorithm exists for computing the LZ77 factorization, or RL-BWT, using $o(\sqrt{zn})$ queries $(o(\sqrt{rn})$ queries resp.$)$ for all texts with $z \in [t^{1-\varepsilon}\dd t]$ $(r \in [t^{1-\varepsilon}\dd t]$ resp.$)$ where $t = n^\xi$, $\xi \in (0,1)$, and $\varepsilon > 0$ is any constant. This holds for alphabets of size at least two.
\end{theorem}

\subsection{Parameterized Lower Bounds for Computing the Value \texorpdfstring{$z$}{z}}
\label{sec:hardness_z_value}

We next turn our attention to the problem of determining the number of factors, $z$, in the LZ77-factorization.
This is potentially an easier problem than actually computing the factorization. 
Unlike the proof for the hardness of computing the actual LZ77 factorization that uses a binary alphabet, this proof uses a larger integer alphabet.

Given inputs $f:[1\dd n] \rightarrow \{0,1\}$ and $t \geq 0$ to the Threshold Problem, we construct an input oracle for a string $X$ for the problem of determining the size of the LZ77 factorization.
We first define the function $s: \{0,1\}\times [1\dd n] \rightarrow [0\dd n]$ as:
\[
s(f(i),i) = 
\begin{cases}
0 & \text{if }f(i) = 0,\\
i & \text{if }f(i) = 1.
\end{cases}
\]
Our reduction will create an oracle for a string $X= 0^{2n} \circ \$ \circ  \left(\bigcirc_{i=1}^n 0 \circ s(f(i),i) \right) \circ 0$.
Formally, we define the oracle $X:[1 \dd 4n+2] \rightarrow \{0, \$\} \cup [1\dd n]$ where
\[
X[i] = \begin{cases}
0 & \text{if }1 \leq i \leq 2n \text{ or } i \text{ is even},\\
\$ & \text{if } i = 2n+1, \\
s\left(f(\tfrac{i - (2n + 1)}{2}), \tfrac{i - (2n + 1)}{2}\right) & \text{if } i > 2n + 1 \text{ and } i \text{ is odd}.
\end{cases}
\]
Observe that every symbol in $X$ can be computed in constant time given access to the oracle $f$. 
The correctness of the reduction will follow from Lemma \ref{lem:thresh_eq_z}.

\begin{lemma}
\label{lem:thresh_eq_z}
The LZ77 factorization of $X$ constructed above has $z = 2|S| + 4$.
\end{lemma}

\begin{proof}
The prefix $0^{2n} \circ \$$ always requires exactly 3 factors: $0$, $0^{2n-1}$, $\$$.
Any subsequent run of $0$'s  only requires one factor. 
This is because any run of $0$'s following this prefix is of length at most $2n$.
Next, note that the $0$ following $\$$ is the start of a new factor. Following this $0$ symbol, we claim that every $i$, where $f(i) = 1$, contributes exactly two new factors. This is due to a new factor being created for the leftmost occurrence of the symbol $i$ followed by a new factor for the run of $0$'s beginning after the symbol $i$. Observe that the first $0$ following the $\$$ is necessary for the lemma; otherwise, the cases where $f(1)= 1$ would result in a different number of factors.
\end{proof}

Similar to Section \ref{sec:hardness_parse}, assume for the sake of contradiction that we have an algorithm that solves the problem of determining the value $z$ in $q(n,z) = o(\sqrt{zn})$ for $z = [t^{1-\varepsilon}\dd t]$. We again define the same query threshold $\kappa(n,t) \in \omega(q(n,t)) \cap o(\sqrt{(t+1)n})$ as given earlier. Then, to solve the instance of the Threshold Problem, the two-step algorithm from Section \ref{sec:hardness_parse} is applied to the oracle for $X$, resulting in a solution using $o(\sqrt{(t+1)n})$ input queries. This demonstrates the following.

\begin{theorem}\label{thm:lb_value}
No quantum algorithm exists for computing the number of LZ77 factors of a text $X[1\dd n]$ that uses only $o(\sqrt{zn})$ queries for all texts with $z \in [t^{1-\varepsilon}\dd t]$, where $t = n^\xi$, $\xi \in (0,1)$, and $\varepsilon > 0$ is any constant.
\end{theorem}

\paragraph{Acknowledgement.}
S. Thankachan is partially supported by the U.S. National Science Foundation (NSF) awards CCF-2315822 and CCF-2316691.
	\bibliographystyle{alphaurl} 
	\bibliography{main}

\newcommand{\etalchar}[1]{$^{#1}$}
\begin{thebibliography}{BCFN22b}

\bibitem[ABI{\etalchar{+}}20]{DBLP:conf/mfcs/AmbainisBIKKPSS20}
Andris Ambainis, Kaspars Balodis, J\={a}nis Iraids, Kamil Khadiev, Vladislavs K\c{l}evickis, Kri\v{s}j\={a}nis Pr\={u}sis, Yixin Shen, Juris Smotrovs, and Jevg\={e}nijs Vihrovs.
\newblock Quantum lower and upper bounds for {2D}-grid and {Dyck} language.
\newblock In {\em 45th International Symposium on Mathematical Foundations of Computer Science, {MFCS} 2020}, pages 8:1--8:14, 2020.
\newblock \href {https://doi.org/10.4230/LIPIcs.MFCS.2020.8} {\path{doi:10.4230/LIPIcs.MFCS.2020.8}}.

\bibitem[AGS19]{DBLP:conf/focs/AaronsonGS19}
Scott Aaronson, Daniel Grier, and Luke Schaeffer.
\newblock A quantum query complexity trichotomy for regular languages.
\newblock In {\em 60th {IEEE} Annual Symposium on Foundations of Computer Science, {FOCS} 2019}, pages 942--965, 2019.
\newblock \href {https://doi.org/10.1109/FOCS.2019.00061} {\path{doi:10.1109/FOCS.2019.00061}}.

\bibitem[AJ22]{DBLP:conf/soda/AkmalJ22}
Shyan Akmal and Ce~Jin.
\newblock Near-optimal quantum algorithms for string problems.
\newblock In {\em 33rd {ACM-SIAM} Symposium on Discrete Algorithms, {SODA} 2022}, pages 2791--2832. {SIAM}, 2022.
\newblock \href {https://doi.org/10.1137/1.9781611977073.109} {\path{doi:10.1137/1.9781611977073.109}}.

\bibitem[Amb00]{DBLP:conf/stoc/Ambainis00}
Andris Ambainis.
\newblock Quantum lower bounds by quantum arguments.
\newblock In {\em 32nd Annual {ACM} Symposium on Theory of Computing, STOC 2000}, pages 636--643. {ACM}, 2000.
\newblock \href {https://doi.org/10.1145/335305.335394} {\path{doi:10.1145/335305.335394}}.

\bibitem[Amb04]{ambainis2004quantum}
Andris Ambainis.
\newblock Quantum query algorithms and lower bounds.
\newblock In {\em Classical and New Paradigms of Computation and their Complexity Hierarchies}, pages 15--32. Springer, 2004.
\newblock \href {https://doi.org/10.1007/978-1-4020-2776-5_2} {\path{doi:10.1007/978-1-4020-2776-5_2}}.

\bibitem[BBC{\etalchar{+}}01]{BBCMW01}
Robert Beals, Harry Buhrman, Richard Cleve, Michele Mosca, and Ronald de~Wolf.
\newblock Quantum lower bounds by polynomials.
\newblock {\em Journal of the ACM}, 48(4):778--797, 2001.
\newblock \href {https://doi.org/10.1145/502090.502097} {\path{doi:10.1145/502090.502097}}.

\bibitem[BCFN22a]{BCFN22a}
Karl Bringmann, Alejandro Cassis, Nick Fischer, and Vasileios Nakos.
\newblock Almost-optimal sublinear-time edit distance in the low distance regime.
\newblock In {\em 54th Annual {ACM} {SIGACT} Symposium on Theory of Computing STOC 2022}, pages 1102--1115. {ACM}, 2022.
\newblock \href {https://doi.org/10.1145/3519935.3519990} {\path{doi:10.1145/3519935.3519990}}.

\bibitem[BCFN22b]{BCFN22b}
Karl Bringmann, Alejandro Cassis, Nick Fischer, and Vasileios Nakos.
\newblock Improved sublinear-time edit distance for preprocessed strings.
\newblock In {\em 49th International Colloquium on Automata, Languages, and Programming, {ICALP} 2022}, volume 229 of {\em LIPIcs}, pages 32:1--32:20. Schloss Dagstuhl--Leibniz-Zentrum f{\"{u}}r Informatik, 2022.
\newblock \href {https://doi.org/10.4230/LIPIcs.ICALP.2022.32} {\path{doi:10.4230/LIPIcs.ICALP.2022.32}}.

\bibitem[BdW02]{DBLP:journals/tcs/BuhrmanW02}
Harry Buhrman and Ronald de~Wolf.
\newblock Complexity measures and decision tree complexity: a survey.
\newblock {\em Theoretical Computer Science}, 288(1):21--43, 2002.
\newblock \href {https://doi.org/10.1016/S0304-3975(01)00144-X} {\path{doi:10.1016/S0304-3975(01)00144-X}}.

\bibitem[BEG{\etalchar{+}}21]{boroujeni2021approximating}
Mahdi Boroujeni, Soheil Ehsani, Mohammad Ghodsi, MohammadTaghi HajiAghayi, and Saeed Seddighin.
\newblock Approximating edit distance in truly subquadratic time: Quantum and {MapReduce}.
\newblock {\em Journal of the ACM}, 68(3):1--41, 2021.
\newblock \href {https://doi.org/10.1145/3456807} {\path{doi:10.1145/3456807}}.

\bibitem[BI18]{BI18}
Arturs Backurs and Piotr Indyk.
\newblock Edit distance cannot be computed in strongly subquadratic time (unless {SETH} is false).
\newblock {\em SIAM Journal on Computing}, 47(3):1087--1097, 2018.
\newblock \href {https://doi.org/10.1137/15M1053128} {\path{doi:10.1137/15M1053128}}.

\bibitem[BK23]{BK23}
Sudatta Bhattacharya and Michal Kouck{\'{y}}.
\newblock Locally consistent decomposition of strings with applications to edit distance sketching.
\newblock In {\em 55th Annual {ACM} Symposium on Theory of Computing, {STOC} 2023}, pages 219--232. {ACM}, 2023.
\newblock \href {https://doi.org/10.1145/3564246.3585239} {\path{doi:10.1145/3564246.3585239}}.

\bibitem[BP16]{DBLP:conf/soda/BelazzouguiP16}
Djamal Belazzougui and Simon~J. Puglisi.
\newblock Range predecessor and {Lempel-Ziv} parsing.
\newblock In {\em 27th Annual {ACM-SIAM} Symposium on Discrete Algorithms, {SODA} 2016}, pages 2053--2071. {SIAM}, 2016.
\newblock \href {https://doi.org/10.1137/1.9781611974331.ch143} {\path{doi:10.1137/1.9781611974331.ch143}}.

\bibitem[BPS21]{BPS21}
Harry Buhrman, Subhasree Patro, and Florian Speelman.
\newblock A framework of quantum strong exponential-time hypotheses.
\newblock In {\em 38th International Symposium on Theoretical Aspects of Computer Science, {STACS} 2021}, volume 187 of {\em LIPIcs}, pages 19:1--19:19. Schloss Dagstuhl-- Leibniz-Zentrum f{\"{u}}r Informatik, 2021.
\newblock \href {https://doi.org/10.4230/LIPIcs.STACS.2021.19} {\path{doi:10.4230/LIPIcs.STACS.2021.19}}.

\bibitem[BW94]{burrows1994block}
Michael Burrows and David~J. Wheeler.
\newblock A block-sorting lossless data compression algorithm.
\newblock Technical Report 124, Digital Equipment Corporation, 1994.

\bibitem[BZ16]{BZ16}
Djamal Belazzougui and Qin Zhang.
\newblock Edit distance: Sketching, streaming, and document exchange.
\newblock In {\em 57th Annual {IEEE} Symposium on Foundations of Computer Science, {FOCS} 2016}, pages 51--60. {IEEE} Computer Society, 2016.
\newblock \href {https://doi.org/10.1109/FOCS.2016.15} {\path{doi:10.1109/FOCS.2016.15}}.

\bibitem[CKK{\etalchar{+}}22]{quantdc}
Andrew~M. Childs, Robin Kothari, Matt Kovacs{-}Deak, Aarthi Sundaram, and Daochen Wang.
\newblock Quantum divide and conquer, 2022.
\newblock \href {https://arxiv.org/abs/2210.06419} {\path{arXiv:2210.06419}}.

\bibitem[CKW23]{CKW23}
Alejandro Cassis, Tomasz Kociumaka, and Philip Wellnitz.
\newblock Optimal algorithms for bounded weighted edit distance.
\newblock In {\em 64th {IEEE} Annual Symposium on Foundations of Computer Science, {FOCS} 2023}. IEEE, 2023.
\newblock \href {https://arxiv.org/abs/2305.06659} {\path{arXiv:2305.06659}}.

\bibitem[CLL{\etalchar{+}}05]{DBLP:journals/tit/CharikarLLPPSS05}
Moses Charikar, Eric Lehman, Ding Liu, Rina Panigrahy, Manoj Prabhakaran, Amit Sahai, and Abhi Shelat.
\newblock The smallest grammar problem.
\newblock {\em IEEE Transactions on Information Theory}, 51(7):2554--2576, 2005.
\newblock \href {https://doi.org/10.1109/TIT.2005.850116} {\path{doi:10.1109/TIT.2005.850116}}.

\bibitem[DGH{\etalchar{+}}23]{DGHKS23}
Debarati Das, Jacob Gilbert, MohammadTaghi Hajiaghayi, Tomasz Kociumaka, and Barna Saha.
\newblock Weighted edit distance computation: Strings, trees, and {Dyck}.
\newblock In {\em 55th Annual {ACM} Symposium on Theory of Computing, {STOC} 2023}, pages 377--390. {ACM}, 2023.
\newblock \href {https://doi.org/10.1145/3564246.3585178} {\path{doi:10.1145/3564246.3585178}}.

\bibitem[FFM00]{DBLP:conf/focs/Farach97}
Martin Farach{-}Colton, Paolo Ferragina, and S.~Muthukrishnan.
\newblock On the sorting-complexity of suffix tree construction.
\newblock {\em Journal of the ACM}, 47(6):987--1011, 2000.
\newblock \href {https://doi.org/10.1145/355541.355547} {\path{doi:10.1145/355541.355547}}.

\bibitem[Fil20]{shiwangshiwang}
Yuval Filmus.
\newblock To find median of $k$ sorted arrays of $n$ elements each in less than {$O(nk\log k)$}, Nov 2020.
\newblock URL: \url{https://cs.stackexchange.com/questions/87695/to-find-median-of-k-sorted-arrays-of-n-elements-each-in-less-than-onk-log/156925\#156925}.

\bibitem[FM05]{FM05}
Paolo Ferragina and Giovanni Manzini.
\newblock Indexing compressed text.
\newblock {\em Journal of the ACM}, 52(4):552--581, 2005.
\newblock \href {https://doi.org/10.1145/1082036.1082039} {\path{doi:10.1145/1082036.1082039}}.

\bibitem[GKLS22]{GKLS22}
Arun Ganesh, Tomasz Kociumaka, Andrea Lincoln, and Barna Saha.
\newblock How compression and approximation affect efficiency in string distance measures.
\newblock In {\em 33rd {ACM-SIAM} Symposium on Discrete Algorithms, {SODA} 2022}, pages 2867--2919. {SIAM}, 2022.
\newblock \href {https://doi.org/10.1137/1.9781611977073.112} {\path{doi:10.1137/1.9781611977073.112}}.

\bibitem[GKS19]{GKS19}
Elazar Goldenberg, Robert Krauthgamer, and Barna Saha.
\newblock Sublinear algorithms for gap edit distance.
\newblock In {\em 60th {IEEE} Annual Symposium on Foundations of Computer Science, {FOCS} 2019}, pages 1101--1120. {IEEE}, 2019.
\newblock \href {https://doi.org/10.1109/FOCS.2019.00070} {\path{doi:10.1109/FOCS.2019.00070}}.

\bibitem[GNP18a]{GNP18}
Travis Gagie, Gonzalo Navarro, and Nicola Prezza.
\newblock On the approximation ratio of lempel-ziv parsing.
\newblock In {\em 13th Latin American Symposium on Theoretical Informatics, {LATIN} 2018}, volume 10807 of {\em LNCS}, pages 490--503. Springer, 2018.
\newblock \href {https://doi.org/10.1007/978-3-319-77404-6_36} {\path{doi:10.1007/978-3-319-77404-6_36}}.

\bibitem[GNP18b]{DBLP:conf/soda/GagieNP18}
Travis Gagie, Gonzalo Navarro, and Nicola Prezza.
\newblock Optimal-time text indexing in bwt-runs bounded space.
\newblock In {\em 29th Annual {ACM-SIAM} Symposium on Discrete Algorithms, {SODA} 2018}, pages 1459--1477. {SIAM}, 2018.
\newblock \href {https://doi.org/10.1137/1.9781611975031.96} {\path{doi:10.1137/1.9781611975031.96}}.

\bibitem[GNP20]{DBLP:journals/jacm/GagieNP20}
Travis Gagie, Gonzalo Navarro, and Nicola Prezza.
\newblock Fully functional suffix trees and optimal text searching in {BWT}-runs bounded space.
\newblock {\em Journal of the ACM}, 67(1):2:1--2:54, 2020.
\newblock \href {https://doi.org/10.1145/3375890} {\path{doi:10.1145/3375890}}.

\bibitem[Gro96]{DBLP:conf/stoc/Grover96}
Lov~K. Grover.
\newblock A fast quantum mechanical algorithm for database search.
\newblock In {\em 28th Annual {ACM} Symposium on the Theory of Computing, STOC 1996}, pages 212--219, 1996.
\newblock \href {https://doi.org/10.1145/237814.237866} {\path{doi:10.1145/237814.237866}}.

\bibitem[GRS20]{GRS20}
Elazar Goldenberg, Aviad Rubinstein, and Barna Saha.
\newblock Does preprocessing help in fast sequence comparisons?
\newblock In {\em 52nd Annual {ACM} {SIGACT} Symposium on Theory of Computing, {STOC} 2020}, pages 657--670. {ACM}, 2020.
\newblock \href {https://doi.org/10.1145/3357713.3384300} {\path{doi:10.1145/3357713.3384300}}.

\bibitem[GS22]{DBLP:conf/innovations/GallS22}
Fran{\c{c}}ois~Le Gall and Saeed Seddighin.
\newblock Quantum meets fine-grained complexity: Sublinear time quantum algorithms for string problems.
\newblock In Mark Braverman, editor, {\em 13th Innovations in Theoretical Computer Science Conference, {ITCS} 2022}, volume 215 of {\em LIPIcs}, pages 97:1--97:23. Schloss Dagstuhl--Leibniz-Zentrum f{\"{u}}r Informatik, 2022.
\newblock \href {https://doi.org/10.4230/LIPIcs.ITCS.2022.97} {\path{doi:10.4230/LIPIcs.ITCS.2022.97}}.

\bibitem[GV05]{DBLP:journals/siamcomp/GrossiV05}
Roberto Grossi and Jeffrey~Scott Vitter.
\newblock Compressed suffix arrays and suffix trees with applications to text indexing and string matching.
\newblock {\em SIAM Journal on Computing}, 35(2):378--407, 2005.
\newblock \href {https://doi.org/10.1137/S0097539702402354} {\path{doi:10.1137/S0097539702402354}}.

\bibitem[HS05]{DBLP:journals/eatcs/HoyerS05}
Peter H{\o}yer and Robert Spalek.
\newblock Lower bounds on quantum query complexity.
\newblock {\em Bulletin of {EATCS}}, 87:78--103, 2005.
\newblock \href {https://arxiv.org/abs/quant-ph/0509153} {\path{arXiv:quant-ph/0509153}}.

\bibitem[HV03]{DBLP:journals/jda/HariharanV03}
Ramesh Hariharan and V.~Vinay.
\newblock String matching in {$\tilde{O}(\sqrt{n}+\sqrt{m})$} quantum time.
\newblock {\em Journal of Discrete Algorithms}, 1(1):103--110, 2003.
\newblock \href {https://doi.org/10.1016/S1570-8667(03)00010-8} {\path{doi:10.1016/S1570-8667(03)00010-8}}.

\bibitem[I17]{I17}
Tomohiro I.
\newblock Longest common extensions with recompression.
\newblock In {\em 28th Annual Symposium on Combinatorial Pattern Matching, {CPM} 2017}, volume~78 of {\em LIPIcs}, pages 18:1--18:15. Schloss Dagstuhl--Leibniz-Zentrum f{\"{u}}r Informatik, 2017.
\newblock \href {https://doi.org/10.4230/LIPIcs.CPM.2017.18} {\path{doi:10.4230/LIPIcs.CPM.2017.18}}.

\bibitem[IP01]{IP01}
Russell Impagliazzo and Ramamohan Paturi.
\newblock On the complexity of $k$-{SAT}.
\newblock {\em Journal of Computer and System Sciences}, 62(2):367--375, 2001.
\newblock \href {https://doi.org/10.1006/jcss.2000.1727} {\path{doi:10.1006/jcss.2000.1727}}.

\bibitem[JN23]{JN23}
Ce~Jin and Jakob Nogler.
\newblock Quantum speed-ups for string synchronizing sets, longest common substring, and \emph{k}-mismatch matching.
\newblock In {\em 34th {ACM-SIAM} Symposium on Discrete Algorithms, {SODA} 2023}, pages 5090--5121. {SIAM}, 2023.
\newblock \href {https://doi.org/10.1137/1.9781611977554.ch186} {\path{doi:10.1137/1.9781611977554.ch186}}.

\bibitem[JNW21]{JNW21}
Ce~Jin, Jelani Nelson, and Kewen Wu.
\newblock An improved sketching algorithm for edit distance.
\newblock In {\em 38th International Symposium on Theoretical Aspects of Computer Science, {STACS} 2021}, volume 187 of {\em LIPIcs}, pages 45:1--45:16. Schloss Dagstuhl--Leibniz-Zentrum f{\"{u}}r Informatik, 2021.
\newblock \href {https://doi.org/10.4230/LIPIcs.STACS.2021.45} {\path{doi:10.4230/LIPIcs.STACS.2021.45}}.

\bibitem[KK17a]{DBLP:conf/dcc/KempaK17}
Dominik Kempa and Dmitry Kosolobov.
\newblock {LZ}-end parsing in compressed space.
\newblock In {\em 2017 Data Compression Conference, {DCC} 2017}, pages 350--359. {IEEE}, 2017.
\newblock \href {https://doi.org/10.1109/DCC.2017.73} {\path{doi:10.1109/DCC.2017.73}}.

\bibitem[KK17b]{DBLP:conf/esa/KempaK17}
Dominik Kempa and Dmitry Kosolobov.
\newblock {LZ}-end parsing in linear time.
\newblock In {\em 25th Annual European Symposium on Algorithms, {ESA} 2017}, volume~87 of {\em LIPIcs}, pages 53:1--53:14. Schloss Dagstuhl--Leibniz-Zentrum f{\"{u}}r Informatik, 2017.
\newblock \href {https://doi.org/10.4230/LIPIcs.ESA.2017.53} {\path{doi:10.4230/LIPIcs.ESA.2017.53}}.

\bibitem[KK22]{DBLP:journals/cacm/KempaK22}
Dominik Kempa and Tomasz Kociumaka.
\newblock Resolution of the {Burrows-Wheeler} transform conjecture.
\newblock {\em Communications of the ACM}, 65(6):91--98, 2022.
\newblock \href {https://doi.org/10.1145/3531445} {\path{doi:10.1145/3531445}}.

\bibitem[KKN{\etalchar{+}}17]{DBLP:conf/stacs/KarkkainenKNPS17}
Juha K{\"{a}}rkk{\"{a}}inen, Dominik Kempa, Yuto Nakashima, Simon~J. Puglisi, and Arseny~M. Shur.
\newblock On the size of {Lempel-Ziv} and {Lyndon} factorizations.
\newblock In {\em 34th Symposium on Theoretical Aspects of Computer Science, {STACS} 2017}, volume~66 of {\em LIPIcs}, pages 45:1--45:13. Schloss Dagstuhl--Leibniz-Zentrum f{\"{u}}r Informatik, 2017.
\newblock \href {https://doi.org/10.4230/LIPIcs.STACS.2017.45} {\path{doi:10.4230/LIPIcs.STACS.2017.45}}.

\bibitem[KKP14]{DBLP:conf/dcc/KarkkainenKP14}
Juha K{\"{a}}rkk{\"{a}}inen, Dominik Kempa, and Simon~J. Puglisi.
\newblock {Lempel-Ziv} parsing in external memory.
\newblock In {\em Data Compression Conference, {DCC} 2014}, pages 153--162. {IEEE}, 2014.
\newblock \href {https://doi.org/10.1109/DCC.2014.78} {\path{doi:10.1109/DCC.2014.78}}.

\bibitem[KL21]{kempa2021fast}
Dominik Kempa and Ben Langmead.
\newblock Fast and space-efficient construction of {AVL} grammars from the {LZ77} parsing.
\newblock In {\em 29th Annual European Symposium on Algorithms, {ESA} 2021}, volume 204 of {\em LIPIcs}, pages 56:1--56:14. Schloss Dagstuhl--Leibniz-Zentrum f{\"{u}}r Informatik, 2021.
\newblock \href {https://doi.org/10.4230/LIPIcs.ESA.2021.56} {\path{doi:10.4230/LIPIcs.ESA.2021.56}}.

\bibitem[KN13]{DBLP:journals/tcs/KreftN13}
Sebastian Kreft and Gonzalo Navarro.
\newblock On compressing and indexing repetitive sequences.
\newblock {\em Theoretical Computer Science}, 483:115--133, 2013.
\newblock \href {https://doi.org/10.1016/j.tcs.2012.02.006} {\path{doi:10.1016/j.tcs.2012.02.006}}.

\bibitem[KNP23]{KNP23}
Tomasz Kociumaka, Gonzalo Navarro, and Nicola Prezza.
\newblock Towards a definitive compressibility measure for repetitive sequences.
\newblock {\em {IEEE} Transactions on Information Theory}, 69(4):2074--20, 2023.
\newblock \href {https://doi.org/10.1109/TIT.2022.3224382} {\path{doi:10.1109/TIT.2022.3224382}}.

\bibitem[Kot14]{DBLP:conf/stacs/Kothari14}
Robin Kothari.
\newblock An optimal quantum algorithm for the oracle identification problem.
\newblock In {\em 31st International Symposium on Theoretical Aspects of Computer Science, {STACS} 2014}, volume~25 of {\em LIPIcs}, pages 482--493. Schloss Dagstuhl--Leibniz-Zentrum f{\"{u}}r Informatik, 2014.
\newblock \href {https://doi.org/10.4230/LIPIcs.STACS.2014.482} {\path{doi:10.4230/LIPIcs.STACS.2014.482}}.

\bibitem[KP18]{DBLP:conf/stoc/KempaP18}
Dominik Kempa and Nicola Prezza.
\newblock At the roots of dictionary compression: string attractors.
\newblock In Ilias Diakonikolas, David Kempe, and Monika Henzinger, editors, {\em Proceedings of the 50th Annual {ACM} {SIGACT} Symposium on Theory of Computing, {STOC} 2018, Los Angeles, CA, USA, June 25-29, 2018}, pages 827--840. {ACM}, 2018.
\newblock \href {https://doi.org/10.1145/3188745.3188814} {\path{doi:10.1145/3188745.3188814}}.

\bibitem[KPS21]{KPS21}
Tomasz Kociumaka, Ely Porat, and Tatiana Starikovskaya.
\newblock Small-space and streaming pattern matching with $k$ edits.
\newblock In {\em 62nd {IEEE} Annual Symposium on Foundations of Computer Science, {FOCS} 2021}, pages 885--896. {IEEE}, 2021.
\newblock \href {https://doi.org/10.1109/FOCS52979.2021.00090} {\path{doi:10.1109/FOCS52979.2021.00090}}.

\bibitem[KS20]{KS20}
Tomasz Kociumaka and Barna Saha.
\newblock Sublinear-time algorithms for computing {\&} embedding gap edit distance.
\newblock In {\em 61st {IEEE} Annual Symposium on Foundations of Computer Science, {FOCS} 2020}, pages 1168--1179. {IEEE}, 2020.
\newblock \href {https://doi.org/10.1109/FOCS46700.2020.00112} {\path{doi:10.1109/FOCS46700.2020.00112}}.

\bibitem[KS22]{DBLP:conf/soda/KempaS22}
Dominik Kempa and Barna Saha.
\newblock An upper bound and linear-space queries on the {LZ}-end parsing.
\newblock In {\em 33rd {ACM-SIAM} Symposium on Discrete Algorithms, {SODA} 2022}, pages 2847--2866. {SIAM}, 2022.
\newblock \href {https://doi.org/10.1137/1.9781611977073.111} {\path{doi:10.1137/1.9781611977073.111}}.

\bibitem[Lev65]{Lev65}
Vladimir~I. Levenshtein.
\newblock Binary codes capable of correcting deletions, insertions and reversals.
\newblock {\em Doklady Akademii Nauk SSSR}, 163(4):845--848, 1965.
\newblock URL: \url{http://mi.mathnet.ru/eng/dan31411}.

\bibitem[Lit87]{DBLP:journals/ml/Littlestone87}
Nick Littlestone.
\newblock Learning quickly when irrelevant attributes abound: {A} new linear-threshold algorithm.
\newblock {\em Machine Learning}, 2(4):285--318, 1987.
\newblock \href {https://doi.org/10.1007/BF00116827} {\path{doi:10.1007/BF00116827}}.

\bibitem[LV88]{LV88}
Gad~M. Landau and Uzi Vishkin.
\newblock Fast string matching with k differences.
\newblock {\em Journal of Computer and System Sciences}, 37(1):63--78, 1988.
\newblock \href {https://doi.org/10.1016/0022-0000(88)90045-1} {\path{doi:10.1016/0022-0000(88)90045-1}}.

\bibitem[McC76]{DBLP:journals/jacm/McCreight76}
Edward~M. McCreight.
\newblock A space-economical suffix tree construction algorithm.
\newblock {\em Journal of the ACM}, 23(2):262--272, 1976.
\newblock \href {https://doi.org/10.1145/321941.321946} {\path{doi:10.1145/321941.321946}}.

\bibitem[MNN17]{DBLP:conf/soda/MunroNN17}
J.~Ian Munro, Gonzalo Navarro, and Yakov Nekrich.
\newblock Space-efficient construction of compressed indexes in deterministic linear time.
\newblock In {\em 28th Annual {ACM-SIAM} Symposium on Discrete Algorithms, {SODA} 2017}, pages 408--424. {SIAM}, 2017.
\newblock \href {https://doi.org/10.1137/1.9781611974782.26} {\path{doi:10.1137/1.9781611974782.26}}.

\bibitem[MP70]{Morris1970}
James~H. Morris, Jr. and Vaughan~R. Pratt.
\newblock A linear pattern-matching algorithm.
\newblock Technical Report~40, Department of Computer Science, University of California, Berkeley, 1970.

\bibitem[Mye86]{Mye86}
Eugene~W. Myers.
\newblock An {$O(ND)$} difference algorithm and its variations.
\newblock {\em Algorithmica}, 1(2):251--266, 1986.
\newblock \href {https://doi.org/10.1007/BF01840446} {\path{doi:10.1007/BF01840446}}.

\bibitem[Nav21]{DBLP:journals/csur/Navarro21a}
Gonzalo Navarro.
\newblock Indexing highly repetitive string collections, part {I:} repetitiveness measures.
\newblock {\em ACM Computing Surveys}, 54(2):29:1--29:31, 2021.
\newblock \href {https://doi.org/10.1145/3434399} {\path{doi:10.1145/3434399}}.

\bibitem[NII{\etalchar{+}}16]{DBLP:conf/mfcs/NishimotoIIBT16}
Takaaki Nishimoto, Tomohiro I, Shunsuke Inenaga, Hideo Bannai, and Masayuki Takeda.
\newblock Fully dynamic data structure for {LCE} queries in compressed space.
\newblock In {\em 41st International Symposium on Mathematical Foundations of Computer Science, {MFCS} 2016}, volume~58 of {\em LIPIcs}, pages 72:1--72:15. Schloss Dagstuhl--Leibniz-Zentrum f{\"{u}}r Informatik, 2016.
\newblock \href {https://doi.org/10.4230/LIPIcs.MFCS.2016.72} {\path{doi:10.4230/LIPIcs.MFCS.2016.72}}.

\bibitem[NII{\etalchar{+}}20]{DBLP:journals/dam/NishimotoIIBT20}
Takaaki Nishimoto, Tomohiro I, Shunsuke Inenaga, Hideo Bannai, and Masayuki Takeda.
\newblock Dynamic index and {LZ} factorization in compressed space.
\newblock {\em Discrete Applied Mathematics}, 274:116--129, 2020.
\newblock \href {https://doi.org/10.1016/j.dam.2019.01.014} {\path{doi:10.1016/j.dam.2019.01.014}}.

\bibitem[NKT22]{NishimotoKT22}
Takaaki Nishimoto, Shunsuke Kanda, and Yasuo Tabei.
\newblock An optimal-time {RLBWT} construction in {BWT}-runs bounded space.
\newblock In {\em 49th International Colloquium on Automata, Languages, and Programming, {ICALP} 2022}, volume 229 of {\em LIPIcs}, pages 99:1--99:20. Schloss Dagstuhl--Leibniz-Zentrum f{\"{u}}r Informatik, 2022.
\newblock \href {https://doi.org/10.4230/LIPIcs.ICALP.2022.99} {\path{doi:10.4230/LIPIcs.ICALP.2022.99}}.

\bibitem[NM07]{DBLP:journals/csur/NavarroM07}
Gonzalo Navarro and Veli M{\"{a}}kinen.
\newblock Compressed full-text indexes.
\newblock {\em ACM Computing Surveys}, 39(1):2, 2007.
\newblock \href {https://doi.org/10.1145/1216370.1216372} {\path{doi:10.1145/1216370.1216372}}.

\bibitem[NT21]{NishimotoT21}
Takaaki Nishimoto and Yasuo Tabei.
\newblock Optimal-time queries on {BWT}-runs compressed indexes.
\newblock In Nikhil Bansal, Emanuela Merelli, and James Worrell, editors, {\em 48th International Colloquium on Automata, Languages, and Programming, {ICALP} 2021}, volume 198 of {\em LIPIcs}, pages 101:1--101:15. Schloss Dagstuhl--Leibniz-Zentrum f{\"{u}}r Informatik, 2021.
\newblock \href {https://doi.org/10.4230/LIPIcs.ICALP.2021.101} {\path{doi:10.4230/LIPIcs.ICALP.2021.101}}.

\bibitem[NW70]{NW70}
Saul~B. Needleman and Christian~D. Wunsch.
\newblock A general method applicable to the search for similarities in the amino acid sequence of two proteins.
\newblock {\em Journal of Molecular Biology}, 48(3):443--453, 1970.
\newblock \href {https://doi.org/10.1016/0022-2836(70)90057-4} {\path{doi:10.1016/0022-2836(70)90057-4}}.

\bibitem[NW99]{DBLP:conf/stoc/NayakW99}
Ashwin Nayak and Felix Wu.
\newblock The quantum query complexity of approximating the median and related statistics.
\newblock In {\em 31st Annual {ACM} Symposium on Theory of Computing, STOC 1999}, pages 384--393. {ACM}, 1999.
\newblock \href {https://doi.org/10.1145/301250.301349} {\path{doi:10.1145/301250.301349}}.

\bibitem[Pat92]{Pat92}
Ramamohan Paturi.
\newblock On the degree of polynomials that approximate symmetric boolean functions (preliminary version).
\newblock In {\em 24th Annual {ACM} Symposium on Theory of Computing, STOC 1992}, pages 468--474. {ACM}, 1992.
\newblock \href {https://doi.org/10.1145/129712.129758} {\path{doi:10.1145/129712.129758}}.

\bibitem[RPE81]{DBLP:journals/jacm/RodehPE81}
Michael Rodeh, Vaughan~R. Pratt, and Shimon Even.
\newblock Linear algorithm for data compression via string matching.
\newblock {\em Journal of the ACM}, 28(1):16--24, 1981.
\newblock \href {https://doi.org/10.1145/322234.322237} {\path{doi:10.1145/322234.322237}}.

\bibitem[RRRS13]{DBLP:journals/algorithmica/RaskhodnikovaRRS13}
Sofya Raskhodnikova, Dana Ron, Ronitt Rubinfeld, and Adam~D. Smith.
\newblock Sublinear algorithms for approximating string compressibility.
\newblock {\em Algorithmica}, 65(3):685--709, 2013.
\newblock \href {https://doi.org/10.1007/s00453-012-9618-6} {\path{doi:10.1007/s00453-012-9618-6}}.

\bibitem[Rub19]{Rub19}
Aviad Rubinstein.
\newblock Quantum dna sequencing and the ultimate hardness hypothesis, 2019.
\newblock URL: \url{https://theorydish.blog/2019/12/09/quantum-dna-sequencing-the-ultimate-hardness-hypothesis/}.

\bibitem[Ryt03]{DBLP:journals/tcs/Rytter03}
Wojciech Rytter.
\newblock Application of {Lempel-Ziv} factorization to the approximation of grammar-based compression.
\newblock {\em Theoretical Computer Science}, 302(1-3):211--222, 2003.
\newblock \href {https://doi.org/10.1016/S0304-3975(02)00777-6} {\path{doi:10.1016/S0304-3975(02)00777-6}}.

\bibitem[Sad07]{DBLP:journals/mst/Sadakane07}
Kunihiko Sadakane.
\newblock Compressed suffix trees with full functionality.
\newblock {\em Theory of Computing Systems}, 41(4):589--607, 2007.
\newblock \href {https://doi.org/10.1007/s00224-006-1198-x} {\path{doi:10.1007/s00224-006-1198-x}}.

\bibitem[Sel74]{Sel74}
Peter~H. Sellers.
\newblock On the theory and computation of evolutionary distances.
\newblock {\em {SIAM} Journal on Applied Mathematics}, 26(4):787--793, 1974.
\newblock \href {https://doi.org/10.1137/0126070} {\path{doi:10.1137/0126070}}.

\bibitem[SS82]{DBLP:journals/jacm/StorerS82}
James~A. Storer and Thomas~G. Szymanski.
\newblock Data compression via textual substitution.
\newblock {\em Journal of the ACM}, 29(4):928--951, 1982.
\newblock \href {https://doi.org/10.1145/322344.322346} {\path{doi:10.1145/322344.322346}}.

\bibitem[Ukk85]{Ukk85}
Esko Ukkonen.
\newblock Algorithms for approximate string matching.
\newblock {\em Information and Control}, 64(1-3):100--118, 1985.
\newblock \href {https://doi.org/10.1016/S0019-9958(85)80046-2} {\path{doi:10.1016/S0019-9958(85)80046-2}}.

\bibitem[UNI{\etalchar{+}}19]{DBLP:conf/cpm/UrabeNIBT19}
Yuki Urabe, Yuto Nakashima, Shunsuke Inenaga, Hideo Bannai, and Masayuki Takeda.
\newblock On the size of overlapping {Lempel-Ziv} and {Lyndon} factorizations.
\newblock In {\em 30th Annual Symposium on Combinatorial Pattern Matching, {CPM} 2019}, volume 128 of {\em LIPIcs}, pages 29:1--29:11. Schloss Dagstuhl--Leibniz-Zentrum f{\"{u}}r Informatik, 2019.
\newblock \href {https://doi.org/10.4230/LIPIcs.CPM.2019.29} {\path{doi:10.4230/LIPIcs.CPM.2019.29}}.

\bibitem[Vin68]{Vin68}
Taras~Klymovych Vintsyuk.
\newblock Speech discrimination by dynamic programming.
\newblock {\em Cybernetics}, 4(1):52--57, 1968.
\newblock \href {https://doi.org/10.1007/BF01074755} {\path{doi:10.1007/BF01074755}}.

\bibitem[Wei73]{Weiner1973}
Peter Weiner.
\newblock Linear pattern matching algorithms.
\newblock In {\em 14th Annual Symposium on Switching and Automata Theory, {SWAT} 1973}, pages 1--11. IEEE, 1973.
\newblock \href {https://doi.org/10.1109/SWAT.1973.13} {\path{doi:10.1109/SWAT.1973.13}}.

\bibitem[WF74]{WF74}
Robert~A. Wagner and Michael~J. Fischer.
\newblock The string-to-string correction problem.
\newblock {\em Journal of the {ACM}}, 21(1):168--173, 1974.
\newblock \href {https://doi.org/10.1145/321796.321811} {\path{doi:10.1145/321796.321811}}.

\bibitem[WY20]{WY20}
Qisheng Wang and Mingsheng Ying.
\newblock Quantum algorithm for lexicographically minimal string rotation, 2020.
\newblock \href {https://arxiv.org/abs/2012.09376} {\path{arXiv:2012.09376}}.

\bibitem[ZL77]{DBLP:journals/tit/ZivL77}
Jacob Ziv and Abraham Lempel.
\newblock A universal algorithm for sequential data compression.
\newblock {\em IEEE Transactions on Information Theory}, 23(3):337--343, 1977.
\newblock \href {https://doi.org/10.1109/TIT.1977.1055714} {\path{doi:10.1109/TIT.1977.1055714}}.

\end{thebibliography}
\end{document}